\documentclass[journal,twoside]{IEEEtran}

\usepackage[pdftex]{graphicx}
\graphicspath{{figures/}}
\DeclareGraphicsExtensions{.pdf,.jpeg,.png}
\usepackage[caption=false,font=footnotesize]{subfig}

\usepackage{amsmath}
\usepackage{amsfonts}
\usepackage{bm} %
\usepackage{amsthm} %
\usepackage{amssymb} %

\usepackage{array}
\usepackage{mdwmath}
\usepackage{mdwtab}
\usepackage{eqparbox}

\usepackage[caption=false,font=footnotesize]{subfig}

\usepackage{booktabs}
\usepackage{url}
\usepackage{hyperref}
\usepackage{cleveref}
\usepackage{cite}
\usepackage{cleveref}

\usepackage{tikz}
\usetikzlibrary{calc}
\usetikzlibrary{fit}
\usetikzlibrary{patterns}
\usetikzlibrary{shapes}
\usetikzlibrary{positioning}
\usetikzlibrary{arrows.meta}

\usepackage{pgfplots}

\usepackage{xcolor} %
\definecolor{nodeAny}{rgb}{0.5, 0.5, 0.5}
\definecolor{nodeA}{rgb}{0.0, 0.0, 0.0}
\definecolor{nodeB}{rgb}{0.0, 0.0, 0.0}
\definecolor{nodeC}{rgb}{0.0, 0.0, 0.0}
\definecolor{BLUE}{RGB}{0,0,255}
\definecolor{myParula01Blue}{RGB}{0,114,189}
\definecolor{myParula02Orange}{RGB}{217,83,25}
\definecolor{myParula03Yellow}{RGB}{237,177,32}
\definecolor{myParula04Purple}{RGB}{126,47,142}
\definecolor{myParula05Green}{RGB}{119,172,48}
\definecolor{myParula06LightBlue}{RGB}{77,190,238}
\definecolor{myParula07Red}{RGB}{162,20,47}
\makeatletter
\def\mathcolor#1#{\@mathcolor{#1}}
\def\@mathcolor#1#2#3{%
  \protect\leavevmode
  \begingroup
    \color#1{#2}#3%
  \endgroup
}
\makeatother
\newcommand{\NA}[0]{\mathcolor{nodeA}{1}}
\newcommand{\NB}[0]{\mathcolor{nodeB}{2}}
\newcommand{\NC}[0]{\mathcolor{nodeC}{3}}
\newcommand{\DAP}[1]{#1}
\newcommand{\DO}[1]{#1}
\newcommand{\DSI}[1]{#1}

\renewcommand{\H}[0]{\ensuremath{\bm{H}}}
\newcommand{\G}[0]{\ensuremath{\bm{G}}}
\newcommand{\V}[0]{\ensuremath{\bm{V}}}

\newcommand{\T}[0]{\ensuremath{\bm{T}}}
\newcommand{\W}[0]{\ensuremath{\bm{W}}}
\newcommand{\X}[0]{\ensuremath{\bm{X}}}
\newcommand{\Y}[0]{\ensuremath{\bm{Y}}}
\newcommand{\Z}[0]{\ensuremath{\bm{Z}}}
\newcommand{\ZcorrS}[0]{\ensuremath{\Z_{\mathrm{corr}}}}
\newcommand{\Zcorr}[1]{\ensuremath{\Z_{\mathrm{corr},#1}}}
\newcommand{\I}[0]{\ensuremath{\bm{I}}}
\newcommand{\Q}[0]{\ensuremath{\bm{Q}}}

\newcommand{\w}[0]{\ensuremath{\bm{w}}}
\newcommand{\x}[0]{\ensuremath{\bm{x}}}
\newcommand{\y}[0]{\ensuremath{\bm{y}}}
\newcommand{\z}[0]{\ensuremath{\bm{z}}}
\newcommand{\zcorrS}[0]{\ensuremath{\z_{\mathrm{corr}}}}
\newcommand{\zcorr}[1]{\ensuremath{\z_{\mathrm{corr},#1}}}
\newcommand{\ntau}[0]{\ensuremath{\overline{\tau}}}
\newcommand{\herm}[0]{\ensuremath{\mathsf{H}}}
\newcommand{\transp}[0]{\ensuremath{\mathsf{T}}}
\newcommand{\pinv}[0]{\ensuremath{\dagger}}
\newcommand{\Prob}[0]{\ensuremath{\operatorname{Pr}}}
\newcommand{\Exp}[0]{\ensuremath{\mathbb E}}

\newcommand{\hEntr}[0]{\ensuremath{h}}
\newcommand{\MInf}[0]{\ensuremath{I}}
\newcommand{\Span}[0]{\ensuremath{\operatorname{span}}}
\newcommand{\Rank}[0]{\ensuremath{\operatorname{rank}}}
\newcommand{\Bern}[0]{\ensuremath{\operatorname{Bern}}}
\newcommand{\CN}[0]{\ensuremath{\mathcal{CN}}}

\newcommand{\IA}[0]{\ensuremath{{[\mathrm{IA}]}}}
\newcommand{\ZF}[0]{\ensuremath{{[\mathrm{ZF}]}}}
\newcommand{\Mq}[0]{\ensuremath{{[q]}}}
\newcommand{\Mqnot}[0]{\ensuremath{{[\overline{q}]}}}
\newcommand{\ZFIA}[0]{\ensuremath{\{ \mathrm{ZF}, \mathrm{IA} \}}}
\newcommand{\logp}[0]{\ensuremath{\log(\rho)}}
\newcommand{\ologp}[0]{\ensuremath{o\left[\logp\right]}}

\newcommand{\eqdef}[0]{\ensuremath{\triangleq}}

\newcommand{\stepA}[0]{\ensuremath{\mathrm{(a)}}}
\newcommand{\stepB}[0]{\ensuremath{\mathrm{(b)}}}
\newcommand{\stepC}[0]{\ensuremath{\mathrm{(c)}}}
\newcommand{\stepD}[0]{\ensuremath{\mathrm{(d)}}}
\newcommand{\stepE}[0]{\ensuremath{\mathrm{(e)}}}
\newcommand{\stepF}[0]{\ensuremath{\mathrm{(f)}}}

\newcommand{\eqA}[0]{\overset{\stepA}{=}}
\newcommand{\eqB}[0]{\overset{\stepB}{=}}
\newcommand{\eqC}[0]{\overset{\stepC}{=}}
\newcommand{\eqD}[0]{\overset{\stepD}{=}}
\newcommand{\eqE}[0]{\overset{\stepE}{=}}

\newcommand{\leqA}[0]{\overset{\stepA}{\leq}}
\newcommand{\leqB}[0]{\overset{\stepB}{\leq}}
\newcommand{\leqC}[0]{\overset{\stepC}{\leq}}

\newcommand{\leqE}[0]{\overset{\stepE}{\leq}}
\newcommand{\leqF}[0]{\overset{\stepF}{\leq}}

\newcommand{\frakWA}[0]{\ensuremath{\bm{\mathfrak W}}_\textnormal{A}}
\newcommand{\frakWB}[0]{\ensuremath{\bm{\mathfrak W}}_\textnormal{B}}
\newcommand{\frakX}[0]{\ensuremath{\bm{\mathfrak X}}}
\newcommand{\frakY}[0]{\ensuremath{\bm{\mathfrak Y}}}
\newcommand{\frakZ}[0]{\ensuremath{\bm{\mathfrak Z}}}
\newtheorem{lemma}{Lemma}
\newtheorem{definition}{Definition}

\newtheorem{corollary}{Corollary}
\newtheorem{theorem}{Theorem}

\renewcommand{\eqref}[1]{(\ref{#1})}

\usepackage{bibentry}
\nobibliography*

\tikzset{sfgAdder/.style={
    circle,draw,minimum size=1em,inner sep=0,
    append after command={
        [every edge/.append style={
            shorten >=\pgflinewidth/2,
            shorten <=\pgflinewidth/2,
        }]
        (\tikzlastnode.west) edge (\tikzlastnode.east)
        (\tikzlastnode.north) edge (\tikzlastnode.south)
    }
}}
\tikzset{sfgGain/.style={
    draw,regular polygon,regular polygon sides=3,shape border rotate=-90,inner sep=0.1em,
}}
\tikzset{sfgModulator/.style={
    circle,draw,minimum size=1em,inner sep=0,
    append after command={
        [every edge/.append style={
            shorten >=\pgflinewidth/2,
            shorten <=\pgflinewidth/2,
        }]
        (\tikzlastnode.north west) edge (\tikzlastnode.south east)
        (\tikzlastnode.north east) edge (\tikzlastnode.south west)
    }
}}

\begin{document}

\title{Degrees-of-Freedom of the MIMO Three-Way Channel with Node-Intermittency}

\author{Joachim~Neu,~\IEEEmembership{Student Member,~IEEE,}
        Anas~Chaaban,~\IEEEmembership{Senior~Member,~IEEE,}
        Aydin~Sezgin,~\IEEEmembership{Senior~Member,~IEEE,}
        and~Mohamed-Slim~Alouini,~\IEEEmembership{Fellow,~IEEE}%
\thanks{Manuscript received August 27, 2017; revised August 30, 2018; accepted May 6, 2019. Date of publication May 28, 2019; date of current version May 21, 2019.
This work was supported in part by the German Research Foundation (DFG), under grant SE 1697/16.
This work has been presented in part at the 2017 IEEE International Symposium on Information Theory \cite{ChaabanSezginAlouini_ISIT17_NodeIntermittency}.}%
\thanks{The work of J.~Neu,
at the time a student at Technische Universit\"at M\"unchen (TUM), 80333 M\"unchen, Germany,
was conducted during a research internship at King Abdullah University of Science and Technology (KAUST).
He is now a student at Stanford University, Stanford, CA 94305, USA.
Email: jneu@stanford.edu.}%
\thanks{A.~Chaaban is with the School of Engineering, University of British Columbia (UBC), Kelowna, BC Canada V1Y 1V7. Email: anas.chaaban@ubc.ca.}%
\thanks{A.~Sezgin is with the Institute of Digital Communication Systems, Ruhr-Universit\"at Bochum (RUB), 44780 Bochum, Germany. Email: aydin.sezgin@rub.de.}%
\thanks{M.-S.~Alouini is with the Division of Computer, Electrical, and Mathematical Sciences and Engineering (CEMSE), King Abdullah University of Science and Technology (KAUST), Thuwal, Saudi Arabia. Email: slim.alouini@kaust.edu.sa.}%
\thanks{Communicated by D.~Tuninetti, Associate Editor for Communications.}%
\thanks{Color versions of one or more of the figures in this paper are available online at http://ieeexplore.ieee.org.}%
\thanks{Digital Object Identifier 10.1109/TIT.2019.2919548}%
}

\markboth{IEEE Transactions on Information Theory (forthcoming)}%
{Neu \MakeLowercase{\textit{et al.}}: Degrees-of-Freedom of the MIMO Three-Way Channel with Node-Intermittency}

\makeatletter
\def\@IEEEpubidpullup{1.5\baselineskip}
\makeatother
\IEEEpubid{\begin{minipage}{\textwidth}
\centering
0018--9448~\copyright~2019 IEEE.
Personal use is permitted, but republication/redistribution requires IEEE permission.\\
See http://www.ieee.org/publications\_standards/publications/rights/index.html for more information.
\end{minipage}}

\IEEEtitleabstractindextext{%
\begin{abstract}
The characterization of fundamental performance bounds of many-to-many communication systems in which participating nodes are active in an intermittent way is one of the major challenges in communication theory. In order to address this issue, we
introduce the multiple-input multiple-output (MIMO) three-way channel (3WC) with an intermittent node and study its degrees-of-freedom (DoF) region and sum-DoF. We devise a non-adaptive encoding scheme based on zero-forcing, interference alignment and erasure coding, and show its DoF region (and thus sum-DoF) optimality for non-intermittent 3WCs and its sum-DoF optimality for (node-)intermittent 3WCs. However, we show by example that in general some DoF tuples in the intermittent 3WC can only be achieved by adaptive schemes, such as 
decode-forward relaying. This shows that non-adaptive encoding is sufficient for the non-intermittent 3WC and for the sum-DoF of intermittent 3WCs, but adaptive encoding is necessary for the DoF region of intermittent 3WCs.
Our work contributes to a better understanding of the fundamental limits of multi-way communication systems with intermittency and the impact of adaptation therein.
\end{abstract}

}

\maketitle

\IEEEdisplaynontitleabstractindextext
\IEEEpeerreviewmaketitle

\allowdisplaybreaks

\section{Introduction}
\label{sec:introduction}

\IEEEPARstart{I}{n} multi-way communication scenarios multiple nodes communicate with each other, each acting as a source, a destination, and possibly a relay at the same time. This mode of communication is especially important for future systems employing full-duplex and device-to-device (D2D) communication, \cite{SabharwalSchniterGuo,TehraniUysalYanikomeroglu,AsadiWangMancuso}. It is an important technique for efficient resource utilization that is expected to gain more prominence in future communication systems, especially with the rise of mesh networks, e.g., in industrial and vehicular networks.

Several multi-way communication scenarios have been studied in the literature, \cite{ChaabanSezgin_FnT}. It started with the two-way channel which was first studied by Shannon in \cite{Shannon_TWC} and subsequently in \cite{HekstraWillems,Han,Varshney,SongAlajajiLinder} for instance. Then, the scope was extended to two-way networks where two groups of nodes communicate with each other in a two-way fashion \cite{ChengDevroye,SuhWangTse,ChengDevroye_TwoWayIC}, and two-way relay networks where two nodes communicate with each other in a two-way fashion via a relay node \cite{NamChungLee_IT,SongDevroyeShaoNgo_JSAC,VazeHeath,WilsonNarayananPfisterSprintson}. This line of research has been further extended to multi-way networks where multiple nodes communicate with each other in a multi-way fashion, each node being a source and a destination at the same time \cite{MaierChaabanMathar,MaierChaabanMathar_ITW,ChaabanMaierSezginMathar,ElmahdyKeyiMohassebElBatt,Ong}, and to multi-way relay networks where the same communication as in multi-way networks takes place via a relay node \cite{GunduzYenerGoldsmithPoor_IT,OngKellettJohnson_IT,SezginAvestimehrKhajehnejadHassibi,ZewailMohassebNafieElGamal,MatthiesenZapponeJorswieck_TWC,NgoLarsson,OngLechnerJohnson,WangYuan,TianYener_MWRC_IT,LeeChunKWay,ChenweiWang,
ChengDevroyeLiu}.

\subsection{The Three-Way Channel with Intermittency}
\label{sec:introduction-3wc-with-intermittency}

\IEEEpubidadjcol

In this work, we focus on the multiple-input multiple-output (MIMO) three-way channel (3WC) which can be described as follows: Consider a system consisting of three terminals communicating with each other in a multi-way fashion, e.g., two D2D user terminals and a base station (BS) where the D2D users communicate with each other while exchanging signals with the BS (control signals or data). This 3WC is an extension of Shannon's two-way channel \cite{Shannon_TWC} and has been studied in \cite{ChaabanMaierSezginMathar,Ong,ElmahdyKeyiMohassebElBatt}.

Therein, it is assumed that the three nodes are connected all the time (Fig.~\ref{fig:intermittency-types-noint}).
This forms a \emph{non-intermittent} 3WC, which, although subsumed by the work in this paper, is not its main focus.
The main focus of this paper is the intermittent 3WC instead.
There are several reasons which motivate studying the intermittent 3WC channel, some of which are discussed next.

One motivation stems from practice, where connectivity can be intermittent.
For instance, a pair of D2D users, commonly chosen to be nearby users \cite{TehraniUysalYanikomeroglu,AsadiWangMancuso}, might be both disconnected from the BS due to shadowing (Fig.~\ref{fig:intermittency-types-nodeint}). Future generations of mobile communication systems (e.g., using mmWaves, where line-of-sight (LOS) propagation will become dominant) are expected to be heavily susceptible to this type of shadowing. In another scenario, the D2D users might operate in an underlay mode over a resource block used by the BS to communicate to cellular users (CUs). The BS connects to the D2D users whenever it does not communicate with a CU, and disconnects from the D2D users otherwise (see Fig.~\ref{fig:scenario-d2dcubs}). Note that D2D users are chosen so that they cause/receive negligible interference (relative to the desired signals) to/from other nodes using the same resource block~\cite{mach2015band,celik2017joint}. In both cases the two links from the D2D users to the BS are jointly intermittent, i.e., both are available or blocked at the same time. We call this \emph{node intermittency} and call the BS an \emph{intermittent node} (from the D2D users' perspective).
Another motivation stems from theory.
A channel with intermittency is a special `extreme' case of a channel with state.
Point-to-point channels, multiple-access channels (MAC), and broadcast channels (BC) with state have been studied in the past, see \cite[ch. 7]{ElgamalKim} and references therein.
The impact of intermittency on these channels can be studied based on these results.
There are two approaches to extend results on channels with intermittency.
The first consists of considering larger networks, such as the X channel \cite{YehWang}. The other is to consider networks with bidirectional links (feedback) such as
\cite{KarakusWangDiggavi,WangSuhDiggaviViswanath,VahidMaddahAliAvestimehr}.
We take the second approach in this paper by focusing on multi-way communication.
In such networks, feedback links enable relaying and hence provide additional paths for information flow that might be interrupted by intermittency.
This gives intermittency a more `global' impact since it affects all nodes indirectly by interrupting useful paths through nodes that act as relays.
Thus it is important to study networks with feedback and intermittency.
The smallest multi-way network (in terms of number of nodes) that one could study in this context is the two-way channel (TWC) \cite{Shannon_TWC}.
However, analyzing the impact of intermittency in the Gaussian TWC is straightforward as we shall see in Section~\ref{sec:prerequisites-channels-with-state}.
Therefore, the 3WC is the smallest viable example of a multi-way network with non-trivial behavior under intermittency.
Moreover, the 3WC subsumes other channels of interest such as
the two-way MAC and BC \cite{ChengDevroye},
and MAC and BC with cooperation
\cite{ieee01246003,ieee04106127}.

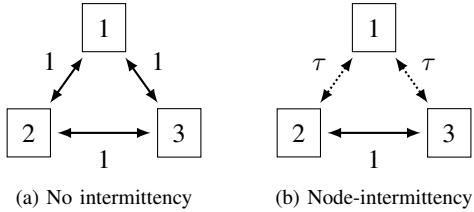
\begin{figure}[!t]
  \centering
  \subfloat[No intermittency]{%
    \hspace{5mm}%
    \begin{tikzpicture}%
      [msgarrow/.append style={
        thick,
        shorten >=\pgflinewidth*3,
        shorten <=\pgflinewidth*3,
      }]

      \node [rectangle,draw,inner sep=2mm] (node2) at (-1,0) {2};
      \node [rectangle,draw,inner sep=2mm] (node3) at (+1,0) {3};
      \node [rectangle,draw,inner sep=2mm] (node1) at (0,+1.4) {1};

      \draw[latex-latex,msgarrow] (node1) -- (node2) node[midway,above left] {1};
      \draw[latex-latex,msgarrow] (node1) -- (node3) node[midway,above right] {1};
      \draw[latex-latex,msgarrow] (node2) -- (node3) node[midway,below,yshift=-1mm] {1};

    \end{tikzpicture}%
    \hspace{5mm}%
    \label{fig:intermittency-types-noint}%
  }
  \subfloat[Node-intermittency]{%
    \hspace{5mm}%
    \begin{tikzpicture}%
      [msgarrow/.append style={
        thick,
        shorten >=\pgflinewidth*3,
        shorten <=\pgflinewidth*3,
      }]

      \node [rectangle,draw,inner sep=2mm] (node2) at (-1,0) {2};
      \node [rectangle,draw,inner sep=2mm] (node3) at (+1,0) {3};
      \node [rectangle,draw,inner sep=2mm] (node1) at (0,+1.4) {1};

      \draw[latex-latex,densely dotted,msgarrow] (node1) -- (node2) node[midway,above left] {$\tau$};
      \draw[latex-latex,densely dotted,msgarrow] (node1) -- (node3) node[midway,above right] {$\tau$};
      \draw[latex-latex,msgarrow] (node2) -- (node3) node[midway,below,yshift=-1mm] {1};

    \end{tikzpicture}%
    \hspace{5mm}%
    \label{fig:intermittency-types-nodeint}%
  }
  \caption{MIMO 3WC with no intermittency such that all nodes are always connected \protect\subref{fig:intermittency-types-noint} and node-intermittency where node $1$ is available only $\tau$ fraction of the time \protect\subref{fig:intermittency-types-nodeint}}
  \label{fig:intermittency-types}
\end{figure}

\tikzset{
  networkBS/.pic = {
    \coordinate (-tip) at (0,0);
    \coordinate (-north) at (0,0.26);
    \coordinate (-south) at (0,-0.5);
    \draw[pic actions] (0,0) -- ++(0.13,-0.5) -- ++(-0.26,0) --cycle;
    \draw[pic actions] (0,0) (-45:0.12) arc(-45:225:0.12);
    \draw[pic actions] (0,0) (-55:0.19) arc(-55:235:0.19);
    \draw[pic actions] (0,0) (-60:0.26) arc(-60:240:0.26);
  },
  networkUE/.pic = {
    \coordinate (-antenna) at (0,0.1);
    \coordinate (-north) at (0.05,0.25);
    \coordinate (-south) at (0.05,-0.50);
    \draw[pic actions] (-antenna) -- ++(0,-0.15) -- ++(0,+0.3);
    \draw[pic actions] (-antenna) ++(-0.15,0.15) -- (-antenna) -- ++(+0.15,0.15);
    \draw[pic actions] (-0.10,-0.05) rectangle (0.20,-0.50);
  },
}

\begin{figure}[!t]
  \centering
  \subfloat[BS serving D2D users]{%
    \begin{tikzpicture}[
        myconnection/.append style={
          latex-latex,
          thick,
        }
      ]

      \pic (bs) at (0,0) {networkBS};
      \pic (ue1) at (2.75,0) {networkUE};
      \pic (ue2) at (2.75,-1.75) {networkUE};

      \node[below=1mm of bs-south] {BS};
      \node[above=1mm of ue1-north] {D2D};
      \node[below=1mm of ue2-south] {D2D};

      \draw[myconnection,shorten <=4mm,shorten >=2mm] (bs-tip) -- (ue1-antenna);
      \draw[myconnection,shorten <=4mm,shorten >=3mm] (bs-tip) -- (ue2-antenna);
      \draw[myconnection,shorten <=7mm,shorten >=3mm] (ue1-antenna) -- (ue2-antenna);

    \end{tikzpicture}%
    \label{fig:scenario-d2dcubs-case1}%
  }
  \hspace{5mm}
  \subfloat[BS serving CU]{%
    \begin{tikzpicture}[
        myconnection/.append style={
          latex-latex,
          thick,
        }
      ]

      \pic (bs) at (0,0) {networkBS};
      \pic (ue1) at (2.75,0) {networkUE};
      \pic (ue2) at (2.75,-1.75) {networkUE};
      \pic (ue3) at (1,-1.75) {networkUE};

      \node[below=1mm of bs-south] {BS};
      \node[above=1mm of ue1-north] {D2D};
      \node[below=1mm of ue2-south] {D2D};
      \node[below=1mm of ue3-south] {CU};

      \draw[myconnection,shorten <=4mm,shorten >=2mm,dotted,thin] (bs-tip) -- (ue1-antenna);
      \draw[myconnection,shorten <=4mm,shorten >=3mm,dotted,thin] (bs-tip) -- (ue2-antenna);
      \draw[myconnection,shorten <=7mm,shorten >=3mm] (ue1-antenna) -- (ue2-antenna);

      \draw[myconnection,shorten <=4mm,shorten >=3mm] (bs-tip) -- (ue3-antenna);
      \draw[myconnection,shorten <=3mm,shorten >=3mm,dotted,thin] (ue1-antenna) -- (ue3-antenna);
      \draw[myconnection,shorten <=2mm,shorten >=3mm,dotted,thin] (ue2-antenna) -- (ue3-antenna);

      \node[rotate=-75,fill=white] at (1.7,-0.9) {\footnotesize negligible interference};

    \end{tikzpicture}%
    \label{fig:scenario-d2dcubs-case2}%
  }
  \caption{A D2D pair sharing the same resources with a cellular user (CU), where the BS communicates with the D2D pair part of the time \protect\subref{fig:scenario-d2dcubs-case1} and with the CU the rest of the time \protect\subref{fig:scenario-d2dcubs-case2}. The D2D pair is far enough from both the BS and the CU and thus causes/receives negligible interference (dotted).}
  \label{fig:scenario-d2dcubs}
\end{figure}
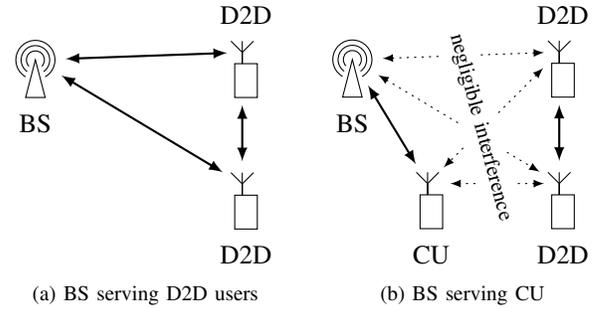

\subsection{Degrees-of-Freedom and Intermittency}

In many previously mentioned works, the focus is generally on the capacity of the studied network. Since finding the capacity is elusive in most cases, some works (and also this work) focus on the degrees-of-freedom (DoF) which provide a good capacity approximation at high signal-to-noise ratio (SNR) \cite{CadambeJafar_XNetworks}, thereby highlighting the interaction between the signals of the different nodes while diminishing the impact of noise. This is of interest since state-of-the-art wireless communication systems operate in a regime where they are essentially interference-limited rather than noise-limited.

The concept of intermittency as a form of channel impairment also fits well with the philosophy of DoFs: With increasing SNR, signal components `harden' in the sense that some allow (almost) noise- and interference-free communication at a rate that scales with SNR on a logarithmic scale, while others are hopelessly burried in uncancelable interference and are therefore useless. This effect is exactly what the DoF perspective captures as it essentially counts available Shannon-Hartley communication `units'. Intermittency is the channel impairment that goes well together with this `all or nothing' perspective, capturing the notion that some signal dimensions might become useless due to stochastic processes in the channel, such as shadowing of dominant LOS components.

\subsection{Scope of this Work}
\label{sec:introduction-scope}

Here, we study the impact of node-intermittency on the DoF region and sum-DoF, i.e., the capacity scaling versus SNR in a dB scale, of a full-duplex MIMO 3WC. In this MIMO network, each node generally has two independent messages, each intended for one of the remaining two nodes. We pay particular attention to the necessity (or the lack thereof) of \emph{adaptive encoding} for reaching DoF region/sum-DoF. Adaptive encoding enables cooperative communication schemes by allowing the transmit signal of a node to depend on its previously received signals, which can be interpreted as a form of feedback. In contrast, with \emph{non-adaptive encoding} transmit signals depend only on the messages to be sent, and cooperation (e.g., in the form of relaying) is excluded.
The issue of (non-)adaptive encoding has been studied for other networks earlier in \cite{Varshney,ChengDevroye,SongAlajajiLinder,ieee4839051} for instance.

We assume that nodes have strictly causal knowledge of the intermittency state of adjacent links.
This can be obtained by estimating the connectivity from the receive signals (e.g., signal strength). Note that in some cases the intermittency state can be known without the need for estimation, occasionally even ahead of time, e.g., in the scenario of D2D/CU/BS users (Fig.~\ref{fig:scenario-d2dcubs}), where the BS knows its scheduling of users in advance and can anticipate when it will not be able to receive the D2D user signals.

The reader might ask how our findings relate to previous results on channels with state and whether the intermittent 3WC is just a channel with state to which known techniques \cite[ch. 7]{ElgamalKim} can be applied. In fact, we make use of the known results for point-to-point channels with random state known to the receiver
\cite[eq. (7.2)]{ElgamalKim}
throughout this paper.
However, the \emph{multi-way} intermittent 3WC asks for a holistic analysis of its fundamental limits and cannot be exhaustively studied as a mere collection of independent \emph{one-way} one-hop multi-user channels with state, e.g., MACs with state
\cite[p. 175]{ElgamalKim}.

\subsection{Outline and Overview of Results}

After introducing the details of the system model in Section~\ref{sec:system-model} and highlighting the main results in Section~\ref{sec:main-results}, we examine the (node-)intermittent 3WC in Section~\ref{sec:node-intermittency}. We devise a non-adaptive encoding scheme based on zero-forcing (ZF), interference alignment (IA) and erasure coding (EC) and derive its achievable DoF region and sum-DoF. We present the genie-aided converse techniques used to derive DoF region and sum-DoF upper bounds, both under non-adaptive and adaptive encoding. We conclude that for the intermittent 3WC the presented non-adaptive scheme is sum-DoF optimal, so adaptation is not necessary to achieve sum-DoF. Then, we provide examples of adaptive relaying schemes that can achieve a DoF region point that no non-adaptive scheme can achieve. This shows that adaptive schemes can achieve strictly larger DoF regions, and therefore adaptation is required to achieve the DoF region of intermittent 3WCs. To complete the picture, we examine the non-intermittent 3WC as special case of intermittent 3WCs in Section~\ref{sec:no-intermittency}. We show that for the non-intermittent 3WC the presented scheme is DoF region optimal (and thus also sum-DoF optimal) and therefore adaptive encoding is not required. This reveals an interesting interplay between intermittency and adaptation. In Section~\ref{sec:conclusion} we provide conclusive remarks and directions for future research.

\subsection{Notation}

Throughout the paper, we use $x_i^n \eqdef (x_{i,1},\ldots,x_{i,n})$ for some index $i$, and $x_{i,\ell}^n \eqdef (x_{i,\ell},\ldots,x_{i,n})$. We use regular letters to denote scalar-valued quantities, boldface letters to denote vector- und matrix-valued quantities; lowercase letters for scalar and vector values (e.g., realizations of random variables), uppercase letters for matrix values and for random variables. The $N \times N$ identity matrix is denoted $\I_N$, the $N \times N$ all-zero matrix is $\bm{0}_N$. We write $\X \sim \CN(\bm{0},\Q)$ to indicate that $\X$ is a multivariate complex Gaussian random variable with zero mean and covariance matrix $\Q$, and $S \sim \Bern(\tau)$ to indicate that $S$ is Bernoulli distributed with $\Prob[S=1] = \tau$ and $\Prob[S=0] = 1-\tau =: \ntau$. We write $x^+$ to denote $\max\{0,x\}$ for some $x \in \mathbb R$, and $\H^\pinv$, $\H^\transp$, $\H^\herm$, and $\Span(\H)$ to denote the Moore-Penrose pseudo-inverse, the transpose, the Hermitian transpose, and the subspace spanned by the columns of the matrix $\H$. By $\log(x)$ we denote the logarithm of $x$ to base $2$, by $i \rightarrow j$ the communication from node $i$ to node $j$ one-way, and by $i \leftrightarrow j$ both $i \rightarrow j$ and $j \rightarrow i$. By $p_X(x)$ we denote the probability density function (PDF) of random variable $X$.

\section{Prerequisites and System Model}

In this section, we briefly recite the basics of channels with state, a fundamental building block of multi-way communication scenarios with intermittency.
We then present the system model of the intermittent 3WC.

\subsection{Channels with State and Intermittency}
\label{sec:prerequisites-channels-with-state}

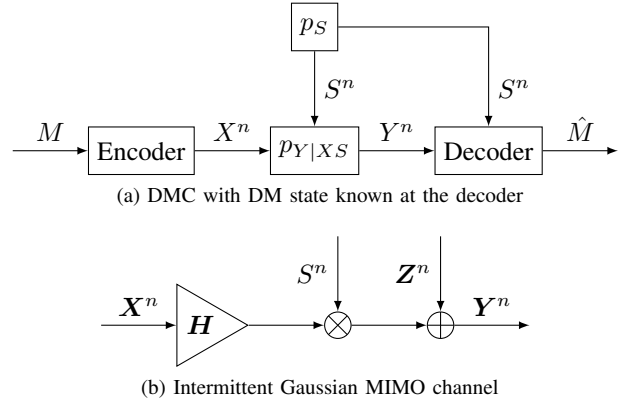
\begin{figure}[!t]
  \centering
  \subfloat[DMC with DM state known at the decoder]{%
    \begin{tikzpicture}

      \coordinate (start) at (0,0);
      \node [draw,right=of start,minimum height=18pt] (encoder) {Encoder};
      \node [draw,right=of encoder,minimum height=18pt] (channel) {$p_{Y|XS}$};
      \node [draw,right=of channel,minimum height=18pt] (decoder) {Decoder};
      \coordinate [right=of decoder] (stop);
      \node [draw,above=of channel,minimum height=18pt] (state) {$p_S$};

      \draw[-latex] (start) -- (encoder) node [midway,above] {$M$};
      \draw[-latex] (encoder) -- (channel) node [midway,above] {$X^n$};
      \draw[-latex] (channel) -- (decoder) node [midway,above] {$Y^n$};
      \draw[-latex] (decoder) -- (stop) node [midway,above] {$\hat{M}$};
      \draw[-latex] (state) -- (channel) node [midway,right] {$S^n$};
      \draw[-latex] (state) -| (decoder) node [near end,right,yshift=-4.5pt] {$S^n$};

    \end{tikzpicture}
    \label{fig:channels-with-state-dmc}
  }
  \\
  \subfloat[Intermittent Gaussian MIMO channel]{%
    \begin{tikzpicture}

      \coordinate (start) at (0,0);
      \node [sfgGain,right=of start] (gain) {$\H$};
      \node [sfgModulator,right=of gain] (modulator) {};
      \coordinate [above=of modulator] (state);
      \node [sfgAdder,right=of modulator] (adder) {};
      \coordinate [above=of adder] (noise);
      \coordinate [right=of adder] (stop);

      \draw[-latex] (start) -- (gain) node [midway,above] {$\X^n$};
      \draw[-latex] (gain) -- (modulator);
      \draw[-latex] (state) -- (modulator) node [midway,left] {$S^n$};
      \draw[-latex] (modulator) -- (adder);
      \draw[-latex] (noise) -- (adder) node [midway,left] {$\Z^n$};
      \draw[-latex] (adder) -- (stop) node [midway,above] {$\Y^n$};

    \end{tikzpicture}
    \label{fig:channels-with-state-gaussian}
  }
  \caption{The capacity of the DMC with state known at the decoder \protect\subref{fig:channels-with-state-dmc} is well-known, and used to derive the capacity of the intermittent Gaussian MIMO channel \protect\subref{fig:channels-with-state-gaussian}}
  \label{fig:channels-with-state}
\end{figure}

Recall the definition of a discrete memoryless channel (DMC) with discrete memoryless (DM) state known at the decoder (Fig.~\ref{fig:channels-with-state-dmc}):
Its state sequence $S^n$ is independent and identically distributed (i.i.d.) $S_\ell \sim p_S$ and independent of the input $X^n$, where $n$ is the number of channel uses. Since both output $Y^n$ and state sequence $S^n$ are known at the decoder, it can be viewed as a DMC with channel law
\[ p_{Y^n S^n | X^n}(y^n,s^n|x^n) = \prod_{\ell=1}^n p_{Y|XS}(y_\ell|x_\ell,s_\ell) p_S(s_\ell). \]
The capacity of this channel is \cite[eq. (7.2)]{ElgamalKim}
\begin{IEEEeqnarray}{C}
  C = \max_{p_X} \MInf( X; Y S ) = \max_{p_X} \MInf( X; Y \mid S).
  \nonumber
\end{IEEEeqnarray}

We use this result to derive the capacity of the intermittent Gaussian MIMO channel (Fig.~\ref{fig:channels-with-state-gaussian}) which will be useful in the sequel.
Here, the input $\X^n$ is a sequence of vectors $\x_\ell$ of length $M$ with average power constraint
\[ \sum_{\ell=1}^n \Exp[\| \X_\ell \|_2^2] \leq n P. \]
$S_\ell \sim \Bern(\tau)$ models the intermittency. The output $\Y^n$ of dimension $N$ is given as
\[ \y_\ell = s_\ell \H \x_\ell + \z_\ell, \quad{} \forall\ell, \]
where $\H \in \mathbb C^{N \times M}$ is the channel matrix, and $\Z^n$ is a noise sequence, i.i.d., with $\Z_\ell \sim \CN(\bm{0}, \sigma^2 \I_{N})$.
\begin{lemma}%
  \label{thm:prelim-capacity-intermittent-gaussian-mimo-channel}
  The capacity of the intermittent Gaussian MIMO point-to-point (P2P) channel is
  \begin{IEEEeqnarray}{rCl}
    C_{\mathrm{P2P}}
    &=& \MInf(\X ; \Y \mid S)   \nonumber\\
    &=& \tau \MInf(\X ; \Y \mid S = 1) + \ntau \MInf(\X ; \Y \mid S = 0)   \nonumber\\
    &=& \tau \log\det\left( \I_N + \frac{P}{\sigma^2} \H \H^\herm \right).
    \nonumber
  \end{IEEEeqnarray}
\end{lemma}

Now we turn to the intermittent Gaussian MIMO TWC, where two nodes communicate full-duplex over a bidirectional intermittent Gaussian MIMO channel and the intermittency state is
known strictly causally at the encoders
and decoders, i.e., only intermittency states $s^{\ell-1}$ can be used in encoding.
For simplicity, let the channel be reciprocal,
such that the channel matrices for the two directions are $\H$ and $\H^\herm$, respectively.
As the outputs $(\Y_1^n, S^n)$ and $(\Y_2^n, S^n)$ at nodes $1$ and $2$, with
\begin{IEEEeqnarray}{rClll}
  \y_{1,\ell} &=& s_\ell \H^\herm &\x_{2,\ell} + \z_{1,\ell},
  \quad
  \Z_{1,\ell} \sim \CN(\bm{0}, \sigma^2 \I_{M}),&
  \quad
   \forall\ell,
     \nonumber\\
  \y_{2,\ell} &=& s_\ell \H       &\x_{1,\ell} + \z_{2,\ell},
  \quad
  \Z_{2,\ell} \sim \CN(\bm{0}, \sigma^2 \I_{N}),&
  \quad
   \forall\ell,
     \nonumber
\end{IEEEeqnarray}
and $\Z_{1,\ell}$ and $\Z_{2,\ell}$ independent,
depend only on the inputs $\X_2^n$ and $\X_1^n$ of the respective other node, the channel law distributes as
\begin{IEEEeqnarray}{rCl}
    \IEEEeqnarraymulticol{3}{l}{
      p_{\Y_1^n \Y_2^n S^n | \X_1^n \X_2^n}(\y_1^n,\y_2^n,s^n|\x_1^n \x_2^n)
    }\nonumber\\
    \quad&=&
      \prod_{\ell=1}^n p_{\Y_1|\X_2 S}(\y_{1,\ell}|\x_{2,\ell},s_\ell) p_{\Y_2|\X_1 S}(\y_{2,\ell}|\x_{1,\ell},s_\ell) p_S(s_\ell).
      \nonumber
\end{IEEEeqnarray}
One can readily show that Shannon's outer bound
\cite[p.~447]{ElgamalKim}
holds despite the shared state $S$ of
$\X_{2} \leadsto \Y_{1}$ and $\X_{1} \leadsto \Y_{2}$.
The bounds are achieved by coding independently in both directions
for an intermittent Gaussian MIMO P2P channel.
Hence, the capacity region of this TWC is the rectangle
\[
  \mathcal C_{\mathrm{TWC}} = \{ (R_1, R_2) \in \mathbb R_+^2 \mid R_1 \leq C_{\mathrm{P2P}}, R_2 \leq C_{\mathrm{P2P}} \},
\]
and the sum-capacity of the TWC is twice the capacity of the P2P channel,
\[
  C_{\mathrm{TWC}} = 2 C_{\mathrm{P2P}}.
\]

We presented the capacity of the intermittent Gaussian MIMO P2P channel and TWC, which follow from established results for channels with state.
In particular, intermittency affects these channels in that it linearly scales the capacity of the non-intermittent channel with the fraction $\tau$ of time in which the channel is non-intermittent.
While the impact of intermittency is straightforward in the TWC, we shall see that this is not the case in larger multi-way communication channels.
The smallest (in terms of number of nodes) scenario larger than the TWC where this can be demonstrated is the 3WC.
Hence, the impact of intermittency on the 3WC is studied in this paper.
In the following, we introduce the system model of the intermittent 3WC.

\subsection{The MIMO Three-Way Channel with Node-Intermittency}
\label{sec:system-model}

Throughout this section we assume $i, j, k \in \{ 1, 2, 3 \}$ and mutually distinct, and $n$ is the number of channel accesses.
The MIMO 3WC with node-intermittency is comprised of three terminals $1$, $2$ and $3$ communicating with each other in full-duplex mode over a shared medium. Each node $i$ has two messages $w_{ij}$ and $w_{ik}$ ($\w_i \eqdef (w_{ij},w_{ik})$) to be delivered to the remaining nodes $j$ and $k$, and two messages $\hat{w}_{ji}$ and $\hat{w}_{ki}$ to be decoded from the received signals (Fig.~\ref{fig:system-model-messages}). Each message $w_{ij}$ is a realization of the random variable $W_{ij}$. All random variables $W_{ij}$ are independent.
In the \mbox{(node-)}intermittent 3WC one link is always available and two links are jointly intermittent with probability of being available $\tau$ (Fig.~\ref{fig:intermittency-types-nodeint} and \ref{fig:system-model-messages}).

Our objects under investigation (sum-DoF and DoF region of the intermittent 3WC -- both are introduced in the sequel) depend on how the numbers of antennas at each node relate to each other, i.e., whether the intermittent node has most, second most, or least antennas. To reduce the number of cases one has to analyze, yet study the system without loss of generality (w.l.o.g.), two symmetries come to mind that can be exploited: Either a) fix a certain node to be intermittent (e.g., node $1$ is intermittent), and investigate all possible relations among the numbers of antennas, or b) fix a relation between the numbers of antennas, and allow any one of the three nodes to be intermittent. The respective remaining cases follow by renaming. Preliminary work \cite{ChaabanSezginAlouini_ISIT17_NodeIntermittency} fixed node $1$ to be intermittent and further assumed a relation on the numbers of antennas, which is not w.l.o.g. We fix node $1$ to be intermittent, but allow for any combination of numbers of antennas, which is approach a) and w.l.o.g.

Node $i$ is equipped with $M_i$ antennas that are used for reception and transmission simultaneously. We assume channel accesses of the nodes are synchronized and time-discretized: At time instance $\ell$ the transmit signal $\x_{i,\ell} \in \mathbb C^{M_i}$ is a realization of a random vector $\X_{i,\ell}$ satisfying the power constraint
\begin{IEEEeqnarray}{rCl}
  \sum_{\ell=1}^n \Exp[\| \X_{i,\ell} \|_2^2] \leq n P.
  \nonumber
\end{IEEEeqnarray}
The receive signals $\y_{i,\ell} \in \mathbb C^{M_i}$ are
\begin{IEEEeqnarray}{rCCCCCCCl}
  \y_{1,\ell} &=&  && s_{\ell} \H_{21} \x_{2,\ell} &+& s_{\ell} \H_{31} \x_{3,\ell} &+& \z_{1,\ell}, \nonumber\\
  \y_{2,\ell} &=& s_{\ell} \H_{12} \x_{1,\ell} &&  &+& \H_{32} \x_{3,\ell} &+& \z_{2,\ell}, \nonumber\\
  \y_{3,\ell} &=& s_{\ell} \H_{13} \x_{1,\ell} &+& \H_{23} \x_{2,\ell} && &+& \z_{3,\ell}. \nonumber
\end{IEEEeqnarray}
$\H_{ij} \in \mathbb C^{M_j \times M_i}$ represents the channel matrix from node $i$ to node $j$ and is constant over time and known to all nodes in advance. The elements of these matrices are drawn independently from the same continuous distribution, such that $\Rank(\H_{ij}) = \min(M_j, M_i)$ almost surely. $\z_{i,\ell} \in \mathbb C^{M_i}$ is a realization of the noise process $\Z_{i,\ell} \sim \CN(\bm{0}, \sigma^2 \I_{M_i})$ independent and identically distributed (i.i.d.) with respect to (w.r.t.) $\ell$, and $s_{\ell} \in \{0,1\}$ is a realization of the intermittency state process $S_{\ell}$. $S_{\ell} \sim \Bern(\tau)$ is assumed to be i.i.d. w.r.t. $\ell$.
The state sequence $s^n$ is known strictly causally at all nodes
(i.e., in time instance $\ell$ all nodes know $s^{\ell-1}$), because every node can correctly estimate $s_\ell$ from its $y_{i,\ell}$ with very high probability, e.g., based on the received signal strength.
Due to the physical properties of the channel and the distinct receivers, all random variables $\H_{ij}$, $\Z_i^n$ and $S^n$ are assumed to be mutually independent.
We remark that our analysis continues to hold if the noise at different receivers is correlated.
We denote $\rho \eqdef \frac{P}{\sigma^2}$ and call it SNR (signal-to-noise-power-ratio) throughout the paper.

\begin{figure}[!t]
  \centering
  \begin{tikzpicture}[
      mynode/.append style={
        rectangle,
        draw,
        inner sep=2mm,
      },
      mymsg/.append style={
      },
      mychannel/.append style={
        thick,
        shorten >=\pgflinewidth*3,
        shorten <=\pgflinewidth*3,
      }
    ]

    \node [mynode] (n2) at (-2,0) {Node 2};
    \node [mynode] (n3) at (+2,0) {Node 3};
    \node [mynode] (n1) at (0,2) {Node 1};

    \draw[latex-latex,mychannel,densely dotted] (n1) -- (n2) node[midway] (e12) {};
    \draw[latex-latex,mychannel,densely dotted] (n1) -- (n3) node[midway] (e13) {};
    \draw[latex-latex,mychannel] (n2) -- (n3);

    \draw[{Rays[n=4]}-{Rays[n=4]},densely dotted,shorten >=-2mm,shorten <=-2mm] (e12) -- (e13) node[midway,fill=white] {$S$};

    \node [mymsg,above left=5mm and -6mm of n1] (n1in) {$W_{12},W_{13}$};
    \node [mymsg,above right=5mm and -6mm of n1] (n1out) {$\hat{W}_{21},\hat{W}_{31}$};
    \draw[-latex] (n1in) -- (n1);
    \draw[-latex] (n1) -- (n1out);

    \node [mymsg,below left=5mm and -6mm of n2] (n2in) {$W_{21},W_{23}$};
    \node [mymsg,below right=5mm and -6mm of n2] (n2out) {$\hat{W}_{12},\hat{W}_{32}$};
    \draw[-latex] (n2in) -- (n2);
    \draw[-latex] (n2) -- (n2out);

    \node [mymsg,below left=5mm and -6mm of n3] (n3in) {$W_{31},W_{32}$};
    \node [mymsg,below right=5mm and -6mm of n3] (n3out) {$\hat{W}_{13},\hat{W}_{23}$};
    \draw[-latex] (n3in) -- (n3);
    \draw[-latex] (n3) -- (n3out);

  \end{tikzpicture}
  \caption{Every node $i$ in the MIMO 3WC with node-intermittency sends messages $w_{ij}$ and $w_{ik}$ to nodes $j$ and $k$, respectively (with $i, j, k \in \{1, 2, 3\}$ mutually distinct). W.l.o.g. node $1$ is assumed to be intermittent.}
  \label{fig:system-model-messages}
\end{figure}
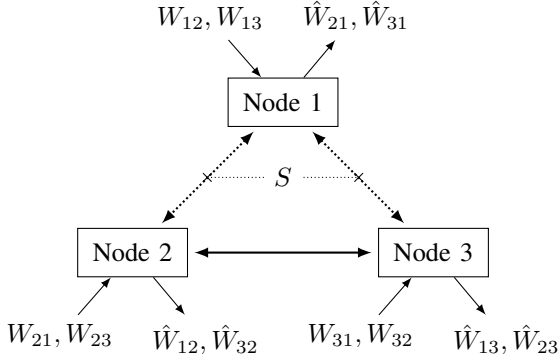

The messages $W_{ij}$ are each uniformly distributed over $\mathcal W_{ij} = \{1,\ldots, |\mathcal W_{ij}(\rho)|\}$. Using an encoding function $\mathcal E_{i,\ell}$ node $i$ constructs $\x_{i,\ell}$ either from $(w_{ij}, w_{ik})$ (non-adaptive encoding) or from $(w_{ij}, w_{ik}, \y_i^{\ell-1}, s^{\ell-1})$ (adaptive encoding).
After $n$ transmissions (where $n$ is the code length), node $i$ decodes its desired messages using a decoding function $\mathcal F_i$ to obtain $(\hat{w}_{ji}, \hat{w}_{ki}) = \mathcal F_i(s^n, \y_i^n, w_{ij}, w_{ik})$. Transmission is considered successful if all messages are recovered successfully ($w_{ij} = \hat{w}_{ij}$), otherwise an error is reported. The average over all messages of the error probability is denoted by $P_{\mathrm{e},n}$.

A rate tuple
\begin{IEEEeqnarray}{rCl}
  \bm{R}(\rho) &\eqdef& \Big( R_{12}(\rho),R_{13}(\rho),R_{21}(\rho), \nonumber\\
    && \qquad R_{23}(\rho),R_{31}(\rho),R_{32}(\rho) \Big) \in \mathbb R_+^6
    \nonumber
\end{IEEEeqnarray}
with $R_{ij}(\rho) = \frac{\log(|\mathcal W_{ij}(\rho)|)}{n}$ is said to be achievable if there exists a sequence of encoder-decoder pairs for increasing code length $n$, where $P_{\mathrm{e},n} \to 0$ as $n \to \infty$. The capacity region $\mathcal C(\rho)$ is the set of all achievable rate tuples.

The DoF region $\mathcal D$ is the set of achievable DoF tuples
\begin{IEEEeqnarray}{rCl}
  \bm{d} &\eqdef& (d_{12},d_{13},d_{21},d_{23},d_{31},d_{32}) \in \mathbb R_+^6
  \nonumber
\end{IEEEeqnarray}
defined as in \cite{JafarShamai_XChannel}%
, i.e.,
\begin{multline}
  \mathcal D \eqdef \Bigg\{
    (d_{12},...,d_{32}) \in \mathbb R_+^6 \;\bigg\vert\;
    \forall (\beta_{12},...,\beta_{32}) \in \mathbb R_+^6: \\
    \qquad
    \sum_{i,j} \beta_{ij} d_{ij} \leq \limsup_{\rho\to\infty} \sup_{\bm{R}(\rho)\in\mathcal C(\rho)} \sum_{i,j} \beta_{ij} \frac{R_{ij}(\rho)}{\logp}
    \Bigg\}.
    \nonumber
\end{multline}
Roughly speaking, if a rate tuple $\bm{R}(\rho)$ as a function of $\rho$ is achievable, i.e., $\bm{R}(\rho)\in\mathcal{C}(\rho)$ for all $\rho > 0$, then the DoF tuple $\bm{d}$ with $d_{ij} = \limsup_{\rho \to \infty}\frac{R_{ij}(\rho)}{\logp}$ is achievable. We denote DoF regions with $\mathcal D$ and use suitable subscripts when further restricting assumptions apply (e.g., under non-adaptive encoding). We define the corresponding sum-DoF as $d_{\mathrm{sum}} = \max_{\bm{d} \in \mathcal D} \sum_{i,j \in\{1,2,3\}, i \neq j} d_{ij}$. The DoF perspective only captures rate contributions that
are non-vanishing relative to $\logp$
as we let $\rho \to \infty$. It neglects vanishing rate portions $f(\rho)$ that grow sublinear in $\logp$ and are therefore $\ologp$, i.e., where $\lim_{\rho \to \infty} \frac{f(\rho)}{\logp} = 0$.

\section{Main Results}
\label{sec:main-results}

\begin{table*}[!t]
  \caption{Overview of Main Results}
  \label{tab:overview-main-results}
  \centering
  \begin{tabular}{llll}
    \toprule
    Channel & Criterion & Necessity of Adaptation & Coding Scheme \\
    \midrule
    Node-intermittent 3WC & Sum-DoF     & \begin{tabular}[t]{@{}l@{}}
                                            Non-adaptive encoding suffices \\
                                            (Theorem~\ref{thm:node-intermittency-sum-dof})
                                          \end{tabular}
                                        & \begin{tabular}[t]{@{}l@{}}
                                            ZF/IA/EC-based \\
                                            (Section~\ref{sec:node-intermittency-achievability}, Theorems~\ref{thm:node-intermittency-dof-region-ib} and \ref{thm:node-intermittency-sum-dof})
                                          \end{tabular} \\
                          & DoF region  & \begin{tabular}[t]{@{}l@{}}
                                            Adaptive encoding required \\
                                            (Theorem~\ref{thm:node-intermittency-dof-region})
                                          \end{tabular}
                                        & \begin{tabular}[t]{@{}l@{}}
                                            Counterexample based on decode-forward relaying \\
                                            (Section~\ref{sec:node-intermittency-dof-region-counterexample-achievability}) 
                                          \end{tabular} \\
    Non-intermittent 3WC  & Sum-DoF     & \begin{tabular}[t]{@{}l@{}}
                                            Non-adaptive encoding suffices \\
                                            (Theorem~\ref{thm:no-intermittency-dof-region}, Corollary~\ref{thm:no-intermittency-sum-dof})
                                          \end{tabular}
                                        & \begin{tabular}[t]{@{}l@{}}
                                            ZF/IA-based \\
                                            (Sections~\ref{sec:node-intermittency-achievability}, \ref{sec:no-intermittency-achievability}
                                            and
                                            \ref{sec:sum-dof-nonintermittent})
                                          \end{tabular} \\
                          & DoF region  & \begin{tabular}[t]{@{}l@{}}
                                            Non-adaptive encoding suffices \\
                                            (Theorem~\ref{thm:no-intermittency-dof-region})
                                          \end{tabular}
                                        & \begin{tabular}[t]{@{}l@{}}
                                            ZF/IA-based \\
                                            (Sections~\ref{sec:node-intermittency-achievability} and \ref{sec:no-intermittency-achievability})
                                          \end{tabular} \\
    \bottomrule
  \end{tabular}
\end{table*}

In this section, we summarize and discuss the main results of this paper,
listed in Table~\ref{tab:overview-main-results}.
We denote the DoF region of the intermittent 3WC under adaptive encoding by $\mathcal D^\mathrm{I}$, the DoF region of the intermittent 3WC under non-adaptive encoding by $\mathcal D_{\overline{\mathrm{A}}}^\mathrm{I}$, and the sum-DoF by $d_{\mathrm{sum}}^\mathrm{I}$. Obviously, $\mathcal D_{\overline{\mathrm{A}}}^\mathrm{I} \subseteq \mathcal D^\mathrm{I}$. We denote the DoF region of the non-intermittent 3WC by $\mathcal D^\mathrm{N}$ and the sum-DoF by $d_{\mathrm{sum}}^\mathrm{N}$. In Section~\ref{sec:node-intermittency-achievability} we devise a non-adaptive encoding scheme whose achievable DoF region (DoF region inner bound) we denote by $\mathcal D_{\mathrm{IB},\overline{\mathrm{A}}}^\mathrm{I}$. We start with this achievable DoF region given in the following theorem.

\begin{theorem}[DoF Region Inner Bound for Node-Intermittent 3WC]
  \label{thm:node-intermittency-dof-region-ib}
  All DoF tuples $\bm{d} \in \mathcal D_{\mathrm{IB},\overline{\mathrm{A}}}^\mathrm{I}$ satisfying the following set of inequalities are achievable in the node-intermittent 3WC using non-adaptive encoding:
  \begin{IEEEeqnarray}{rCl}
    \max\{ d_{12} + d_{13}, d_{21} + d_{31} \} &\leq& \tau M_{1} \nonumber\\
    \max\{ d_{21} + \tau d_{23}, d_{12} + \tau d_{32} \} &\leq& \tau M_{2} \nonumber\\
    \max\{ d_{31} + \tau d_{32}, d_{13} + \tau d_{23} \} &\leq& \tau M_{3} \nonumber\\
    \max\{ d_{12} + d_{13} + \tau d_{23}, d_{21} + d_{31} + \tau d_{32} \} &\leq& \tau \max\left\{ M_{1}, M_{3} \right\} \nonumber\\
    \max\{ d_{12} + d_{13} + \tau d_{32}, d_{21} + d_{31} + \tau d_{23} \} &\leq& \tau \max\left\{ M_{1}, M_{2} \right\} \nonumber\\
    \max\{ d_{12} + d_{31} + \tau d_{32}, d_{21} + d_{13} + \tau d_{23} \} &\leq& \tau \max\left\{ M_{2}, M_{3} \right\} \nonumber\\
    \min \{ d_{12}, d_{13}, d_{21}, d_{23}, d_{31}, d_{32} \} &\geq& 0 \nonumber
  \end{IEEEeqnarray}
  Therefore, $\mathcal D_{\mathrm{IB},\overline{\mathrm{A}}}^\mathrm{I}$ constitutes an inner bound on the DoF region of the node-intermittent 3WC, such that
  \[ \mathcal D_{\mathrm{IB},\overline{\mathrm{A}}}^\mathrm{I} \subseteq \mathcal D_{\overline{\mathrm{A}}}^\mathrm{I} \subseteq \mathcal D^\mathrm{I}. \]
\end{theorem}
\begin{proof}
  This theorem follows by the ZF/IA/EC-based construction in Section~\ref{sec:node-intermittency-achievability}.
\end{proof}

Subsequently we investigate whether this non-adaptive scheme is optimal. For the intermittent 3WC it turns out to be sum-DoF optimal, establishing the sum-DoF of the intermittent 3WC and the fact that adaptive encoding is not necessary to achieve it.

\begin{theorem}[Sum-DoF of Node-Intermittent 3WC]
  \label{thm:node-intermittency-sum-dof}
  A non-adaptive encoding scheme achieves the sum-DoF of the node-intermittent 3WC given by
  \begin{multline}
    d_{\mathrm{sum}}^\mathrm{I} = 2 \ntau \min\{ M_2, M_3 \} + 2 \tau \Big( M_1+M_2+M_3 \\ -\min\{M_1,M_2,M_3\}-\max\{M_1,M_2,M_3\} \Big). \nonumber
  \end{multline}
\end{theorem}
\begin{proof}
  The theorem follows from achievability and converse results developed in Sections~\ref{sec:nodeint-sum-dof-lb} and \ref{sec:nodeint-sum-dof-ub} (Lemmas~\ref{thm:node-intermittency-sum-dof-lb} and \ref{thm:node-intermittency-sum-dof-ub}).
\end{proof}

Theorem~\ref{thm:node-intermittency-sum-dof} is based on an instrumental genie-aided upper bound, which we prove to be tighter than cut-set bounds in Section~\ref{sec:necessity-of-genieaided-upper-bounds}.
Hence, cut-set bounds alone can not describe the DoF of this network comprehensively.

The sum-DoF optimality of non-adaptive schemes, i.e., of schemes that dispense with relaying, agrees with the following intuition: Assume relaying was used for (any part of) any message, say $w_{ik}$ was relayed via node $j$. Then this message occupies communication resources on two links, $i \to j$ and $j \to k$, introducing redundancy and thus waste of resources. Since the sum-DoF criterion allows to trade DoFs among messages, one could instead use the resources used by the one relayed message $w_{ik}$ to increase the DoFs of the two non-relayed messages $w_{ij}$ and $w_{jk}$. This is possible since every node has messages for the two other nodes in the 3WC. Such a reassignment could improve resource utilization and thus sum-DoF, rendering relaying dispensable.

From a DoF region perspective however it turns out that non-adaptive schemes cannot be optimal, and therefore the presented scheme is not DoF region optimal in the intermittent 3WC.

\begin{theorem}[Necessity of Adaptive Encoding for DoF Region of Node-Intermittent 3WC]
  \label{thm:node-intermittency-dof-region}
  Adaptive encoding is required to achieve the DoF region of the node-intermittent 3WC, i.e.,
  \[ \mathcal D^\mathrm{I} \setminus \mathcal D_{\overline{\mathrm{A}}}^\mathrm{I} \neq \emptyset. \]
\end{theorem}
\begin{proof}
  The theorem follows from an upper bound on $d_{31}$ under non-adaptive encoding, presented in Section~\ref{sec:node-intermittency-dof-region-counterexample-converse}, and counterexamples of adaptive schemes exceeding this bound, devised in Section~\ref{sec:node-intermittency-dof-region-counterexample-achievability}.
\end{proof}

Theorems~\ref{thm:node-intermittency-sum-dof} and \ref{thm:node-intermittency-dof-region} show that non-adaptive encoding is sufficient to achieve sum-DoF, but not sufficient to achieve the DoF region of the intermittent 3WC.
Unlike the sum-DoF criterion, the DoF region criterion does not allow to trade DoFs among messages. Instead, e.g., `extreme' DoF tuples need to be achievable as well, that make every effort (e.g., through relaying) to maximize a single message's DoF, usually at the cost of low sum-DoF.
The result is particularly interesting in light of the fact, that non-adaptive encoding is sufficient to achieve the DoF region of the non-intermittent 3WC, which the presented non-adaptive scheme does.

\begin{theorem}[DoF Region of Non-Intermittent 3WC]
  \label{thm:no-intermittency-dof-region}
  The DoF region of the non-intermittent 3WC $\mathcal D^\mathrm{N}$ (with $M_1 \geq M_2 \geq M_3$ w.l.o.g.) is given by
  \begin{IEEEeqnarray}{rCl}
    \max\{ d_{12} + d_{13} + d_{23}, d_{12} + d_{13} + d_{32} \} &\leq& M_1 \nonumber\\
    \max\{ d_{21} + d_{31} + d_{32}, d_{21} + d_{31} + d_{23} \} &\leq& M_1 \nonumber\\
    \max\{ d_{21} + d_{13} + d_{23}, d_{12} + d_{31} + d_{32} \} &\leq& M_2 \nonumber\\
    \max\{ d_{31} + d_{32}, d_{13} + d_{23} \} &\leq& M_3 \nonumber\\
    \min \{ d_{12}, d_{13}, d_{21}, d_{23}, d_{31}, d_{32} \} &\geq& 0 \nonumber
  \end{IEEEeqnarray}
  and achievable using a non-adaptive encoding scheme.
\end{theorem}
\begin{proof}
  The theorem follows from the ZF/IA-based achievability results in Section~\ref{sec:no-intermittency-achievability} (\eqref{eq:no-intermittency-dof-region-1} to \eqref{eq:no-intermittency-dof-region-9}) and the converse results in Section~\ref{sec:no-intermittency-converses} (\eqref{eq:no-intermittency-dof-region-ob-3} to \eqref{eq:no-intermittency-dof-region-ob-2}).
\end{proof}

This complements earlier work \cite{MaierChaabanMathar} that characterized the sum-DoF of the non-intermittent 3WC and showed its achievability using non-adaptive encoding.

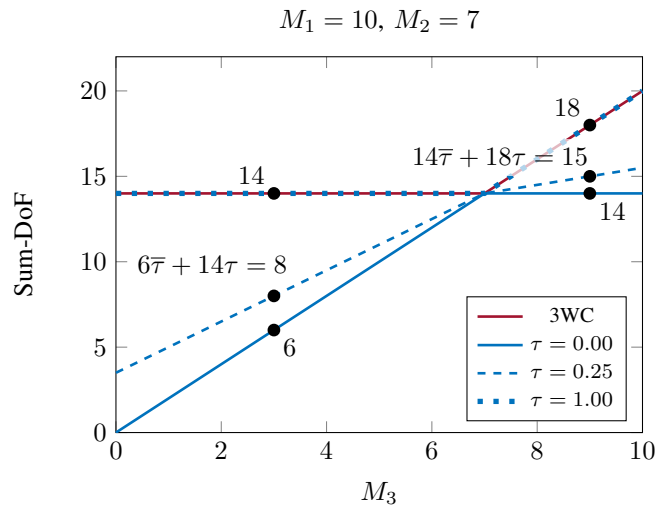
\begin{figure}[!t]
  \centering
  \begin{tikzpicture}
    \begin{axis}[
      scale only axis,
      width=7cm,
      height=5cm,
      title={$M_1 = 10$, $M_2 = 7$},
      xmin=0, xmax=10,
      ymin=0, ymax=22,
      xlabel={$M_3$},
      ylabel={Sum-DoF},
      legend pos=south east,
    ]
      \addplot[color=myParula07Red,line width=1pt,solid] coordinates {
        (0, 14)(1, 14)(2, 14)(3, 14)(4, 14)(5, 14)(6, 14)(7, 14)(8, 16)(9, 18)(10, 20)
      };
      \addplot[color=myParula01Blue,line width=1pt,solid] coordinates {
        (0, 0.0)(1, 2.0)(2, 4.0)(3, 6.0)(4, 8.0)(5, 10.0)(6, 12.0)(7, 14.0)(8, 14.0)(9, 14.0)(10, 14.0)
      };
      \addplot[color=myParula01Blue,line width=1pt,dashed] coordinates {
        (0, 3.5)(1, 5.0)(2, 6.5)(3, 8.0)(4, 9.5)(5, 11.0)(6, 12.5)(7, 14.0)(8, 14.5)(9, 15.0)(10, 15.5)
      };
      \addplot[color=myParula01Blue,line width=2pt,loosely dotted] coordinates {
        (0, 14.0)(1, 14.0)(2, 14.0)(3, 14.0)(4, 14.0)(5, 14.0)(6, 14.0)(7, 14.0)(8, 16.0)(9, 18.0)(10, 20.0)
      };
      \addplot[color=black,only marks,mark=*,thick] coordinates {
        (3, 14)(3, 6.0)(3, 8.0)
        (9, 18)(9, 14.0)(9, 15.0)
      };
      \node at (axis cs:3,14) [anchor=south east] {$14$};
      \node at (axis cs:3,6) [anchor=north west] {$6$};
      \node at (axis cs:3,8) [anchor=south east,xshift=3mm,yshift=1.5mm] {$6 \ntau + 14 \tau = 8$};
      \node at (axis cs:9,18) [anchor=south east] {$18$};
      \node at (axis cs:9,14) [anchor=north west] {$14$};
      \node at (axis cs:9,15) [anchor=south east,xshift=1mm,fill=white,text opacity=1,fill opacity=0.7] {$14 \ntau + 18 \tau = 15$};
      \legend{\footnotesize 3WC, \footnotesize $\tau=0.00$, \footnotesize $\tau=0.25$, \footnotesize $\tau=1.00$}
    \end{axis}
  \end{tikzpicture}
  \caption{Sum-DoF of 3WC (red) and intermittent 3WC (blue) with varying $\tau$ (black marks show the case $\tau=0.25$ as convex combination of the extreme cases $\tau=0$ and $\tau=1$)}
  \label{fig:conclusion-sum-dof-vary-tau}
\end{figure}

\begin{figure}[!t]
  \centering
  \begin{tikzpicture}
    \begin{axis}[
      scale only axis,
      width=7cm,
      height=5cm,
      title={$M_1 = 10$, $\tau = 0.7$},
      xmin=0, xmax=10,
      ymin=0, ymax=22,
      xlabel={$M_3$},
      ylabel={Sum-DoF},
      legend pos=south east,
      legend columns=3,
      transpose legend,
      legend style={/tikz/every even column/.append style={column sep=0.5em},cells={align=left}}
    ]
      \addplot[color=myParula07Red,line width=1pt,solid] coordinates {
        (0, 4)(1, 4)(2, 4)(3, 6)(4, 8)(5, 10)(6, 12)(7, 14)(8, 16)(9, 18)(10, 20)
      };
      \addplot[color=myParula07Red,line width=1pt,dashed] coordinates {
        (0, 10)(1, 10)(2, 10)(3, 10)(4, 10)(5, 10)(6, 12)(7, 14)(8, 16)(9, 18)(10, 20)
      };
      \addplot[color=myParula07Red,line width=1pt,dotted] coordinates {
        (0, 18)(1, 18)(2, 18)(3, 18)(4, 18)(5, 18)(6, 18)(7, 18)(8, 18)(9, 18)(10, 20)
      };
      \addplot[color=myParula01Blue,line width=1pt,solid] coordinates {
        (0, 2.8)(1, 3.4)(2, 4.0)(3, 5.3999999999999995)(4, 6.8)(5, 8.2)(6, 9.599999999999998)(7, 11.0)(8, 12.399999999999999)(9, 13.8)(10, 15.2)
      };
      \addplot[color=myParula01Blue,line width=1pt,dashed] coordinates {
        (0, 7.0)(1, 7.6)(2, 8.2)(3, 8.8)(4, 9.4)(5, 10.0)(6, 11.399999999999999)(7, 12.799999999999999)(8, 14.2)(9, 15.6)(10, 17.0)
      };
      \addplot[color=myParula01Blue,line width=1pt,dotted] coordinates {
        (0, 12.6)(1, 13.2)(2, 13.8)(3, 14.4)(4, 15.0)(5, 15.6)(6, 16.2)(7, 16.8)(8, 17.4)(9, 18.0)(10, 19.4)
      };
      \legend{\footnotesize $M_2=2$, \footnotesize $M_2=5$, \footnotesize $M_2=9$, \footnotesize $M_2=2$, \footnotesize $M_2=5$, \footnotesize $M_2=9$}
    \end{axis}
  \end{tikzpicture}
  \caption{Sum-DoF of 3WC (red) and intermittent 3WC (blue) for varying $M_2$}
  \label{fig:conclusion-sum-dof-vary-m2}
\end{figure}
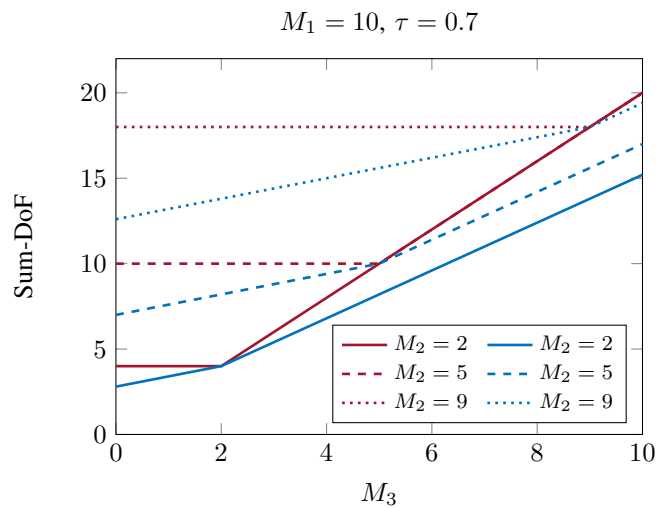

Our results show that intermittency does not decrease sum-DoF as long as node $1$ has the smallest number of antennas. Otherwise, intermittency does affect sum-DoF, which is affine-linearly increasing in $\tau$.
Fig.~\ref{fig:conclusion-sum-dof-vary-tau} and \ref{fig:conclusion-sum-dof-vary-m2} plot sum-DoF of non-intermittent and node-intermittent 3WC against $M_3$, for different values of $\tau$ and $M_2$, respectively, in a scenario where node $1$ has the largest number of antennas. The resulting graphs are piecewise linear with a change in slope at $M_3 = M_2$. For $M_3 < M_2$ the sum-DoF of the 3WC is a constant depending only on $M_2$, for $M_3 > M_2$ it is linear in $M_3$. The slope of the sum-DoF of the intermittent 3WC is proportional to $\tau$ for $M_3 < M_2$ and proportional to $\ntau$ for $M_3 > M_2$.

The sum-DoF of the intermittent 3WC with $0 < \tau < 1$ is a convex time-sharing combination of the two extreme cases $\tau=0$ and $\tau=1$ (Fig.~\ref{fig:conclusion-sum-dof-vary-tau}). Note, that such a time-sharing combination is the best any non-adaptive coding scheme can achieve, when each node knows the intermittency state ahead of time and codes optimally for the respective state in each time instance $\ell$. The achievability scheme presented in this work does not use intermittency state information at the encoder, yet achieves the same sum-DoF, averaging out intermittency state through erasure coding. This shows that intermittency state information is dispensable at the transmitter for the non-adaptive scheme in this case.
Furthermore, note that the sum-DoF of the 3WC is larger than the sum-DoF of the intermittent 3WC (Fig.~\ref{fig:conclusion-sum-dof-vary-m2}).

Intermittency impacts the DoF region of the 3WC in a pivotal way:
While non-adaptive encoding schemes, if suitably designed, are optimal for the non-intermittent 3WC,
adaptive encoding techniques are indispensable for the intermittent 3WC.
For sum-DoF this is not the case; instead, non-adaptive encoding suffices for both intermittent and non-intermittent 3WC.
This reinforces that changes in fundamental qualitative channel properties might not be recognizable from the sum-DoF perspective, corroborating the necessity to study the full DoF region of multi-way communication scenarios.

Considering multi-way communication networks, recall that adaptive encoding opens new paths for information flow which are otherwise unavailable to non-adaptive schemes.
As intermittency impairs parts of the network, the ability of adaptive schemes to exploit path diversity and steer clear of the impairment becomes crucial to achieve `extreme' DoF tuples (e.g., tuples where all resources are used to maximize a certain DoF).
It is in line with this informal reasoning
to find that adaptive encoding is required to achieve the DoF region of a multi-way communication network with intermittency.

In the sequel we provide detailed derivations and proofs of our main results highlighted in this section.

\section{Node Intermittency}
\label{sec:node-intermittency}

In this section we first introduce a non-adaptive ZF/IA/EC-based scheme and derive its achievable sum-DoF and DoF region. We then show, using enhanced genie-aided bounds for both adaptive and non-adaptive encoding, that this scheme is sum-DoF optimal, but not DoF region optimal. Furthermore, no non-adaptive scheme can be DoF region optimal, since tighter outer bounds hold for non-adaptive schemes, that however can be exceeded by adaptive schemes, as we show by the example
of decode-forward relaying. This establishes that adaptive encoding is necessary from a DoF region perspective, but non-adaptive encoding is sufficient to achieve sum-DoF.

\subsection{A Non-Adaptive Scheme Based on ZF, IA and EC}
\label{sec:node-intermittency-achievability}

In this subsection we present a non-adaptive transmission scheme based on ZF, IA and EC that provides an inner bound on sum-DoF and DoF region of 3WCs.
Previous works \cite{MaierChaabanMathar,MaierChaabanMathar_ITW,ElmahdyKeyiMohassebElBatt,ElmahdyKeyiMohassebElBatt_journal} developed similar ZF/IA-based schemes only to the limited extent necessary to analyze the sum-DoF of various 3WCs.
We take EC as additional technique to mitigate intermittency and develop the resulting ZF/IA/EC-based scheme in full generality, i.e., for arbitrary numbers of antennas and flexible allocation of DoFs to the available transmission techniques ZF, IA and EC.
The resulting DoF region is optimal for the non-intermittent 3WC.
Furthermore, the DoF region/sum-DoF constitutes inner/lower bounds for 3WCs with arbitrary intermittency models beyond node-intermittency, some of which are mentioned in Section~\ref{sec:conclusion}.

Throughout this section we continue to assume $i, j, k \in \{ 1, 2, 3 \}$ and mutually distinct. For notational simplicity (capturing the symmetries inherent in the model thereby avoiding case distinctions), we introduce the following aliases that will be resubstituted towards the end of the section,
\begin{IEEEeqnarray}{rCl}
  \label{eq:def-abbreviation-tau-s}
  s_{1,\ell} \eqdef 1,
  \quad{}
  s_{2,\ell} \eqdef s_{3,\ell} \eqdef s_{\ell},
  \quad{}
  \tau_1 \eqdef 1,
  \quad{}
  \tau_2 \eqdef \tau_3 \eqdef \tau.
  \IEEEeqnarraynumspace
\end{IEEEeqnarray}
The idea behind these aliases is visualized in Fig.~\ref{fig:explain-helper-variables-coding-scheme}: Using these aliases we can, for instance, find a general expression for the receive signal $\y_{i,\ell}$ of node $i$ in terms of the transmit signals $\x_{j,\ell}$ and $\x_{k,\ell}$ of nodes $j$ and $k$, and the intermittency states $s_{k,\ell}$ and $s_{j,\ell}$ (marginally distributed as $\Bern(\tau_k)$ and $\Bern(\tau_j)$, respectively) of the links $j \leftrightarrow i$ and $k \leftrightarrow i$, respectively. The fact that $2 \leftrightarrow 3$ is not intermittent, and $1 \leftrightarrow 2$ and $1 \leftrightarrow 3$ are jointly intermittent, is accounted for by the appropriate resubstitution at due time.

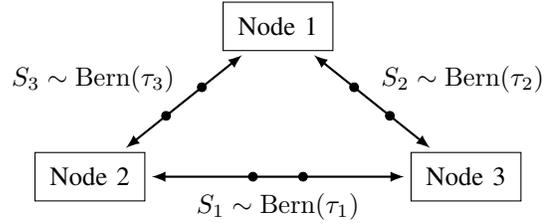
\begin{figure}[!t]
  \centering
  \begin{tikzpicture}[
      mynode/.append style={
        rectangle,
        draw,
        inner sep=2mm,
      },
      mymsg/.append style={
      },
      mychannel/.append style={
        thick,
        shorten >=\pgflinewidth*3,
        shorten <=\pgflinewidth*3,
      }
    ]

    \node [mynode] (n2) at (-2.5,0) {Node 2};
    \node [mynode] (n3) at (+2.5,0) {Node 3};
    \node [mynode] (n1) at (0,2) {Node 1};

    \draw[latex-latex,mychannel] (n1) -- (n2) node[pos=0.25] (e12a) {} node[pos=0.75] (e12b) {};
    \draw[latex-latex,mychannel] (n1) -- (n3) node[pos=0.25] (e13a) {} node[pos=0.75] (e13b) {};
    \draw[latex-latex,mychannel] (n2) -- (n3) node[pos=0.35] (e23a) {} node[pos=0.65] (e23b) {};

    \draw[{Circle[]}-{Circle[]}] (e12a) -- (e12b) node[midway,above left] {$S_3 \sim \Bern(\tau_3)$};
    \draw[{Circle[]}-{Circle[]}] (e13a) -- (e13b) node[midway,above right] {$S_2 \sim \Bern(\tau_2)$};
    \draw[{Circle[]}-{Circle[]}] (e23a) -- (e23b) node[midway,below,yshift=-1mm] {$S_1 \sim \Bern(\tau_1)$};

  \end{tikzpicture}
  \caption{Rationale behind the aliases $s_{1,\ell} \eqdef 1, s_{2,\ell} \eqdef s_{3,\ell} \eqdef s_{\ell}, \tau_1 \eqdef 1, \tau_2 \eqdef \tau_3 \eqdef \tau$ introduced in Section~\ref{sec:node-intermittency-achievability}: treating every link $j \leftrightarrow k$ as potentially intermittent with intermittency state variable $s_{i,\ell}$ (marginally distributed as $\Bern(\tau_i)$) allows to derive general expressions for, e.g., $\y_{i,\ell}$ based on $\x_{j,\ell}$, $\x_{k,\ell}$, $s_{j,\ell}$ and $s_{k,\ell}$, independent of whether $i=1$, $i=2$ or $i=3$}
  \label{fig:explain-helper-variables-coding-scheme}
\end{figure}

\subsubsection{Encoding}

Each node splits each message $w_{ij}$ into $w_{ij}^\ZF$ and $w_{ij}^\IA$ to be sent via zero-forcing and interference alignment, respectively. At node $i$, the four messages $w_{ij}^\Mq$ ($q \in \ZFIA$) are encoded into codewords $\x_{ij}^{\Mq n}$ with symbols $\x_{ij,\ell}^\Mq \in \mathbb C^{a_{ij}^\Mq}$ each, for some vector lengths $a_{ij}^\Mq \in \mathbb N_0$. The symbols of these codewords are chosen i.i.d. $\CN(\bm{0}, p_i \I_{a_{ij}^\Mq})$ respectively, where $p_i$ is the power. The power constraints on $\x_i^n$ stated in Section~\ref{sec:system-model} are satisfied by choosing
\begin{IEEEeqnarray}{rCl}
  \label{eq:tx-power-condition}
  p_i &=& \frac{P}{a_{ij}^\ZF + a_{ij}^\IA + a_{ik}^\ZF + a_{ik}^\IA}.
\end{IEEEeqnarray}

For encoding, the codes are designed to employ EC to be able to tolerate a certain number of symbol erasures (e.g., caused by intermittency), by not using all codeword symbols for net user data, but deliberately adding some redundancy. The rates and DoFs are thereby reduced accordingly. While for the non-intermittent 3WC this additional layer of EC is not required, it is made use of for intermittent 3WCs to cope with intermittency.

\subsubsection{Transmission}

At time $\ell$, node $i$ sends
\begin{IEEEeqnarray}{rCl}
  \x_{i,\ell} &=& \sum_{q \in \ZFIA} \left[ \V_{ij}^\Mq \x_{ij,\ell}^\Mq + \V_{ik}^\Mq \x_{ik,\ell}^\Mq \right]
  \nonumber
\end{IEEEeqnarray}
where $\V_{ij}^\Mq \in \mathbb C^{M_i \times a_{ij}^\Mq}$ are pre-coding matrices with unit-norm column vectors.
Zero-forcing is achieved by choosing the $\V_{ij}^\ZF$ such that
\begin{IEEEeqnarray}{rCl}
  \label{eq:zf-condition-precoder}
  \H_{ik} \V_{ij}^\ZF = \bm{0}.
\end{IEEEeqnarray}
Such matrices $\V_{ij}^\ZF$ exist if node $i$ has enough antennas to send $a_{ij}^\ZF$ streams to node $j$ without interfering with node $k$, i.e.
\begin{IEEEeqnarray}{rCl}
  \label{eq:zf-condition-dofs}
  a_{ij}^\ZF \leq (M_i - M_k)^+.
\end{IEEEeqnarray}
To avoid any overlap of the different transmit signal subspaces, we require furthermore that
\begin{IEEEeqnarray}{rCl}
  \label{eq:condition-dofs-tx-subspaces}
  a_{ij}^\ZF + a_{ij}^\IA + a_{ik}^\ZF + a_{ik}^\IA \leq M_i.
\end{IEEEeqnarray}

\subsubsection{Decoding}

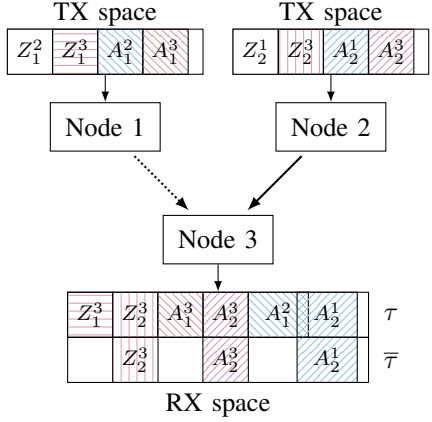
\begin{figure}[!t]
  \centering
  \begin{tikzpicture}[
      mynode/.append style={
        rectangle,
        draw,
        inner sep=2mm,
      },
      mymsg/.append style={
      },
      mychannel/.append style={
        thick,
        shorten >=\pgflinewidth*3,
        shorten <=\pgflinewidth*3,
      }
    ]%

    \node (n1) at (-1.5,1) [mynode] {Node 1};
    \node (n2) at (1.5,1) [mynode] {Node 2};
    \node (n3) at (0,-0.5) [mynode] {Node 3};

    \node (t1) at ($(n1)+(0,1)$) [rectangle,draw,minimum width=2.6cm,minimum height=.6cm] {};
    \node at ($(t1)+(0,.5)$) {TX space};

    \node (t2) at ($(n2)+(0,1)$) [rectangle,draw,minimum width=2.6cm,minimum height=.6cm] {};
    \node at ($(t2)+(0,.5)$) {TX space};

    \node (r3_t) at ($(n3)+(0,-1)$) [rectangle,draw,minimum width=4cm,minimum height=.6cm] {};
    \node at ($(r3_t)+(2.3,0)$) {$\tau$};
    \node (r3_nt) at ($(n3)+(0,-1.6)$) [rectangle,draw,minimum width=4cm,minimum height=.6cm] {};
    \node at ($(r3_nt)+(2.3,0)$) {$\ntau$};
    \node at ($(r3_nt)+(0,-0.6)$) {RX space};

    \draw[-latex] (t1) to (n1);
    \draw[-latex] (t2) to (n2);
    \draw[-latex] (n3) to (r3_t);

    \draw[-latex,mychannel,densely dotted] (n1) to (n3);
    \draw[-latex,mychannel] (n2) to (n3);

    \node at ($(t1)+(-1.0,0)$) [rectangle,draw,fill=white,minimum height=.6cm,minimum width=.6cm] {\footnotesize $Z_1^2$};
    \node at ($(t1)+(-0.4,0)$) [rectangle,draw,pattern=horizontal lines,pattern color=myParula07Red!40,minimum height=.6cm,minimum width=.6cm] {\footnotesize $Z_1^3$};
    \node at ($(t1)+(+0.2,0)$) [rectangle,draw,pattern=north west lines,pattern color=myParula01Blue!40,minimum height=.6cm,minimum width=.6cm] {\footnotesize $A_1^2$};
    \node at ($(t1)+(+0.8,0)$) [rectangle,draw,pattern=north west lines,pattern color=myParula07Red!40,minimum height=.6cm,minimum width=.6cm] {\footnotesize $A_1^3$};

    \node at ($(t2)+(-1.0,0)$) [rectangle,draw,fill=white,minimum height=.6cm,minimum width=.6cm] {\footnotesize $Z_2^1$};
    \node at ($(t2)+(-0.4,0)$) [rectangle,draw,pattern=vertical lines,pattern color=myParula07Red!40,minimum height=.6cm,minimum width=.6cm] {\footnotesize $Z_2^3$};
    \node at ($(t2)+(+0.2,0)$) [rectangle,draw,pattern=north east lines,pattern color=myParula01Blue!40,minimum height=.6cm,minimum width=.6cm] {\footnotesize $A_2^1$};
    \node at ($(t2)+(+0.8,0)$) [rectangle,draw,pattern=north east lines,pattern color=myParula07Red!40,minimum height=.6cm,minimum width=.6cm] {\footnotesize $A_2^3$};

    \node at ($(r3_t)+(-1.7,0)$) [rectangle,draw,pattern=horizontal lines,pattern color=myParula07Red!40,minimum height=.6cm,minimum width=.6cm] {\footnotesize $Z_1^3$};
    \node at ($(r3_t)+(-1.1,0)$) [rectangle,draw,pattern=vertical lines,pattern color=myParula07Red!40,minimum height=.6cm,minimum width=.6cm] {\footnotesize $Z_2^3$};
    \node at ($(r3_t)+(-0.5,0)$) [rectangle,draw,pattern=north west lines,pattern color=myParula07Red!40,minimum height=.6cm,minimum width=.6cm] {\footnotesize $A_1^3$};
    \node at ($(r3_t)+(+0.1,0)$) [rectangle,draw,pattern=north east lines,pattern color=myParula07Red!40,minimum height=.6cm,minimum width=.6cm] {\footnotesize $A_2^3$};
    \node at ($(r3_t)+(+0.8,0)$) [rectangle,draw,minimum height=.6cm,minimum width=.8cm,pattern=north west lines,pattern color=myParula01Blue!40] {\footnotesize $A_1^2$};
    \node at ($(r3_t)+(+1.45,0)$) [rectangle,draw,minimum height=.6cm,minimum width=.8cm,pattern=north east lines,pattern color=myParula01Blue!40] {\footnotesize $A_2^1$};

    \node at ($(r3_nt)+(-1.1,0)$) [rectangle,draw,pattern=vertical lines,pattern color=myParula07Red!40,minimum height=.6cm,minimum width=.6cm] {\footnotesize $Z_2^3$};
    \node at ($(r3_nt)+(+0.1,0)$) [rectangle,draw,pattern=north east lines,pattern color=myParula07Red!40,minimum height=.6cm,minimum width=.6cm] {\footnotesize $A_2^3$};
    \node at ($(r3_nt)+(+1.45,0)$) [rectangle,draw,minimum height=.6cm,minimum width=.8cm,pattern=north east lines,pattern color=myParula01Blue!40] {\footnotesize $A_2^1$};
  \end{tikzpicture}
  \caption{Visualization of transmit (TX) signal spaces at nodes $1$ and $2$ and receive (RX) signal space at node $3$ depending on intermittency state of $1 \leftrightarrow 3$ where the first and second row are received with probability $\tau$ and $\ntau$, respectively ($Z_i^j$ denotes $w_{ij}^\ZF$, $A_i^j$ denotes $w_{ij}^\IA$, desired signals in red, interfering signals in blue)}
  \label{fig:achievability-example-signal-spaces}
\end{figure}

Node $i$ receives (Fig.~\ref{fig:achievability-example-signal-spaces})
\begin{IEEEeqnarray}{rCl}
  \y_{i,\ell} &=& s_{k,\ell} \H_{ji} \x_{j,\ell} + s_{j,\ell} \H_{ki} \x_{k,\ell} + \z_{i,\ell} \nonumber\\
  &=& \sum_{q \in \ZFIA} \left[ s_{k,\ell} \H_{ji} \V_{ji}^\Mq \x_{ji,\ell}^\Mq + s_{j,\ell} \H_{ki} \V_{ki}^\Mq \x_{ki,\ell}^\Mq \right] \nonumber\\
  && \quad{}+{} \G_i \begin{bmatrix}
    \x_{jk,\ell}^\IA \\
    \x_{kj,\ell}^\IA
  \end{bmatrix} + \z_{i,\ell}
  \nonumber
\end{IEEEeqnarray}
with $\G_i = \left[ \begin{smallmatrix} s_{k,\ell} \H_{ji} \V_{jk}^\IA & s_{j,\ell} \H_{ki} \V_{kj}^\IA \end{smallmatrix} \right]$. Note that the terms $\H_{ji} \V_{jk}^\ZF \x_{jk,\ell}^\ZF$ and $\H_{ki} \V_{kj}^{\ZF} \x_{kj,\ell}^\ZF$ vanished due to the zero-forcing condition \eqref{eq:zf-condition-precoder}. The first summand represents the four desired signals from $j \rightarrow i$ and $k \rightarrow i$, the second summand represents occasional interference from $j \leftrightarrow k$, and the third summand is noise.

To decode a desired signal $\x_{ji}^{\Mq n}$, node $i$ zero-forces the remaining signals by multiplying $\y_i^n$ with a suitable post-coder $\T_{ji}^\Mq \in \mathbb C^{a_{ji}^\Mq \times M_i}$ with unit-norm row vectors satisfying:
\begin{IEEEeqnarray}{rCl}
  \label{eq:zf-condition-postcoder-1}
  \T_{ji}^\Mq \T_{ji}^{\Mq\herm} &=& \I_{a_{ji}^\Mq} \IEEEeqnarraynumspace \\
  \label{eq:zf-condition-postcoder-2}
  \T_{ji}^\Mq \begin{bmatrix} \H_{ji} \V_{ji}^\Mqnot & \H_{ki} \V_{ki}^\ZF & \H_{ki} \V_{ki}^\IA & \G_i \end{bmatrix} &=& \bm{0} \nonumber\\
  \IEEEeqnarraymulticol{3}{r}{
      \text{with } \overline{q} \in \ZFIA \setminus \{ q \}
  } \\
  \label{eq:zf-condition-postcoder-3}
  \Rank(\T_{ji}^\Mq \H_{ji} \V_{ji}^\Mq) &=& a_{ji}^\Mq
\end{IEEEeqnarray}
Here, \eqref{eq:zf-condition-postcoder-2} ensures zero-forcing of the remaining three messages and the interference and \eqref{eq:zf-condition-postcoder-3} ensures post-coding without loss of meaningful signal dimensions. The existence of such post-coders $\T_{ji}^\Mq$ is guaranteed as long as the columns of
\begin{IEEEeqnarray}{rCl}
  \begin{bmatrix} \H_{ji} \V_{ji}^\ZF & \H_{ji} \V_{ji}^\IA & \H_{ki} \V_{ki}^\ZF & \H_{ki} \V_{ki}^\IA & \G_i \end{bmatrix}
  \nonumber
\end{IEEEeqnarray}
are linearly independent. Let $\gamma_i$ be the dimension of $\Span(\H_{ji} \V_{jk}^\IA) \cap \Span(\H_{ki} \V_{kj}^\IA)$. Then, the dimension of $\Span(\G_i)$ is $a_{jk}^\IA + a_{kj}^\IA - \gamma_i$, and the above linear independence is possible almost surely if we choose
\begin{IEEEeqnarray}{rCl}
  \label{eq:condition-postcoder-linindep}
  a_{ji}^\ZF + a_{ji}^\IA + a_{ki}^\ZF + a_{ki}^\IA + (a_{jk}^\IA + a_{kj}^\IA - \gamma_i) \leq M_i.
  \nonumber\\*
\end{IEEEeqnarray}
To minimize the impact of interference, we choose the pre-coders $\V_{ij}^\Mq$ such that all $\gamma_i$ are maximized, i.e., we `maximally' align the interference subspaces at the receivers. The $\gamma_i$ cannot be chosen arbitrarily large, the dimension of the intersection of the interference subspaces ($\gamma_i$) is upper bounded by the dimensions of the interference subspaces ($a_{jk}^\IA$ and $a_{kj}^\IA$). Furthermore, $\gamma_i$ needs to be smaller than the dimension of $\Span(\H_{ji}) \cap \Span(\H_{ki})$, which is $(M_j + M_k - M_i)^+$ almost surely. Therefore, we require that
\begin{IEEEeqnarray}{rCl}
  \label{eq:condition-postcoder-gamma}
  \gamma_i \leq \min\{ a_{jk}^\IA, a_{kj}^\IA, (M_j + M_k - M_i)^+ \}.
\end{IEEEeqnarray}

After post-coding, node $i$ is left with the signals
\begin{IEEEeqnarray}{rCl}
  \y_{ji,\ell}^\Mq &=& s_{k,\ell} \T_{ji}^\Mq \H_{ji} \V_{ji}^\Mq \x_{ji,\ell}^\Mq + \T_{ji}^\Mq \z_{i,\ell}.
  \nonumber
\end{IEEEeqnarray}
The resulting channel is an erasure-Gaussian MIMO channel with erasure probability $\ntau_k$, i.e., a channel whose output is
\begin{IEEEeqnarray}{rCl}
  \Y_{ji}^\Mq &=& S_{k} \T_{ji}^\Mq \H_{ji} \V_{ji}^\Mq \X_{ji}^\Mq + \T_{ji}^\Mq \Z_{i}
  \nonumber
\end{IEEEeqnarray}
with random variables $S_k \sim \Bern(\tau_k)$, $\X_{ji}^\Mq \sim \CN(\bm{0}, p_j \I_{a_{ji}^\Mq})$ and $\Z_i \sim \CN(\bm{0}, \sigma^2 \I_{M_i})$. We treat this as a channel with state known causally to the receiver,
cf.~Section~\ref{sec:prerequisites-channels-with-state}, Lemma~\ref{thm:prelim-capacity-intermittent-gaussian-mimo-channel},
such that for large $n$, the achievable rate over this channel is the mutual information between $\X_{ji}^\Mq$ and $(\Y_{ji}^\Mq, S_k)$:
\begin{IEEEeqnarray}{rCl}
  \IEEEeqnarraymulticol{3}{l}{
      \MInf(\X_{ji}^\Mq ; \Y_{ji}^\Mq S_k)
  } \nonumber\\
  \quad&=&
      \MInf(\X_{ji}^\Mq ; \Y_{ji}^\Mq \mid S_k) \nonumber\\
    &=& \tau_k \log\det\left( \I_{a_{ji}^\Mq} + \frac{p_j}{\sigma^2} \T_{ji}^\Mq \H_{ji} \V_{ji}^\Mq \V_{ji}^{\Mq\herm} \H_{ji}^\herm \T_{ji}^{\Mq\herm} \right)
    \nonumber
\end{IEEEeqnarray}
Due to \eqref{eq:tx-power-condition} and \eqref{eq:zf-condition-postcoder-3}, this leads to a DoF of $\tau_k a_{ji}^\Mq$, and the code that achieves it is an EC.

\subsubsection{Achievable DoF Region}

For $n$ large, $i \rightarrow j$ has a total of $\tau_k a_{ij}^\ZF + \tau_k a_{ij}^\IA$ DoFs per channel use,
\begin{IEEEeqnarray}{rCl}
  \label{eq:def-dofs-from-substreams}
  d_{ij} &=& \tau_k a_{ij}^\ZF + \tau_k a_{ij}^\IA.
\end{IEEEeqnarray}

Collecting \eqref{eq:zf-condition-dofs}, \eqref{eq:condition-dofs-tx-subspaces}, \eqref{eq:condition-postcoder-linindep} and \eqref{eq:condition-postcoder-gamma} as well as non-negativity of every $a_{ij}^\Mq$, we obtain:
\begin{IEEEeqnarray}{rCl}
  \label{eq:condition-dofs-tx-subspaces-REP}
  a_{ij}^\ZF + a_{ij}^\IA + a_{ik}^\ZF + a_{ik}^\IA &\leq& M_i   \nonumber\\*\\
  \label{eq:condition-postcoder-linindep-REP}
  a_{ji}^\ZF + a_{ji}^\IA + a_{ki}^\ZF + a_{ki}^\IA + (a_{jk}^\IA + a_{kj}^\IA - \gamma_i) &\leq& M_i   \nonumber\\*
\end{IEEEeqnarray}
\begin{IEEEeqnarray}{rCl}
  \label{eq:zf-condition-dofs-REP}
  a_{ij}^\ZF &\leq& (M_i - M_k)^+ \\
  \label{eq:condition-postcoder-gamma-REP}
  \gamma_i &\leq& \min\{ a_{jk}^\IA, a_{kj}^\IA, (M_j + M_k - M_i)^+ \} \\
  \label{eq:condition-nonneg-dimensions-REP}
  0 &\leq& a_{ij}^\Mq
\end{IEEEeqnarray}

Using \eqref{eq:def-dofs-from-substreams}, we obtain:
\begin{IEEEeqnarray}{rCl}
  \label{eq:dof-region-general-variables-1}
  \tau_j d_{ij} + \tau_k d_{ik} &\leq& \tau_j \tau_k M_i \\
  \label{eq:dof-region-general-variables-2}
  \tau_j \tau_i d_{ji} + \tau_k \tau_i d_{ki} + \tau_j \tau_k d_{jk} + \tau_k \tau_j d_{kj} \quad\nonumber\\ {}-{} \tau_i \tau_j \tau_k a_{jk}^\ZF - \tau_i \tau_k \tau_j a_{kj}^\ZF - \tau_i \tau_j \tau_k \gamma_i &\leq& \tau_i \tau_j \tau_k M_i
  \IEEEeqnarraynumspace
\end{IEEEeqnarray}
\begin{IEEEeqnarray}{rCl}
  \label{eq:dof-region-general-variables-3}
  \tau_i (\gamma_i + a_{jk}^\ZF) &\leq& d_{jk} \\
  \label{eq:dof-region-general-variables-4}
  \gamma_i &\leq& (M_j + M_k - M_i)^+ \\
  \label{eq:dof-region-general-variables-5}
  a_{ij}^\ZF &\leq& (M_i - M_k)^+ \\
  \label{eq:dof-region-general-variables-6}
  \min \{ a_{ij}^\ZF, d_{ij}, \gamma_i \} &\geq& 0
\end{IEEEeqnarray}

Instantiating these constraints for every possible combination of $i, j, k \in \{ 1, 2, 3 \}$ mutually distinct, resubstituting all $\tau_i$ from \eqref{eq:def-abbreviation-tau-s}, collecting the resulting bounds and eliminating redundant bounds, finally yields (for both $(i, \overline{i}) \in \{ (2,3), (3,2) \}$):
\begin{IEEEeqnarray}{rCl}
  \label{eq:achievable-scheme-dof-region-nodeint-first}
  d_{12} + d_{13} &\leq& \tau M_1 \\
  d_{i1} + \tau d_{i\overline{i}} &\leq& \tau M_i \\
  d_{21} + d_{31} + \tau d_{23} + \tau d_{32} \quad\nonumber\\ {}-{} \tau a_{23}^\ZF - \tau a_{32}^\ZF - \tau \gamma_1 &\leq& \tau M_1 \\
  d_{1i} + \tau d_{\overline{i}i} + d_{1\overline{i}} + d_{\overline{i}1} \quad\nonumber\\ {}-{} \tau a_{1\overline{i}}^\ZF - \tau a_{\overline{i}1}^\ZF - \tau \gamma_i &\leq& \tau M_i \\
  0 \leq{} (\gamma_1 + a_{i\overline{i}}^\ZF) &\leq& d_{i\overline{i}} \\
  0 \leq{} \tau (\gamma_i + a_{1\overline{i}}^\ZF) &\leq& d_{1\overline{i}} \\
  0 \leq{} \tau (\gamma_i + a_{\overline{i}1}^\ZF) &\leq& d_{\overline{i}1} \\
  0 \leq{} \gamma_1 &\leq& (M_2 + M_3 - M_1)^+ \\
  0 \leq{} \gamma_i &\leq& (M_1 + M_{\overline{i}} - M_i)^+ \\
  0 \leq{} a_{1i}^\ZF &\leq& (M_1 - M_{\overline{i}})^+ \\
  0 \leq{} a_{i1}^\ZF &\leq& (M_i - M_{\overline{i}})^+ \\
  \label{eq:achievable-scheme-dof-region-nodeint-last}
  0 \leq{} a_{i\overline{i}}^\ZF &\leq& (M_i - M_1)^+
\end{IEEEeqnarray}

In order to resolve the $(.)^+$ expressions, we do the following for each of the twelve cases of $M_i \geq M_j + M_k \geq M_j \geq M_k$ and $M_j + M_k \geq M_i \geq M_j \geq M_k$ (for all possible combinations $i, j, k \in \{ 1, 2, 3 \}$ mutually distinct):
\begin{enumerate}
  \item Instantiate the $(.)^+$ expressions under the respective assumption on the numbers of antennas, therefore some of the $a_{ij}^\ZF$ and $\gamma_i$ will be forced to zero.
  \item Perform Fourier-Motzkin's elimination to remove all remaining $a_{ij}^\ZF$ and $\gamma_i$ and obtain the achievable DoF region.
\end{enumerate}
We then combine the resulting achievable DoF regions into the following compact formulation:
\begin{IEEEeqnarray}{rCl}
  \label{eq:node-intermittency-dof-region-ib-first}
  \max\{ d_{\NA{}2} + d_{\NA{}3}, d_{2\NA{}} + d_{3\NA{}} \} &\leq& \tau M_{\NA{}} \\
  \max\{ d_{\NB{}1} + \tau d_{\NB{}3}, d_{1\NB{}} + \tau d_{3\NB{}} \} &\leq& \tau M_{\NB{}} \\
  \label{eq:node-intermittency-dof-region-ib-last-cutset}
  \max\{ d_{\NC{}1} + \tau d_{\NC{}2}, d_{1\NC{}} + \tau d_{2\NC{}} \} &\leq& \tau M_{\NC{}} \\
  \label{eq:node-intermittency-dof-region-ib-first-genie}
  \max\{ d_{\NA{}2} + d_{\NA{}\NC{}} + \tau d_{2\NC{}}, \quad\nonumber\\ d_{2\NA{}} + d_{\NC{}\NA{}} + \tau d_{\NC{}2} \} &\leq& \tau \max\left\{ M_{\NA{}}, M_{\NC{}} \right\} \IEEEeqnarraynumspace\\
  \max\{ d_{\NA{}\NB{}} + d_{\NA{}3} + \tau d_{3\NB{}}, \quad\nonumber\\ d_{\NB{}\NA{}} + d_{3\NA{}} + \tau d_{\NB{}3} \} &\leq& \tau \max\left\{ M_{\NA{}}, M_{\NB{}} \right\} \IEEEeqnarraynumspace\\
  \label{eq:node-intermittency-dof-region-ib-last-genie}
  \max\{ d_{1\NB{}} + d_{\NC{}1} + \tau d_{\NC{}\NB{}}, \quad\nonumber\\ d_{\NB{}1} + d_{1\NC{}} + \tau d_{\NB{}\NC{}} \} &\leq& \tau \max\left\{ M_{\NB{}}, M_{\NC{}} \right\} \IEEEeqnarraynumspace\\
  \label{eq:node-intermittency-dof-region-ib-last}
  \min \{ d_{12}, d_{13}, d_{21}, d_{23}, d_{31}, d_{32} \} &\geq& 0
\end{IEEEeqnarray}
This region is achievable for tuples $\bm{d}$ with non-negative integer entries. Tuples with non-negative real entries (such as the corner points of the region) are first approximated by non-negative rationals which then can be achieved using symbol extension, as in \cite{ChaabanSezgin_YC_Reg}.

The set of all DoF tuples $\bm{d}$ satisfying constraints \eqref{eq:node-intermittency-dof-region-ib-first} to \eqref{eq:node-intermittency-dof-region-ib-last} is denoted by $\mathcal D_{\mathrm{IB},\overline{\mathrm{A}}}^\mathrm{I}$. This proves Theorem~\ref{thm:node-intermittency-dof-region-ib}. What is the rationale behind the DoF region inner bound \eqref{eq:node-intermittency-dof-region-ib-first} to \eqref{eq:node-intermittency-dof-region-ib-last}?

The first three inequalities constrain the sum-DoF of outbound and inbound streams at each node, similar to cut-set bounds. The next three inequalities follow this rule: For each of the three links $1 \leftrightarrow 2$, $1 \leftrightarrow 3$ and $2 \leftrightarrow 3$ there are two DoF variables, one for each direction (i.e., $d_{12}$ and $d_{21}$, etc.). For each link pick one direction. There are eight such combinations. Whenever a combination contains both DoF variables that occur in a node's outbound or inbound sum-DoF constraint, the number of antennas at this node appears in the $\max\{.\}$ operator at the right side of the inequality. This means that whenever a node's index appears two times as left or two times as right index, this node's index appears also in the $\max\{.\}$ on the right side.

As can be seen above, there are six cases where each case applies to two nodes each, while for the third node it does not, because the third node's index appears once as left and once as right index. There are two cases missing altogether, $d_{12} + d_{31} + \tau d_{23}$ and $d_{21} + d_{13} + \tau d_{32}$, where the indices of all three nodes appear once as left and once as right index. Depending on numbers of antennas, this achievable DoF region yields four bounds for the largest and two bounds for the second largest node from the inequalities \eqref{eq:node-intermittency-dof-region-ib-first-genie} to \eqref{eq:node-intermittency-dof-region-ib-last-genie}, and two bounds for the third largest node from \eqref{eq:node-intermittency-dof-region-ib-first} to \eqref{eq:node-intermittency-dof-region-ib-last-cutset}. The remaining bounds from \eqref{eq:node-intermittency-dof-region-ib-first} to \eqref{eq:node-intermittency-dof-region-ib-last-cutset} are inactive due to the tighter bounds from \eqref{eq:node-intermittency-dof-region-ib-first-genie} to \eqref{eq:node-intermittency-dof-region-ib-last-genie}.

Although we proved that the above region is achievable, it is still useful to provide a `recipe' which describes how a specific DoF tuple can be achieved. To obtain the actual allocation of signal dimensions $a_{ij}^\Mq$ for a DoF tuple $\bm{d}$ satisfying \eqref{eq:node-intermittency-dof-region-ib-first} to \eqref{eq:node-intermittency-dof-region-ib-last}, proceed as follows: First, use as many ZF resources as possible. Only once the ZF dimensions are exhausted, assign IA resources and align as much of the resulting interference as possible. Throughout the process, account for redundancy required by EC to be able to tolerate intermittency. As an example (Fig.~\ref{fig:achievability-numeric-example-signal-spaces}), assume $(M_1, M_2, M_3, \tau) = (5, 7, 4, 0.5)$. There is one ZF dimension $1 \rightarrow 2$, three ZF dimensions $2 \rightarrow 1$ and two ZF dimensions $2 \rightarrow 3$, all other communication cannot be zero-forced. We try to achieve $\bm{d} = (0.5, 0, 0.5, 4, 0, 4)$. To transmit on average half a symbol per channel access full-duplex over the intermittent $1 \leftrightarrow 2$, we use a rate $\frac{1}{2}$ EC over one available ZF dimension in each direction ($a_{12}^\ZF = a_{21}^\ZF = 1$), the remaining two ZF dimensions $2 \rightarrow 1$ remain unused. IA is not required for $1 \leftrightarrow 2$ ($a_{12}^\IA = a_{21}^\IA = 0$). No communication $1 \leftrightarrow 3$ takes place ($a_{13}^\ZF = a_{13}^\IA = a_{31}^\ZF = a_{31}^\IA = 0$). To transmit four symbols per channel access $2 \rightarrow 3$, we use the two available ZF dimensions for two of them ($a_{23}^\ZF = 2$), and two IA dimensions ($a_{23}^\IA = 2$) that occupy a two-dimensional interference subspace at node $1$. To transmit four symbols per channel access $3 \rightarrow 2$, we use four IA dimensions ($a_{32}^\IA = 4$), since ZF is not possible ($a_{32}^\ZF = 0$). All of them create interference at node $1$, but this four-dimensional interference subspace can be aligned with the two-dimensional interference subspace caused by $2 \rightarrow 3$ ($\gamma_1 = 2$). As a result, four of the five receive dimensions at node $1$ are interference of $2 \leftrightarrow 3$ communication, while the remaining was used for zero-forced and erasure-coded communication $2 \rightarrow 1$. At nodes $2$ and $3$ no interference is caused, such that no interference alignment takes place ($\gamma_2 = \gamma_3 = 0$).

\begin{figure}[!t]
  \centering
  \begin{tikzpicture}[
      mynode/.append style={
        rectangle,
        draw,
        inner sep=2mm,
      },
      mymsg/.append style={
      },
      mychannel/.append style={
        thick,
        shorten >=\pgflinewidth*3,
        shorten <=\pgflinewidth*3,
      }
    ]%

    \node (n1) at (0,0.5) [mynode] {Node 1};
    \node (n2) at (-1.5,-1) [mynode] {Node 2};
    \node (n3) at (1.5,-1) [mynode] {Node 3};

    \node (t1) at ($(n1)+(0,2.5)$) [rectangle,draw,minimum width=3.0cm,minimum height=.6cm] {};
    \node (r1t) at ($(n1)+(0,1.6)$) [rectangle,draw,minimum width=3.0cm,minimum height=.6cm] {};
    \node (r1nt) at ($(n1)+(0,1)$) [rectangle,draw,minimum width=3.0cm,minimum height=.6cm] {};
    \node at ($(t1)-(1.85,0)$) {\footnotesize TX};
    \node at ($(r1t)-(1.85,0)$) {\footnotesize RX};
    \node at ($(r1t)+(1.75,0)$) {$\tau$};
    \node at ($(r1nt)+(1.75,0)$) {$\ntau$};

    \node (t2) at ($(n2)-(0.85,1)$) [rectangle,draw,minimum width=4.2cm,minimum height=.6cm] {};
    \node (r2t) at ($(n2)-(0.85,1.9)$) [rectangle,draw,minimum width=4.2cm,minimum height=.6cm] {};
    \node (r2nt) at ($(n2)-(0.85,2.5)$) [rectangle,draw,minimum width=4.2cm,minimum height=.6cm] {};
    \node at ($(t2)-(2.45,0)$) {\footnotesize TX};
    \node at ($(r2t)-(2.45,0)$) {\footnotesize RX};
    \node at ($(r2t)+(2.35,0)$) {$\tau$};
    \node at ($(r2nt)+(2.35,0)$) {$\ntau$};

    \node (t3) at ($(n3)-(-0.65,1)$) [rectangle,draw,minimum width=2.4cm,minimum height=.6cm] {};
    \node (r3t) at ($(n3)-(-0.65,1.9)$) [rectangle,draw,minimum width=2.4cm,minimum height=.6cm] {};
    \node (r3nt) at ($(n3)-(-0.65,2.5)$) [rectangle,draw,minimum width=2.4cm,minimum height=.6cm] {};
    \node at ($(t3)-(1.55,0)$) {\footnotesize TX};
    \node at ($(r3t)-(1.55,0)$) {\footnotesize RX};
    \node at ($(r3t)+(1.45,0)$) {$\tau$};
    \node at ($(r3nt)+(1.45,0)$) {$\ntau$};

    \draw[latex-latex,mychannel,densely dotted] (n1) to (n2);
    \draw[latex-latex,mychannel,densely dotted] (n1) to (n3);
    \draw[latex-latex,mychannel] (n2) to (n3);

    \node at ($(t1)+(-1.2,0)$) [rectangle,draw,pattern=north west lines,pattern color=myParula07Red!40,minimum height=.6cm,minimum width=.6cm] {\footnotesize $Z_1^2$};

    \node at ($(r1t)+(+0.3,0)$) [rectangle,draw,pattern=vertical lines,pattern color=myParula05Green!60,minimum height=.6cm,minimum width=2.4cm] {\hspace{0.8cm}\footnotesize $A_3^2$};
    \node at ($(r1t)+(-1.2,0)$) [rectangle,draw,pattern=north west lines,pattern color=myParula01Blue!40,minimum height=.6cm,minimum width=.6cm] {\footnotesize $Z_2^1$};
    \node at ($(r1t)+(-0.3,0)$) [rectangle,draw,pattern=horizontal lines,pattern color=myParula01Blue!40,minimum height=.6cm,minimum width=1.2cm] {\footnotesize $A_2^3$};

    \node at ($(t2)+(-1.8,0)$) [rectangle,draw,pattern=north west lines,pattern color=myParula01Blue!40,minimum height=.6cm,minimum width=.6cm] {\footnotesize $Z_2^1$};
    \node at ($(t2)+(-0.9,0)$) [rectangle,draw,pattern=north east lines,pattern color=myParula01Blue!40,minimum height=.6cm,minimum width=1.2cm] {\footnotesize $Z_2^3$};
    \node at ($(t2)+(+0.3,0)$) [rectangle,draw,pattern=horizontal lines,pattern color=myParula01Blue!40,minimum height=.6cm,minimum width=1.2cm] {\footnotesize $A_2^3$};

    \node at ($(r2t)+(-1.8,0)$) [rectangle,draw,pattern=north west lines,pattern color=myParula07Red!40,minimum height=.6cm,minimum width=.6cm] {\footnotesize $Z_1^2$};
    \node at ($(r2t)+(-0.3,0)$) [rectangle,draw,pattern=vertical lines,pattern color=myParula05Green!60,minimum height=.6cm,minimum width=2.4cm] {\footnotesize $A_3^2$};
    \node at ($(r2nt)+(-0.3,0)$) [rectangle,draw,pattern=vertical lines,pattern color=myParula05Green!60,minimum height=.6cm,minimum width=2.4cm] {\footnotesize $A_3^2$};

    \node at ($(t3)+(0,0)$) [rectangle,draw,pattern=vertical lines,pattern color=myParula05Green!40,minimum height=.6cm,minimum width=2.4cm] {\footnotesize $A_3^2$};

    \node at ($(r3t)+(-0.6,0)$) [rectangle,draw,pattern=north east lines,pattern color=myParula01Blue!40,minimum height=.6cm,minimum width=1.2cm] {\footnotesize $Z_2^3$};
    \node at ($(r3t)+(+0.6,0)$) [rectangle,draw,pattern=horizontal lines,pattern color=myParula01Blue!40,minimum height=.6cm,minimum width=1.2cm] {\footnotesize $A_2^3$};
    \node at ($(r3nt)+(-0.6,0)$) [rectangle,draw,pattern=north east lines,pattern color=myParula01Blue!40,minimum height=.6cm,minimum width=1.2cm] {\footnotesize $Z_2^3$};
    \node at ($(r3nt)+(+0.6,0)$) [rectangle,draw,pattern=horizontal lines,pattern color=myParula01Blue!40,minimum height=.6cm,minimum width=1.2cm] {\footnotesize $A_2^3$};
  \end{tikzpicture}
  \caption{Visualization of transmit (TX) and receive (RX) signal spaces achieving $\bm{d} = (0.5, 0, 0.5, 4, 0, 4)$ under the assumption of $(M_1, M_2, M_3, \tau) = (5, 7, 4, 0.5)$  (depending on intermittency state, where the first and second row are received with probability $\tau$ and $\ntau$, respectively; $Z_i^j$ denotes $w_{ij}^\ZF$, $A_i^j$ denotes $w_{ij}^\IA$)}
  \label{fig:achievability-numeric-example-signal-spaces}
\end{figure}
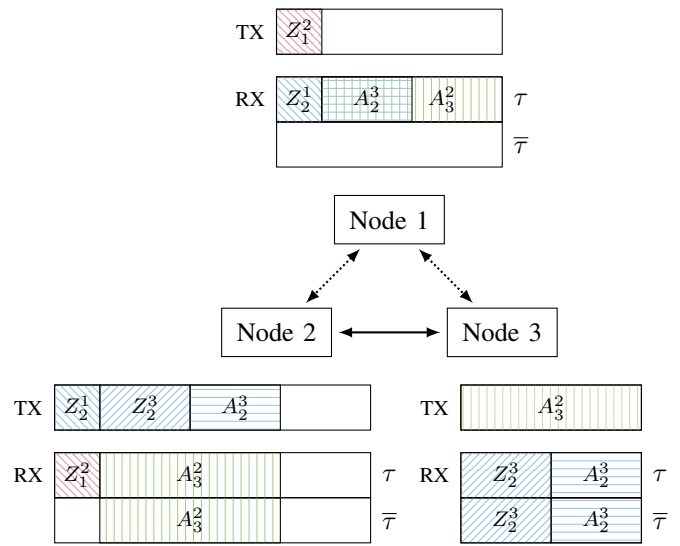

\subsection{Sum-DoF}

The devised non-adaptive encoding scheme based on ZF, IA and EC is sum-DoF optimal in the intermittent 3WC, as we will prove in this section.

\subsubsection{Lower Bounds}
\label{sec:nodeint-sum-dof-lb}

Using \eqref{eq:achievable-scheme-dof-region-nodeint-first} to \eqref{eq:achievable-scheme-dof-region-nodeint-last}, we derive a lower bound on the sum-DoF of the intermittent 3WC with intermittent node $1$. For each of the six cases of $M_i \geq M_j \geq M_k$ (for all possible combinations $i, j, k \in \{ 1, 2, 3 \}$ mutually distinct), we solve the linear program maximizing sum-DoF using, e.g., the simplex algorithm.

\begin{definition}
  We denote
  \[ d_{\mathrm{sum},\mathrm{LB},\overline{\mathrm{A}}}^\mathrm{I} \eqdef \max_{\bm{d} \in \mathcal D_{\mathrm{IB},\overline{\mathrm{A}}}^\mathrm{I}} \left[ d_{12} + d_{13} + d_{21} + d_{23} + d_{31} + d_{32} \right]. \]
\end{definition}

\begin{table}[!t]
	\caption{Achievable Sum-DoFs of Intermittent 3WC}
  \label{tab:node-intermittency-sum-dof-lb-cases}
	\centering
	\begin{tabular}{cl}
		\toprule
		Case & Sum-DoF Lower Bound \\
		\midrule
		$M_1 \geq M_2 \geq M_3$ & $d_{\mathrm{sum},\mathrm{LB},\overline{\mathrm{A}}}^\mathrm{I} = 2 M_3 + 2 \tau M_2 - 2 \tau M_3$ \\
    $M_2 \geq M_1 \geq M_3$ & $d_{\mathrm{sum},\mathrm{LB},\overline{\mathrm{A}}}^\mathrm{I} = 2 M_3 + 2 \tau M_1 - 2 \tau M_3$ \\
    $M_2 \geq M_3 \geq M_1$ & $d_{\mathrm{sum},\mathrm{LB},\overline{\mathrm{A}}}^\mathrm{I} = 2 M_3$ \\
		$M_1 \geq M_3 \geq M_2$ & $d_{\mathrm{sum},\mathrm{LB},\overline{\mathrm{A}}}^\mathrm{I} = 2 M_2 + 2 \tau M_3 - 2 \tau M_2$ \\
    $M_3 \geq M_1 \geq M_2$ & $d_{\mathrm{sum},\mathrm{LB},\overline{\mathrm{A}}}^\mathrm{I} = 2 M_2 + 2 \tau M_1 - 2 \tau M_2$ \\
    $M_3 \geq M_2 \geq M_1$ & $d_{\mathrm{sum},\mathrm{LB},\overline{\mathrm{A}}}^\mathrm{I} = 2 M_2$ \\
		\bottomrule
	\end{tabular}
\end{table}

The resulting lower bounds for each of the different cases are listed in Table~\ref{tab:node-intermittency-sum-dof-lb-cases}, and are condensed into a single expression in the following lemma:

\begin{lemma}[Sum-DoF Lower Bound for Node-Intermittent 3WC]
  \label{thm:node-intermittency-sum-dof-lb}
  The sum-DoF
  \begin{IEEEeqnarray}{rCl}
      \IEEEeqnarraymulticol{3}{l}{
          d_{\mathrm{sum},\mathrm{LB},\overline{\mathrm{A}}}^\mathrm{I}
      }\nonumber\\
      \quad&=&
          2 \ntau \min\{ M_2, M_3 \} + 2 \tau \Big( M_1+M_2+M_3 \nonumber\\
          && \quad {}-{} \min\{M_1,M_2,M_3\}-\max\{M_1,M_2,M_3\} \Big)
          \nonumber\\
        &\leq& d_{\mathrm{sum}}^\mathrm{I} \nonumber
  \end{IEEEeqnarray}
  is achievable in the node-intermittent 3WC and therefore constitutes a lower bound on the sum-DoF of the node-intermittent 3WC.
\end{lemma}

\subsubsection{Upper Bounds}
\label{sec:nodeint-sum-dof-ub}

We first motivate the converse techniques used throughout this section, then summarize the resulting upper bounds in Lemma~\ref{thm:node-intermittency-sum-dof-ub}, and in the sequel provide rigorous proofs. The general approach for upper bounding the sum-DoF of the intermittent 3WC is as follows: Partition the DoF sum $d_{ij} + d_{ik} + d_{ji} + d_{jk} + d_{ki} + d_{kj}$ into two partial sums $d_{ij} + d_{kj} + d_{ki}$ and $d_{ik} + d_{jk} + d_{ji}$ (Fig.~\ref{fig:nodeint-converse-general-nodes}), where $w_{ij}, w_{kj}, w_{ki}$ are to be decoded by node $j$ and $w_{ik}, w_{jk}, w_{ji}$ are to be decoded by node $k$. There are three such partitions, and the partition is fully determined by choosing which node takes the role of node $i$. While nodes $j$ and $k$ function exclusively as source and sink in any one of the partial sums, node $i$ is an intermediary node in both.

\begin{figure}[t]
  \centering
  \subfloat[]{%
    \begin{tikzpicture}%
      [msgarrow/.append style={
        thick,
        shorten >=\pgflinewidth*3,
        shorten <=\pgflinewidth*3,
      }]

      \node [rectangle,draw,inner sep=2mm] (node2) at (-1,0) {$j$};
      \node [rectangle,draw,inner sep=2mm] (node3) at (+1,0) {$k$};
      \node [rectangle,draw,inner sep=2mm] (node1) at (0,+1.4) {$i$};

      \draw[-latex,msgarrow] (node1) -- (node2);
      \draw[latex-,msgarrow] (node1) -- (node3);
      \draw[latex-,msgarrow] (node2) -- (node3);

    \end{tikzpicture}%
    \label{fig:nodeint-converse-general-nodes-around-j}%
  }
  \hspace{1cm}
  \subfloat[]{%
    \begin{tikzpicture}%
      [msgarrow/.append style={
        thick,
        shorten >=\pgflinewidth*3,
        shorten <=\pgflinewidth*3,
      }]

      \node [rectangle,draw,inner sep=2mm] (node2) at (-1,0) {$j$};
      \node [rectangle,draw,inner sep=2mm] (node3) at (+1,0) {$k$};
      \node [rectangle,draw,inner sep=2mm] (node1) at (0,+1.4) {$i$};

      \draw[latex-,msgarrow] (node1) -- (node2);
      \draw[-latex,msgarrow] (node1) -- (node3);
      \draw[-latex,msgarrow] (node2) -- (node3);

    \end{tikzpicture}%
    \label{fig:nodeint-converse-general-nodes-around-k}%
  }
  \caption{Partition the DoF sum $d_{ij} + d_{ik} + d_{ji} + d_{jk} + d_{ki} + d_{kj}$ into two partial sums $d_{ij} + d_{kj} + d_{ki}$ to be decoded at node $j$ \protect\subref{fig:nodeint-converse-general-nodes-around-j} and $d_{ik} + d_{jk} + d_{ji}$ to be decoded at node $k$ \protect\subref{fig:nodeint-converse-general-nodes-around-k}, where nodes $j$ and $k$ are provided enough side information such that they can recover the receive signal $\y_i^n$ of node $i$ and from it decode the one message not originally intended for them (additional $d_{ki}$ and $d_{ji}$ DoFs, respectively).}
  \label{fig:nodeint-converse-general-nodes}
\end{figure}
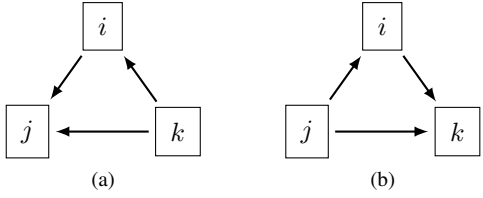

Have an imaginary genie provide just enough side information to node $j$ and $k$ (hence the name `genie-aided' bound), respectively, such that they can recover the receive signal $\y_i^n$ of node $i$, then (assuming existence of a suitable coding scheme) nodes $j$ and $k$ can decode the additional messages $w_{ki}$ and $w_{ji}$, respectively. The details of this decoding process and the required side information will be presented in due course. Here we only remark, that the side information can serve for the following four purposes (some cases might not require some of the types of side information):
\begin{enumerate}
  \item An additional message is required such that node $j$ (or $k$) can decode the additional message $w_{ki}$ (or $w_{ji}$) from the recovered $\y_i^n$, because decoding requires knowledge of $\w_i$, and node $j$ (or $k$) only knows $w_{ij}$ (or $w_{ik}$) from decoding its own receive signal $\y_j^n$ (or $\y_k^n$).
  \item Node $j$ (or $k$) might be incapable of capturing enough information about $\y_i^n$ because $M_j$ (or $M_k$) is small. In this case, additional measurements about $\y_i^n$ (or alternatively about the unknown `ingredient' of interest, $\x_k^n$ or $\x_j^n$, respectively) need to be provided by side information.
  \item Node $j$ (or $k$) might be incapable of capturing enough information about $\y_i^n$ because of intermittency. In this case, additional measurements about $\y_i^n$ for those time instances $\ell$ where $s_\ell = 0$ need to be provided by side information.
  \item For rather technical reasons a noise correction signal is required to accurately recover $\y_i^n$. However, the information contained in this signal about the three desired messages in question scales only as $\ologp$.
\end{enumerate}
Assuming reliable communication, the partial DoF sum of each three messages is necessarily upper bounded by a mutual information expression, using Fano's inequality.
Adding the two resulting bounds yields an upper bound on the sum-DoF. The challenge with this approach is two-fold: a) Provide as little side information as possible to the respective nodes. Otherwise the genie can be `mis-used' for information exchange, resulting in a larger DoF for the genie-enhanced system and thus loose bounds for the non-enhanced system. b) Use tight bounding when expanding the mutual information expression from Fano's inequality.

For the intermittent 3WC the partition that yields the tightest upper bound on the sum-DoF depends on the numbers of antennas. The node with largest number of antennas should take the role of the intermediary node $i$ (Fig.~\ref{fig:nodeint-converse-general-nodes}). The order among the remaining two nodes decides about which side information to give to which node, to compensate for insufficient number of antennas or intermittency. We prove upper bounds for the three cases $M_1 \geq M_2 \geq M_3$, $M_2 \geq M_1 \geq M_3$ and $M_2 \geq M_3 \geq M_1$, the remaining three cases go by renaming node $2$ and $3$. We obtain:

\begin{lemma}[Sum-DoF Upper Bound for Node-Intermittent 3WC]
  \label{thm:node-intermittency-sum-dof-ub}
  \begin{multline}
    d_{\mathrm{sum}}^\mathrm{I} \leq 2 \ntau \min\{ M_2, M_3 \} + 2 \tau \Big( M_1+M_2+M_3 \\ \quad {}-{} \min\{M_1,M_2,M_3\}-\max\{M_1,M_2,M_3\} \Big)
    \nonumber
  \end{multline}
\end{lemma}
\begin{proof}
  The lemma follows from \eqref{eq:nodeint-converse-M1geqM2geqM3}, \eqref{eq:nodeint-converse-M2geqM1geqM3}, \eqref{eq:nodeint-converse-M2geqM3geqM1} below, and symmetry of node $2$ and $3$.
\end{proof}

Using the achievability and converse results developed in Sections~\ref{sec:nodeint-sum-dof-lb} and \ref{sec:nodeint-sum-dof-ub}, we establish the sum-DoF of the intermittent 3WC for which non-adaptive encoding is sufficient, i.e.
\begin{IEEEeqnarray}{rCl}
    \IEEEeqnarraymulticol{3}{l}{
        d_{\mathrm{sum},\mathrm{LB},\overline{\mathrm{A}}}^\mathrm{I}
    }\nonumber\\
    \quad&=&
        2 \ntau \min\{ M_2, M_3 \} + 2 \tau \Big( M_1+M_2+M_3 \nonumber\\
      &&\quad {}-{} \min\{M_1,M_2,M_3\}-\max\{M_1,M_2,M_3\} \Big)
        \nonumber\\
      &=& d_{\mathrm{sum}}^\mathrm{I}.
        \nonumber
\end{IEEEeqnarray}

This proves Theorem~\ref{thm:node-intermittency-sum-dof}.
We proceed to present rigorous proofs for the aforementioned three cases.

\paragraph{Case 1: $M_1 \geq M_2 \geq M_3$}

For the case where $M_1$ is the largest number of antennas, we develop two partial sums around nodes $2$ and $3$ (Fig.~\ref{fig:nodeint-converse-general-nodes-around-j} and \ref{fig:nodeint-converse-general-nodes-around-k}, with $(i,j,k) = (1,2,3)$), respectively.
We start with the bound around node $2$ (Fig.~\ref{fig:nodeint-converse-general-nodes-around-j}, with $(i,j,k) = (1,2,3)$), as it requires less side information and is therefore simpler to argue, and in the sequel extend the basic technique to develop the bound around node $3$ (Fig.~\ref{fig:nodeint-converse-general-nodes-around-k}, with $(i,j,k) = (1,2,3)$), which requires more side information and is therefore slightly more involved.

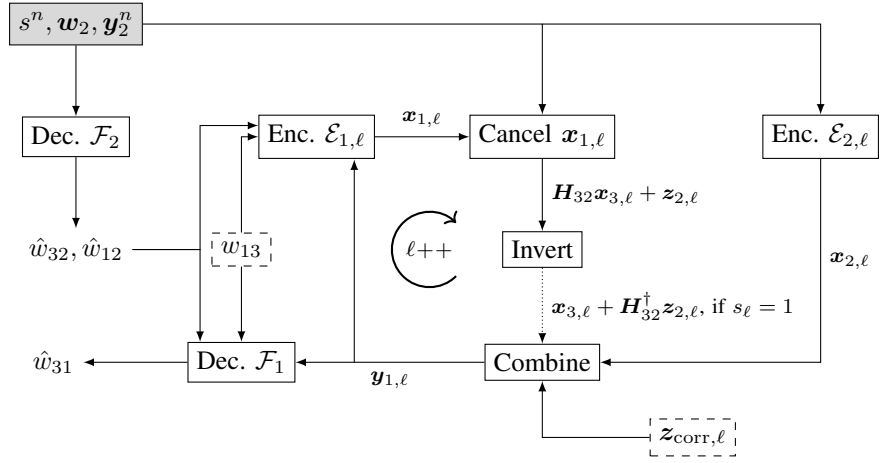
\begin{figure*}[!t]
  \centering
  \begin{tikzpicture}
    \node (in) at (0,0) [draw,fill=gray!30] {$s^n,\w_2,\y_2^n$};
    \node (d3) at ($(in)-(0,1.5)$) [rectangle,draw] {Dec. $\mathcal{F}_2$};
    \draw[-latex] (in) to (d3);
    \node (Wh3) at ($(d3)-(0,1.5)$) {$\hat{w}_{32},\hat{w}_{12}$};
    \draw[-latex] (d3) to (Wh3);
    \node (e1) at ($(d3)+(3.2,0)$) [rectangle,draw] {Enc. $\mathcal{E}_{1,\ell}$};
    \draw[-latex] (Wh3) to ($(e1)-(1.55,1.5)$) to ($(e1)-(1.55,-0.15)$) to ($(e1)-(0.775,-0.15)$);
    \draw[-latex] (Wh3) to ($(e1)-(1.55,1.5)$) to ($(e1)-(1.55,2.73)$);
    \node (si1) at ($(e1)-(1.0,1.5)$) [rectangle,draw,dashed] {$w_{13}$};
    \draw[-latex] (si1) to ($(si1)+(0,1.5)$) to (e1);
    \node (b) at ($(e1)+(3,0)$) [rectangle,draw] {Cancel $\x_{1,\ell}$};
    \draw[-latex] (e1) to node[above] {\footnotesize $\x_{1,\ell}$} (b);
    \draw[-latex] (in) to ($(b)+(0,1.5)$) to (b);
    \node (bb) at ($(b)-(0,1.5)$) [rectangle,draw] {Invert};
    \draw[<-,thick] ($(bb)-(1.5,0)$) ++(45:0.5) arc(45:315:0.5);
    \node (iteratelabel) at ($(bb)-(1.5,0)$) {\footnotesize $\mathbf{\ell\!+\!+}$};
    \draw[-latex] (b) to node[right] {\footnotesize $\H_{32}\x_{3,\ell}+\z_{2,\ell}$} (bb);
    \node (bbb) at ($(bb)-(0,1.5)$) [rectangle,draw] {Combine};
    \draw[-latex,densely dotted] (bb) to node[right] {\footnotesize $\x_{3,\ell} + \H_{32}^\pinv \z_{2,\ell}$, if $s_{\ell}=1$} (bbb);
    \node (e3) at ($(b)+(3.7,0)$) [rectangle,draw] {Enc. $\mathcal{E}_{2,\ell}$};
    \draw[-latex] ($(b)+(0,1.5)$) to ($(e3)+(0,1.5)$) to (e3);
    \draw[-latex] (e3) to node[right] {\footnotesize $\x_{2,\ell}$}  ($(e3)-(0,3)$) to (bbb);
    \node (d1) at ($(bbb)-(4,0)$) [rectangle,draw] {Dec. $\mathcal{F}_{1}$};
    \draw[-latex] (si1) to (d1);
    \draw[-latex] (bbb) to node[below] {\footnotesize $\y_{1,\ell}$} (d1);
    \draw[-latex] ($(bbb)-(2.5,0)$) to ($(bbb)+(-2.5,2.7)$) ;
    \node (Wh1) at ($(d1)-(2.5,0)$) {$\hat{w}_{31}$};
    \draw[-latex] (d1) to (Wh1);
    \node (zc) at ($(bbb)+(2,-1)$) [rectangle,draw,dashed] {$\zcorr{\ell}$};
    \draw[-latex] (zc) to ($(bbb)-(0,1)$) to (bbb);
  \end{tikzpicture}
  \caption{Decoding of $w_{31}$ at node $2$ by iterative reconstruction of $\y_1^n$ from \DAP{a-priori knowledge} $\w_2$ and \DO{observations} $(\y_2^n, s^n)$ using \DSI{\text{side information}} $(w_{13}, \z_{\mathrm{corr}}^n)$ (grey box: a-priori knowledge and channel output, dashed boxes: side information as introduced in the running text): successively obtain $\x_{1,\ell}$, cancel its effect from $\y_{2,\ell}$ to obtain a noisy version of $\x_{3,\ell}$, and combine this with $\x_{2,\ell}$ and $\zcorr{\ell}$ to finally obtain $\y_{1,\ell}$; repeat for next $\ell$.}
  \label{fig:converse-nodeint-M1geqM2geqM3-around2}
\end{figure*}

Which side information does node $2$ need to be able to recover $\y_1^n$ and decode $w_{31}$ from it, assuming a scheme allowing every node to decode its desired messages with high probability? We present a suitable iterative multi-step process depicted in Fig.~\ref{fig:converse-nodeint-M1geqM2geqM3-around2}.
Along the way, side information is introduced (\emph{highlighted in italic}) as found necessary for the reconstruction process.
At the end of the transmission, node $2$ has $\w_2$, $\y_2^n$, $s^n$ and $\x_2^n$, as shown on the top left of the figure. Using the decoder $\mathcal F_2$ it can decode $(\hat{w}_{12}, \hat{w}_{32}) = (w_{12}, w_{32})$ with high probability. \emph{Assume we provide $w_{13}$ as side information, so that node $2$ can decode messages intended for node $1$ using $\mathcal F_1$ as soon as it obtains $\y_1^n$, as shown on the bottom left of the figure.} Node $2$ can now obtain $\x_{1,1}$ from $\w_1$ using $\mathcal E_{1,1}$. From $(s_1, \y_{2,1})$ it can obtain $\H_{32} \x_{3,1} + \z_{2,1}$ using $\x_{1,1}$. Since $\H_{32}$ is a tall matrix, a noisy version of $\x_{3,1}$ can be obtained from $\H_{32} \x_{3,1} + \z_{2,1}$ using the pseudo-inverse $\H_{32}^\pinv$, i.e.,
\[ \H_{32}^\pinv \left( \H_{32} \x_{3,1} + \z_{2,1} \right) = \x_{3,1} + \H_{32}^\pinv \z_{2,1}. \]
Using the noisy version of $\x_{3,1}$, node $2$ can obtain a noisy version of $\y_{1,1}$, i.e.,
\begin{IEEEeqnarray}{rCl}
    \IEEEeqnarraymulticol{3}{l}{
        s_1 \left( \H_{21} \x_{2,1} + \H_{31} \left( \x_{3,1} + \H_{32}^\pinv \z_{2,1} \right) \right)
    }\nonumber\\
    \quad&=&
        s_1 \H_{21} \x_{2,1} + s_1 \H_{31} \x_{3,1} + s_1 \H_{31} \H_{32}^\pinv \z_{2,1}.
        \nonumber
\end{IEEEeqnarray}
\emph{Given a suitably formed noise correction term $\zcorr{1} \eqdef \z_{1,1} - s_1 \H_{31} \H_{32}^\pinv \z_{2,1}$ as side information, node $2$ can finally obtain $\y_{1,1}$, i.e.}
\[ \y_{1,1} = s_1 \H_{21} \x_{2,1} + s_1 \H_{31} \x_{3,1} + s_1 \H_{31} \H_{32}^\pinv \z_{2,1} + \zcorr{1}. \]
Then, this reconstruction cycle repeats for the next $\ell = 2, ..., n$, where all previous
$(\y_1^{\ell-1}, s^{\ell-1})$
are used to obtain $\x_{1,\ell}$ using $\mathcal E_{1,\ell}$.
After completing reconstruction of $\y_1^n$, node $2$ uses $\mathcal F_1$ and $(\hat{w}_{12}, w_{13})$ to decode $\hat{w}_{31} = w_{31}$ with high probability.
From the four abstract types of side information introduced before, only two are required for the reconstruction and subsequently for this bound: a message and a noise correction term. \emph{No side information to compensate for insufficient number of antennas or intermittency is required.}

In a nutshell, side information $w_{13}$ and $\zcorrS^n$ is provided to node $2$ by a genie, defined as
\begin{IEEEeqnarray}{rCl}
  \Z_{\mathrm{corr},\ell} &\eqdef& \Z_{1,\ell} - S_{\ell} \H_{31} \H_{32}^\pinv \Z_{2,\ell}. \nonumber
\end{IEEEeqnarray}
Since the scheme ought to be reliable, we bound the sum rate of $w_{12}$, $w_{32}$ and $w_{31}$ using Fano's inequality:
\begin{IEEEeqnarray}{rCl}
  \IEEEeqnarraymulticol{3}{l}{
      n(R_{12} + R_{32} + R_{31} - \varepsilon_n^{(1)})
  }\nonumber\\
  \quad
    &\leq& \MInf(W_{12} W_{32} W_{31} ; \W_2 \Y_2^n S^n \overbrace{W_{13} \ZcorrS^n}^{\text{side information}} ) \nonumber\\
    &\eqA& \MInf(W_{12} W_{32} W_{31} ; \Y_2^n \mid W_{13} \W_2 S^n \ZcorrS^n) \nonumber\\
    &\eqB& \sum_{\ell=1}^n \MInf(W_{12} \W_3; \Y_{2,\ell} \mid \Y_2^{\ell-1} S^n W_{13} \W_2 \ZcorrS^n) \nonumber\\
    &=& \sum_{\ell=1}^n \Big[ \hEntr(\Y_{2,\ell} \mid \Y_2^{\ell-1} S^n W_{13} \W_2 \ZcorrS^n) \nonumber\\&& \quad {}-{} \hEntr(\Y_{2,\ell} \mid \Y_2^{\ell-1} S^n \W_1 \W_2 \W_3 \ZcorrS^n) \Big] \nonumber\\
    &\leqC& \sum_{\ell=1}^n \Big[ \hEntr(\Y_{2,\ell} \mid S_{\ell}) \nonumber\\ && \quad {}-{} \hEntr(\Y_{2,\ell} \mid \Y_2^{\ell-1} S^n \W_1 \W_2 \W_3 \ZcorrS^n \X_{1,\ell} \X_{3,\ell}) \Big] \nonumber\\
    &\eqD& \sum_{\ell=1}^n \Big[ \hEntr(\Y_{2,\ell} \mid S_{\ell}) - \hEntr(\Y_{2,\ell} \mid S_{\ell} \Zcorr{\ell} \X_{1,\ell} \X_{3,\ell}) \Big] \nonumber\\
    &=& \sum_{\ell=1}^n \MInf(\Zcorr{\ell} \X_{1,\ell} \X_{3,\ell} ; \Y_{2,\ell} \mid S_{\ell}) \nonumber\\
    &\eqB& \sum_{\ell=1}^n \Big[ \MInf(\X_{1,\ell} \X_{3,\ell} ; \Y_{2,\ell} \mid S_{\ell}) \nonumber\\&&\quad {}+{} \MInf(\Zcorr{\ell} ; \Y_{2,\ell} \mid S_{\ell} \X_{1,\ell} \X_{3,\ell}) \Big] \nonumber\\
    &\eqE& \sum_{\ell=1}^n \Big[ \MInf(\X_{1,\ell} \X_{3,\ell} ; \Y_{2,\ell} \mid S_{\ell}) + \overbrace{\MInf(\Zcorr{\ell} ; \Z_{2,\ell} \mid S_{\ell})}^{=\ologp} \Big] \nonumber\\
    &\leqF& n \left[ \tau M_2 + \ntau M_3 \right] \logp + n \ologp \nonumber
\end{IEEEeqnarray}
These steps are justified as follows:
\begin{itemize}
  \item[\stepA] $(W_{12}, \W_3)$ is independent of $(W_{13}, \W_2, S^n, \ZcorrS^n)$ %
  \item[\stepB] Chain rule for mutual information
  \item[\stepC] Conditioning reduces entropy %
  \item[\stepD] $\Y_{2,\ell}$ is independent of $(\Y_2^{\ell-1}, S^{\ell-1}, S_{\ell+1}^n,\allowbreak \ZcorrS^{\ell-1},\allowbreak \Zcorr{\ell+1}^n,\allowbreak \W_1,\allowbreak \W_2, \W_3)$ given $(S_{\ell}, \Zcorr{\ell}, \X_{1,\ell}, \X_{3,\ell})$
  \item[\stepE] %
  $\MInf(\Zcorr{\ell} ; \Y_{2,\ell} \mid S_{\ell} \X_{1,\ell} \X_{3,\ell}) = \MInf(\Zcorr{\ell} ; \Z_{2,\ell} \mid S_{\ell} \X_{1,\ell}\allowbreak \X_{3,\ell})$, and $(\Zcorr{\ell},\Z_{2,\ell})$ is independent of $(\X_{1,\ell}, \X_{3,\ell})$ given $S_{\ell}$
  \item[\stepF] $(\X_{1,\ell}, \X_{3,\ell}) \leadsto \Y_{2,\ell}$ is a MIMO channel with $\min\{M_1+M_3,M_2\}=M_2$ DoFs if $s_\ell=1$, and $\min\{0+M_3,M_2\}=M_3$ DoFs if $s_\ell=0$
\end{itemize}
Dividing both sides by $n \logp$ and letting $\rho, n \to \infty$ we obtain
\begin{IEEEeqnarray}{rCl}
  \label{eq:nodeint-converse-M1geqM2geqM3-around2}
  d_{12} + d_{32} + d_{31} &\leq& \tau M_2 + \ntau M_3.
\end{IEEEeqnarray}

We turn to the second partial sum, developed around node $3$ (Fig.~\ref{fig:nodeint-converse-general-nodes-around-k}, with $(i,j,k) = (1,2,3)$). This bound is slightly more involved, as an additional type of side information is required which compensates for the small number of antennas $M_3$.
Which side information does node $3$ need to be able to recover $\y_1^n$ and decode $w_{21}$ from it, assuming a scheme allowing every node to decode its desired messages with high probability? A suitable process is depicted in Fig.~\ref{fig:converse-nodeint-M1geqM2geqM3-around3}. At the end of the transmission, node $3$ has $\w_3$, $\y_3^n$, $s^n$ and $\x_3^n$, as shown on the top left of the figure. Using the decoder $\mathcal F_3$ it can decode $(\hat{w}_{13}, \hat{w}_{23}) = (w_{13}, w_{23})$ with high probability. \emph{Assume we provide $w_{12}$ as side information, so that node $3$ can decode messages intended for node $1$ using $\mathcal F_1$ as soon as it obtains $\y_1^n$, as shown on the bottom left of the figure.} Node $3$ can now obtain $\x_{1,1}$ from $\w_1$ using $\mathcal E_{1,1}$. From $(s_1, \y_{3,1})$ it can obtain $\H_{23} \x_{2,1} + \z_{3,1}$ using $\x_{1,1}$. \emph{Assume we `virtually' increase the number of antennas at node~$3$ so that it can fully observe $\x_{2,1}$ whenever node $1$ can, by providing $\tilde{\y}_{3,1} = \tilde{\H}_{23} \x_{2,1} + \tilde{\z}_{3,1}$ as side information if $s_1 = 1$, with $\tilde{\H}_{23} \in \mathbb C^{(M_2 - M_3) \times M_2}$ such that $\Rank(\left[\begin{smallmatrix} \H_{23} \\ \tilde{\H}_{23} \end{smallmatrix}\right]) = M_2$.} Note that if $s_1 = 0$ then $\x_{2,1}$ does not contribute to $\y_{1,1}$. In this case, $\x_{2,1}$ does not need to be reconstructed, and therefore no side information to compensate for insufficient number of antennas is required.
We define shortcuts to group receive signal $\y_3^n$ and side information $\tilde{\y}_3^n$ into a joint signal $\hat{\y}_3^n$, i.e.,
\[ \hat{\y}_{3,1} \eqdef \left[\begin{smallmatrix} \y_{3,1} \\ \tilde{\y}_{3,1} \end{smallmatrix}\right], \qquad{} \hat{\H}_{23} \eqdef \left[\begin{smallmatrix} \H_{23} \\ \tilde{\H}_{23} \end{smallmatrix}\right], \qquad{} \hat{\z}_{3,1} \eqdef \left[\begin{smallmatrix} \z_{3,1} \\ \tilde{\z}_{3,1} \end{smallmatrix}\right]. \]
A matrix $\tilde{\H}_{23}$ satisfying $\Rank(\hat{\H}_{23}) = M_2$ exists almost surely and it allows to obtain $\x_{2,1} + \hat{\H}_{23}^{-1} \hat{\z}_{3,1}$ if $s_1 = 1$. \emph{Assume we provide $\z_{\mathrm{corr},1} = \z_{1,1} - s_1 \H_{21} \hat{\H}_{23}^{-1} \hat{\z}_{3,1}$ as side information.} Then node $3$ can obtain $\y_{1,1}$ from $\z_{\mathrm{corr},1}$ if $s_1 = 0$, and from $\x_{3,1}$, $\x_{2,1} + \hat{\H}_{23}^{-1} \hat{\z}_{3,1}$ and $\z_{\mathrm{corr},1}$ if $s_1 = 1$, i.e.,
\begin{multline}
  \y_{1,1} = s_1 \H_{21} (\x_{2,1} + \hat{\H}_{23}^{-1} \hat{\z}_{3,1}) + s_1 \H_{31} \x_{3,1} + \zcorr{1} \nonumber\\ = s_1 \H_{21} \x_{2,1} + s_1 \H_{31} \x_{3,1} + \z_{1,1}. \nonumber
\end{multline}
Using $(\y_{1,1}, s_1, \w_1)$ and the encoder $\mathcal E_{1,2}$ node $3$ can obtain $\x_{1,2}$ and the cycle repeats, for $\ell=2,...,n$. Finally, node $3$ obtains $\y_1^n$, and decodes $w_{21}$ from $(\y_1^n, \w_1, s^n)$. \emph{Side information to compensate for intermittency is not required.}

\begin{figure*}[!t]
  \centering
  \begin{tikzpicture}
    \node (in) at (0,0) [draw,fill=gray!30] {$s^n,\w_3,\y_3^n$};
    \node (d3) at ($(in)-(0,1.5)$) [rectangle,draw] {Dec. $\mathcal{F}_3$};
    \draw[-latex] (in) to (d3);
    \node (Wh3) at ($(d3)-(0,1.5)$) {$\hat{w}_{23},\hat{w}_{13}$};
    \draw[-latex] (d3) to (Wh3);
    \node (e1) at ($(d3)+(3.2,0)$) [rectangle,draw] {Enc. $\mathcal{E}_{1,\ell}$};
    \draw[-latex] (Wh3) to ($(e1)-(1.55,1.5)$) to ($(e1)-(1.55,-0.15)$) to ($(e1)-(0.775,-0.15)$);
    \draw[-latex] (Wh3) to ($(e1)-(1.55,1.5)$) to ($(e1)-(1.55,2.73)$);
    \node (si1) at ($(e1)-(1.0,1.5)$) [rectangle,draw,dashed] {$w_{12}$};
    \draw[-latex] (si1) to ($(si1)+(0,1.5)$) to (e1);
    \node (b) at ($(e1)+(3,0)$) [rectangle,draw] {Cancel $\x_{1,\ell}$};
    \draw[-latex] (e1) to node[above] {\footnotesize $\x_{1,\ell}$} (b);
    \draw[-latex] (in) to ($(b)+(0,1.5)$) to (b);
    \node (bb) at ($(b)-(0,1.5)$) [rectangle,draw] {Invert};
    \draw[<-,thick] ($(bb)-(1.5,0)$) ++(45:0.5) arc(45:315:0.5);
    \node (iteratelabel) at ($(bb)-(1.5,0)$) {\footnotesize $\mathbf{\ell\!+\!+}$};
    \draw[-latex] (b) to node[right] {\footnotesize $\H_{23}\x_{2,\ell}+\z_{3,\ell}$} (bb);
    \node (s2) at ($(bb)+(2,0)$) [rectangle,draw,dashed] {$\tilde{\y}_{3,\ell}$};
    \draw[-latex] (s2) to (bb);
    \node (bbb) at ($(bb)-(0,1.5)$) [rectangle,draw] {Combine};
    \draw[-latex,densely dotted] (bb) to node[right] {\footnotesize $\x_{2,\ell} + \hat{\H}_{23}^{-1}\hat{\z}_{3,\ell}$, if $s_{\ell}=1$} (bbb);
    \node (e3) at ($(b)+(3.7,0)$) [rectangle,draw] {Enc. $\mathcal{E}_{3,\ell}$};
    \draw[-latex] ($(b)+(0,1.5)$) to ($(e3)+(0,1.5)$) to (e3);
    \draw[-latex] (e3) to node[right] {\footnotesize $\x_{3,\ell}$}  ($(e3)-(0,3)$) to (bbb);
    \node (d1) at ($(bbb)-(4,0)$) [rectangle,draw] {Dec. $\mathcal{F}_{1}$};
    \draw[-latex] (si1) to (d1);
    \draw[-latex] (bbb) to node[below] {\footnotesize $\y_{1,\ell}$} (d1);
    \draw[-latex] ($(bbb)-(2.5,0)$) to ($(bbb)+(-2.5,2.7)$) ;
    \node (Wh1) at ($(d1)-(2.5,0)$) {$\hat{w}_{21}$};
    \draw[-latex] (d1) to (Wh1);
    \node (zc) at ($(bbb)+(2,-1)$) [rectangle,draw,dashed] {$\zcorr{\ell}$};
    \draw[-latex] (zc) to ($(bbb)-(0,1)$) to (bbb);
  \end{tikzpicture}
  \caption{Decoding of $w_{21}$ at node $3$ by iterative reconstruction of $\y_1^n$ from \DAP{a-priori knowledge} $\w_3$ and \DO{observations} $(\y_3^n, s^n)$ using \DSI{\text{side information}} $(w_{12}, \tilde{\y}_3^n, \z_{\mathrm{corr}}^n)$ (grey box: a-priori knowledge and channel output, dashed boxes: side information): successively obtain $\x_{1,\ell}$, cancel its effect from $\y_{3,\ell}$ (using side information $\tilde{\y}_{3,\ell}$) to obtain a noisy version of $\x_{2,\ell}$, and combine this with $\x_{3,\ell}$ and $\zcorr{\ell}$ to finally obtain $\y_{1,\ell}$; repeat for next $\ell$.}
  \label{fig:converse-nodeint-M1geqM2geqM3-around3}
\end{figure*}
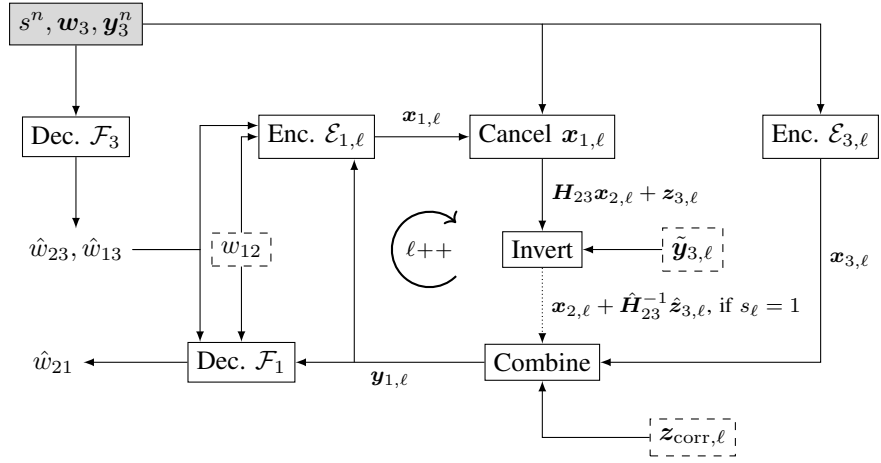

In a nutshell, side information $w_{12}$, $\tilde{\y}_3^n$ and $\z_{\mathrm{corr}}^n$ is provided to node $3$ by a genie, following the definitions
\begin{IEEEeqnarray}{rCl}
    \tilde{\Y}_{3,\ell} &\eqdef& S_{\ell} ( \tilde{\H}_{23} \X_{2,\ell} + \tilde{\Z}_{3,\ell} ), \nonumber\\
    \Z_{\mathrm{corr},\ell} &\eqdef& \Z_{1,\ell} - S_{\ell} ( \H_{21} \hat{\H}_{23}^{-1} \hat{\Z}_{3,\ell} ), \nonumber
\end{IEEEeqnarray}
where
\begin{IEEEeqnarray}{rCl}
  \hat{\H}_{23} &\eqdef& \left[\begin{smallmatrix} \H_{23} \nonumber\\ \tilde{\H}_{23} \end{smallmatrix}\right], \nonumber\\
  \tilde{\H}_{23} &\in& \mathbb C^{(M_2-M_3) \times M_2} \quad{} \text{such that } \Rank(\hat{\H}_{23}) = M_2, \nonumber\\
  \tilde{\Z}_{3,\ell} &\sim& \CN(\bm{0}, \sigma^2 \I_{M_2-M_3}), \nonumber\\
  \hat{\Z}_{3,\ell} &\eqdef& \left[\begin{smallmatrix} \Z_{3,\ell} \nonumber\\ \tilde{\Z}_{3,\ell} \end{smallmatrix}\right],
  \qquad{}
  \hat{\Y}_{3,\ell} \eqdef \left[\begin{smallmatrix} \Y_{3,\ell} \nonumber\\ \tilde{\Y}_{3,\ell} \end{smallmatrix}\right]. \nonumber
\end{IEEEeqnarray}

Since the scheme ought to be reliable, we again bound the sum rate of $w_{13}$, $w_{23}$ and $w_{21}$ using Fano's inequality, and following similar steps as before (for details see Appendix~\ref{sec:node-intermittency-sum-dof-case1-part2}) we obtain
\begin{IEEEeqnarray}{rCl}
  \label{eq:nodeint-converse-M1geqM2geqM3-around3}
  d_{13} + d_{23} + d_{21} &\leq& \tau M_2 + \ntau M_3.
\end{IEEEeqnarray}

Adding \eqref{eq:nodeint-converse-M1geqM2geqM3-around2} and \eqref{eq:nodeint-converse-M1geqM2geqM3-around3} yields a sum-DoF upper bound for the case $M_1 \geq M_2 \geq M_3$,
\begin{IEEEeqnarray}{rCl}
  \label{eq:nodeint-converse-M1geqM2geqM3}
  d_{\mathrm{sum}}^\mathrm{I} \leq 2 \tau M_2 + 2 \ntau M_3.
\end{IEEEeqnarray}

\paragraph{Case 2: $M_2 \geq M_1 \geq M_3$}

Since $M_2$ is the largest number of antennas, we develop two partial sums around nodes $3$ and $1$ (Fig.~\ref{fig:nodeint-converse-general-nodes-around-j} and \ref{fig:nodeint-converse-general-nodes-around-k}, with $(i,j,k)=(2,3,1)$), respectively.

The reasoning around node $3$ in this case proceeds in close analogy to the bound around node $3$ in the previous case, just with $1$ and $2$ interchanged. We provide as side information $w_{21}$ (to allow for decoding using $\mathcal F_2$), $\tilde{\y}_3^n$ (to compensate for small number of antennas $M_3$) and $\zcorrS^n$ (a noise correction), defined as
\begin{IEEEeqnarray}{rCl}
    \tilde{\Y}_{3,\ell} &\eqdef& S_{\ell} ( \tilde{\H}_{13} \X_{1,\ell} + \tilde{\Z}_{3,\ell} ), \nonumber\\
    \Z_{\mathrm{corr},\ell} &\eqdef& \Z_{2,\ell} - S_{\ell} ( \H_{12} \hat{\H}_{13}^{-1} \hat{\Z}_{3,\ell} ), \nonumber
\end{IEEEeqnarray}
where
\begin{IEEEeqnarray}{rCl}
  \hat{\H}_{13} &\eqdef& \left[\begin{smallmatrix} \H_{13} \nonumber\\ \tilde{\H}_{13} \end{smallmatrix}\right], \nonumber\\
  \tilde{\H}_{13} &\in& \mathbb C^{(M_1-M_3) \times M_1} \quad{} \text{such that } \Rank(\hat{\H}_{13}) = M_1, \nonumber\\
  \tilde{\Z}_{3,\ell} &\sim& \CN(\bm{0}, \sigma^2 \I_{M_1-M_3}), \nonumber\\
  \hat{\Z}_{3,\ell} &\eqdef& \left[\begin{smallmatrix} \Z_{3,\ell} \nonumber\\ \tilde{\Z}_{3,\ell} \end{smallmatrix}\right],
  \qquad{}
  \hat{\Y}_{3,\ell} \eqdef \left[\begin{smallmatrix} \Y_{3,\ell} \\ \tilde{\Y}_{3,\ell} \end{smallmatrix}\right]. \nonumber
\end{IEEEeqnarray}
With this information, node $3$ can construct $\x_{2,1}$ from $\w_2$ using $\mathcal E_{2,1}$ (since $w_{23}$ is a message intended for node $3$ and assumed to have been decoded from $\y_3^n$ using $\mathcal F_3$, and $w_{21}$ is side information), obtain a noisy version of $\x_{1,1}$ from channel output $\y_{3,1}$ and side information $\tilde{\y}_{3,1}$ as necessary for $\y_{2,1}$ (i.e., if $s_1=1$), combine all relevant signals into $\y_{2,1}$, encode $\w_2$ using $\mathcal E_{2,2}$ and
$(\y_{2,1}, s_1)$
to obtain $\x_{2,2}$, and continue this cycle for the next $\ell=2,...,n$ until $\y_2^n$ is complete, from which $w_{12}$ can be decoded with the help of $\w_2$ using $\mathcal F_2$.

After similar steps as before (see Appendix~\ref{sec:node-intermittency-sum-dof-case2-part1} for details) we obtain
\begin{IEEEeqnarray}{rCl}
  \label{eq:nodeint-converse-M2geqM1geqM3-around3}
  d_{13} + d_{23} + d_{12} &\leq& \tau M_1 + \ntau M_3.
\end{IEEEeqnarray}

We turn to the second partial sum, developed around node $1$ (Fig.~\ref{fig:nodeint-converse-general-nodes-around-k}, with $(i,j,k)=(2,3,1)$), where node $1$ should be enabled to decode $w_{32}$. The main difference to the previous cases is that the link $3 \leftrightarrow 2$ is always available, while the link $3 \leftrightarrow 1$ is intermittent. Therefore, $\y_{1,\ell}$ does not contain information about $\x_{3,\ell}$ if $s_\ell = 0$; this needs to be compensated for by side information, here $\tilde{\y}_1^n$ defined as
\[ \tilde{\y}_{1,\ell} \eqdef \overline{s}_{\ell} ( \H_{31} \x_{3,\ell} + \tilde{\z}_{1,\ell} ), \]
which provides measurements of $\x_3^n$ for those time instances where node $1$ is intermittent, i.e., $s_\ell = 0$. This is an instance of the fourth type of side information, that for previous bounds was not necessary, namely side information that compensates for intermittency. Given this side information and the customary side information (a message, to allow for decoding, and a noise correction signal), the reconstruction proceeds in analogy to the previous cases. In a nutshell, the genie provides $w_{23}$, $\tilde{\y}_1^n$ and $\zcorrS^n$ to node $1$, defined as
\begin{IEEEeqnarray}{rCl}
    \tilde{\Y}_{1,\ell} &\eqdef& \overline{S}_{\ell} ( \H_{31} \X_{3,\ell} + \tilde{\Z}_{1,\ell} ), \nonumber\\
    \Z_{\mathrm{corr},\ell} &\eqdef& \Z_{2,\ell} - \H_{32} \H_{31}^\pinv (S_{\ell} \Z_{1,\ell} + \overline{S}_{\ell} \tilde{\Z}_{1,\ell}), \nonumber
\end{IEEEeqnarray}
with
\begin{IEEEeqnarray}{rCl}
    \tilde{\Z}_{1,\ell} &\sim& \CN(\bm{0}, \sigma^2 \I_{M_1}), \nonumber\\
    \hat{\Z}_{1,\ell} &\eqdef& (\Z_{1,\ell}, \overline{S}_{\ell} \tilde{\Z}_{1,\ell}), \nonumber\\
    \hat{\Y}_{1,\ell} &\eqdef& (\Y_{1,\ell}, \tilde{\Y}_{1,\ell}). \nonumber
\end{IEEEeqnarray}
At the end of the transmission, node $1$ has $\w_1$, $\y_1^n$, $s^n$ and $\x_1^n$. It decodes $(w_{21}, w_{31})$ from $\y_1^n$ using its decoder $\mathcal F_1$, and gets $w_{23}$ from side information. It generates $\x_{2,1}$, then uses its channel output $\y_{1,1}$ (if $s_1 = 1$) or side information $\tilde{\y}_{1,1}$ (if $s_1 = 0$) to obtain a noisy version of $\x_{3,1}$, and with it $\y_{2,1}$. From there the cycle repeats, until $\y_2^n$ is obtained and $w_{32}$ can be decoded.

After similar steps as before (see Appendix~\ref{sec:node-intermittency-sum-dof-case2-part2} for details) we obtain
\begin{IEEEeqnarray}{rCl}
  \label{eq:nodeint-converse-M2geqM1geqM3-around1}
  d_{21} + d_{31} + d_{32} &\leq& \tau M_1 + \ntau M_3.
\end{IEEEeqnarray}

Adding \eqref{eq:nodeint-converse-M2geqM1geqM3-around3} and \eqref{eq:nodeint-converse-M2geqM1geqM3-around1} yields a sum-DoF upper bound for the case $M_2 \geq M_1 \geq M_3$,
\begin{IEEEeqnarray}{rCl}
  \label{eq:nodeint-converse-M2geqM1geqM3}
  d_{\mathrm{sum}}^\mathrm{I} \leq 2 \tau M_1 + 2 \ntau M_3.
\end{IEEEeqnarray}

\paragraph{Case 3: $M_2 \geq M_3 \geq M_1$}

Since $M_2$ is still the largest number of antennas, we again develop two partial sums around nodes $3$ and $1$ (Fig.~\ref{fig:nodeint-converse-general-nodes-around-j} and \ref{fig:nodeint-converse-general-nodes-around-k}, with $(i,j,k)=(2,3,1)$), respectively. The only difference to the previous case is that this time $M_3 \geq M_1$, therefore the number of antennas at node $1$ needs to be augmented `virtually' to obtain sufficient measurements of $\x_3^n$, while node $3$ remains unchanged.

We turn to the partial sum around node $3$ and provide as side information $w_{21}$ and $\zcorrS^n$ with
\begin{IEEEeqnarray}{rCl}
  \Z_{\mathrm{corr},\ell} &\eqdef& \Z_{2,\ell} - S_{\ell} ( \H_{12} \H_{13}^\pinv \Z_{3,\ell} ). \nonumber
\end{IEEEeqnarray}
Note that $\x_{1,\ell}$ contributes to $\y_{2,\ell}$ only if $s_{\ell} = 1$. In these instances, node $3$ has sufficient information about $\x_{1,\ell}$ from $\y_{3,\ell}$. If $s_{\ell}=0$, node $3$ does not have information about $\x_{1,\ell}$, but $\x_{1,\ell}$ does not contribute to $\y_{2,\ell}$ anyhow, so node $3$ does not need additional side information in these cases. Therefore, with the given side information, node $3$ can construct $\x_{2,1}$ from $\w_2$, obtain a noisy version of $\x_{1,1}$ from channel output $\y_{3,1}$ as necessary for $\y_{2,1}$ (i.e., if $s_1=1$), generate $\y_{2,1}$, and continue this cycle for the next $\ell=2,...,n$ until $\y_2^n$ is complete, from which $w_{12}$ can be decoded with the help of $\w_2$.

After similar steps as before (see Appendix~\ref{sec:node-intermittency-sum-dof-case3-part1} for details) we obtain
\begin{IEEEeqnarray}{rCl}
  \label{eq:nodeint-converse-M2geqM3geqM1-around3}
  d_{13} + d_{23} + d_{12} &\leq& M_3.
\end{IEEEeqnarray}

We turn to the second partial sum, developed around node $1$ (Fig.~\ref{fig:nodeint-converse-general-nodes-around-k}, with $(i,j,k)=(2,3,1)$), where node $1$ should be enabled to decode $w_{32}$. Again the main difference to the previous cases is that the link $3 \leftrightarrow 2$ is always available, while the link $3 \leftrightarrow 1$ is intermittent. Therefore, $\y_{1,\ell}$ does not contain information about $\x_{3,\ell}$ if $s_\ell = 0$; this effect of intermittency needs to be compensated for by side information, here $\tilde{\y}_1^n$ (a formal definition follows). Furthermore, the number of antennas at node $1$ needs to be increased to fully capture $\x_{3,\ell}$, here accomplished by side information $\breve{\y}_1^n$ (a formal definition follows). In addition, the genie provides the message $w_{23}$ and noise correction $\zcorrS^n$ to node $1$. In a nutshell, side information variables are defined as
\begin{IEEEeqnarray}{rCl}
    \tilde{\Y}_{1,\ell} &\eqdef& \overline{S}_{\ell} ( \H_{31} \X_{3,\ell} + \tilde{\Z}_{1,\ell} ), \nonumber\\
    \breve{\Y}_{1,\ell} &\eqdef& \breve{\H}_{31} \X_{3,\ell} + \breve{\Z}_{1,\ell}, \nonumber\\
    \Z_{\mathrm{corr},\ell} &\eqdef& \Z_{2,\ell} - \H_{32} \hat{\H}_{31}^{-1} \left[\begin{smallmatrix} S_{\ell} \Z_{1,\ell} + \overline{S}_{\ell} \tilde{\Z}_{1,\ell} \\ \breve{\Z}_{1,\ell} \end{smallmatrix} \right], \nonumber
\end{IEEEeqnarray}
with auxiliary variables
\begin{IEEEeqnarray}{rCl}
  \tilde{\Z}_{1,\ell} &\sim& \CN(\bm{0}, \sigma^2 \I_{M_1}), \nonumber\\
  \breve{\Z}_{1,\ell} &\sim& \CN(\bm{0}, \sigma^2 \I_{M_3-M_1}), \nonumber\\
  \hat{\H}_{31} &\eqdef& \left[\begin{smallmatrix} \H_{31} \\ \breve{\H}_{31} \end{smallmatrix}\right], \nonumber\\
  \breve{\H}_{31} &\in& \mathbb C^{(M_3-M_1) \times M_3} \quad{} \text{such that } \Rank(\hat{\H}_{31}) = M_3, \nonumber\\
  \hat{\Y}_{1,\ell} &\eqdef& (\Y_{1,\ell}, \tilde{\Y}_{1,\ell}, \breve{\Y}_{1,\ell}), \nonumber\\
  \hat{\Z}_{1,\ell} &\eqdef& (\Z_{1,\ell}, \overline{S}_{\ell} \tilde{\Z}_{1,\ell}, \breve{\Z}_{1,\ell}). \nonumber
\end{IEEEeqnarray}
At the end of the transmission, node $1$ has $\w_1$, $\y_1^n$, $s^n$ and $\x_1^n$. It decodes $(w_{21}, w_{31})$ using $\mathcal F_1$, and gets $w_{23}$ from side information. It generates $\x_{2,1}$ using $\mathcal E_{2,1}$, then uses side information $\breve{\y}_{1,1}$ and channel output $\y_{1,1}$ (if $s_1 = 1$) or side information $\breve{\y}_{1,1}$ and $\tilde{\y}_{1,1}$ (if $s_1 = 0$) to obtain a noisy version of $\x_{3,1}$, and with it $\y_{2,1}$ using $\zcorr{1}$. From there the cycle repeats, until $\y_2^n$ is obtained and $w_{32}$ can be decoded using $\mathcal F_2$.

After similar steps as before (see Appendix~\ref{sec:node-intermittency-sum-dof-case3-part2} for details) we obtain
\begin{IEEEeqnarray}{rCl}
  \label{eq:nodeint-converse-M2geqM3geqM1-around1}
  d_{21} + d_{31} + d_{32} &\leq& M_3.
\end{IEEEeqnarray}

Adding \eqref{eq:nodeint-converse-M2geqM3geqM1-around3} and \eqref{eq:nodeint-converse-M2geqM3geqM1-around1} yields a sum-DoF upper bound for the case $M_2 \geq M_3 \geq M_1$,
\begin{IEEEeqnarray}{rCl}
  \label{eq:nodeint-converse-M2geqM3geqM1}
  d_{\mathrm{sum}}^\mathrm{I} \leq 2 M_3.
\end{IEEEeqnarray}

Lemma~\ref{thm:node-intermittency-sum-dof-ub} follows from \eqref{eq:nodeint-converse-M1geqM2geqM3}, \eqref{eq:nodeint-converse-M2geqM1geqM3}, \eqref{eq:nodeint-converse-M2geqM3geqM1}, and symmetry of node $2$ and $3$. The achievability and converse results developed in Sections~\ref{sec:nodeint-sum-dof-lb} and \ref{sec:nodeint-sum-dof-ub} establish the sum-DoF of the intermittent 3WC and prove Theorem~\ref{thm:node-intermittency-sum-dof}.

Note that neither our achievability nor our converse result requires the noise at different receivers to be uncorrelated.
This complies with the intuition that the DoF perspective captures the impact of interference rather than noise.
For $\rho\to\infty$, usable DoFs become practically noise-free, hence it is also insignificant whether the noise is correlated or not.

\subsection{Necessity of Genie-Aided Upper Bounds}
\label{sec:necessity-of-genieaided-upper-bounds}

To underline the necessity of the genie-aided upper bounds devised in Section~\ref{sec:nodeint-sum-dof-ub},
we show that `classic' cut-set type bounds \cite{CoverThomas} admit DoF tuples $\bm{d} = (d_{12}, d_{13}, d_{21}, d_{23}, d_{31}, d_{32})$ that strictly exceed the sum-DoF of the intermittent 3WC stated in Theorem~\ref{thm:node-intermittency-sum-dof}.
For the intermittent 3WC with $M_1 \geq M_2 + M_3 \geq M_2 \geq M_3$, the cut-set bounds read
\begin{IEEEeqnarray}{rCll}
  \max\{ d_{12} + d_{13}, d_{21} + d_{31} \} &\leq& \tau (M_2+M_3),        &\quad \text{(Cut node $1$)}   \nonumber\\
  \max\{ d_{21} + d_{23}, d_{12} + d_{32} \} &\leq& \tau M_2 + \ntau M_3,  &\quad \text{(Cut node $2$)}   \nonumber\\
  \max\{ d_{31} + d_{32}, d_{13} + d_{23} \} &\leq& M_3.                   &\quad \text{(Cut node $3$)}   \nonumber
\end{IEEEeqnarray}
Let $M_2 = M_3 = 2$, $\tau = \ntau = \frac{1}{2}$, $M_1 = M_2 + M_3 = 4$.
Then, the cut-set bounds admit $\bm{d} = (1, 1, 1, 1, 1, 1)$, with a sum-DoF of $6$.
However, according to Theorem~\ref{thm:node-intermittency-sum-dof}, for $M_1 \geq M_2 + M_3 \geq M_2 \geq M_3$,
\[ d_{\mathrm{sum}}^\mathrm{I} = 2 \tau M_2 + 2 \ntau M_3, \]
in the example at hand, $d_{\mathrm{sum}}^\mathrm{I} = 4$.
Hence, cut-set type bounds are too loose to characterize the sum-DoF of the intermittent 3WC.
The genie-aided upper bounds devised in Section~\ref{sec:nodeint-sum-dof-ub} fill this gap.

\subsection{DoF Region}

In this section we first derive an upper bound on $d_{31}$ under the assumption of non-adaptive encoding.
We then show that adaptive schemes can exceed this bound, e.g., decode-forward relaying.
This proves that some DoF region points are only achievable by adaptive encoding schemes, hence adaptive encoding is in general required to achieve the DoF region of the intermittent 3WC.
Note that
we leave design and analysis of comprehensive adaptive encoding schemes for future work
and prove our claim by means of minimal counterexamples.
For the counterexample we may assume $M_1 \geq M_2 \geq M_3$.

\subsubsection{Upper Bound on $d_{31}$ under Non-Adaptive Encoding}
\label{sec:node-intermittency-dof-region-counterexample-converse}

Node $1$ is able to decode $w_{31}$ from its channel output $(\y_1^n, s^n)$ and a-priori knowledge $\w_1$ with high probability. We further provide $\x_2^n$ as side information and bound the rate of $w_{31}$ using Fano's inequality:
\begin{IEEEeqnarray}{rCl}
  \IEEEeqnarraymulticol{3}{l}{
      n(R_{31} - \varepsilon_n)
  }\nonumber\\
  \quad
    &\leq& \MInf(W_{31} ; \W_1 \Y_1^n S^n \overbrace{\X_2^n}^{\text{side information}} ) \nonumber\\
    &\eqA& \MInf(W_{31} ; \Y_1^n \mid \W_1 \X_2^n S^n) \nonumber\\
    &\eqB& \sum_{\ell=1}^n \MInf(W_{31} ; \Y_{1,\ell} \mid \Y_1^{\ell-1} \W_1 \X_2^n S^n) \nonumber\\
    &=& \sum_{\ell=1}^n \Big[ \hEntr(\Y_{1,\ell} \mid \Y_1^{\ell-1} \W_1 \X_2^n S^n) \nonumber\\&&\quad {}-{} \hEntr(\Y_{1,\ell} \mid \Y_1^{\ell-1} \W_1 \X_2^n S^n W_{31}) \Big] \nonumber\\
    &\leqC& \sum_{\ell=1}^n \Big[ \hEntr(\Y_{1,\ell} \mid \X_{2,\ell} S_{\ell}) \nonumber\\&&\quad {}-{} \hEntr(\Y_{1,\ell} \mid \Y_1^{\ell-1} \W_1 \X_2^n S^n W_{31} \X_{3,\ell}) \Big] \nonumber\\
    &\eqD& \sum_{\ell=1}^n \Big[ \hEntr(\Y_{1,\ell} \mid \X_{2,\ell} S_{\ell}) - \hEntr(\Y_{1,\ell} \mid \X_{2,\ell} S_{\ell} \X_{3,\ell}) \Big] \nonumber\\
    &=& \sum_{\ell=1}^n \MInf(\X_{3,\ell} ; \Y_{1,\ell} \mid \X_{2,\ell} S_{\ell}) \nonumber\\
    \label{eq:node-intermittency-dof-region-ob-d31}
    &\leqE& n \tau M_3 \logp + n \tau \ologp
\end{IEEEeqnarray}
These steps are justified as follows:
\begin{itemize}
  \item[\stepA] $W_{31}$ is independent of $(\W_1, \X_2^n, S^n)$ %
    due to non-adaptive encoding $\X_2^n = \mathcal E_{2}(\W_2)$
  \item[\stepB] Chain rule for mutual information
  \item[\stepC] Conditioning reduces entropy %
  \item[\stepD] $\Y_{1,\ell}$ is independent of $(\Y_1^{\ell-1}, \W_1, \X_{2,1}^{\ell-1}, \X_{2,\ell+1}^n,\allowbreak S^{\ell-1},\allowbreak S_{\ell+1}^n,\allowbreak W_{31})$ given $(S_{\ell}, \X_{2,\ell}, \X_{3,\ell})$
  \item[\stepE] $\X_{3,\ell} \leadsto \Y_{1,\ell}$ given $\X_{2,\ell}$ is a MIMO channel with maximum DoF $M_3$ and $0$ for $s_\ell=1$ and $s_\ell=0$, respectively
\end{itemize}
Dividing both sides of \eqref{eq:node-intermittency-dof-region-ob-d31} by $n \logp$ and taking $\rho, n \to \infty$, we see that the achievable DoF tuples of non-adaptive encoding schemes are constrained by
\begin{IEEEeqnarray}{rCl}
  \label{eq:nodeint-dof-region-restricted-ub1}
  d_{31} &\leq& \tau M_3.
\end{IEEEeqnarray}

\subsubsection{Adaptive Schemes Achieving $d_{31} > \tau M_3$}
\label{sec:node-intermittency-dof-region-counterexample-achievability}

We assume all messages are fixed to $0$, except for $w_{31}$, which node $3$ wants to convey to node $1$, potentially with the help of node $2$.
Consider using decode-forward relaying at node $2$. This can be used to achieve $\min\{ \tau M_1, \tau (M_2 + M_3), M_3 \}$ DoFs and outperforms any non-adaptive scheme as soon as $M_1 > M_3$. We derive the achievable DoF for $d_{31}$ based on the well-known lower bound for decode-forward relaying \cite{CoverElgamal}:
\begin{IEEEeqnarray}{rCl}
  C
    &\geq& \max_{p_{\X_3\X_2}} \min\{ \MInf( \X_3 \X_2 ; \Y_1 S ), \MInf( \X_3 ; \Y_2 S \mid \X_2 ) \} \nonumber\\
    &\geq& \min\{ \MInf( \X_3 \X_2 ; \Y_1 \mid S), \MInf( \X_3 ; \Y_2 \mid \X_2 ) \} \nonumber\\
    \IEEEeqnarraymulticol{3}{r}{
        \text{with } \X_2, \X_3 \text{ Gaussian}
    }\nonumber\\
    &=& \min\{ \tau \min\{ M_1, (M_2 + M_3) \}, M_3 \} \logp + \ologp \nonumber\\
    \label{eq:node-intermittency-dof-region-lb-d31}
    &=& \min\{ \tau M_1, \tau (M_2 + M_3), M_3 \} \logp + \ologp
\end{IEEEeqnarray}

Dividing both sides of \eqref{eq:node-intermittency-dof-region-lb-d31} by $\logp$ and taking $\rho \to \infty$, we see that the decode-forward relaying achieves
\begin{IEEEeqnarray}{rCl}
  \label{eq:nodeint-dof-region-adaptive-lb1}
  d_{31} \geq \min\{ \tau M_1, \tau (M_2 + M_3), M_3 \}.
\end{IEEEeqnarray}

Note that if $\tau M_2 > \ntau M_3$ and $\tau M_1 > M_3$, then we can transmit at $M_3$ DoF from node $3$ to node $1$ using this adaptive scheme, compensating all the the loss due to intermittency.

We proved in
\eqref{eq:nodeint-dof-region-restricted-ub1}
that the DoF region point
\begin{IEEEeqnarray}{rCl}
    \bm{d} &=& (0,0,0,0,d_{31,\mathrm{A}},0) \nonumber\\
    d_{31,\mathrm{A}} &\eqdef& \min\{ \tau M_1, \tau (M_2 + M_3), M_3 \} \nonumber
\end{IEEEeqnarray}
is not achievable for any non-adaptive encoding scheme if $M_1 > M_3$, while we showed in
\eqref{eq:nodeint-dof-region-adaptive-lb1}
that there exist adaptive schemes that achieve it. This proves Theorem~\ref{thm:node-intermittency-dof-region}, which states that adaptive encoding is in general required to achieve the DoF region of the intermittent 3WC.

Theorems~\ref{thm:node-intermittency-sum-dof} and \ref{thm:node-intermittency-dof-region} show that non-adaptive encoding is sufficient to achieve sum-DoF, but not sufficient to achieve the DoF region of the intermittent 3WC. This is particularly interesting in light of the next section, where we show that adaptive encoding is not beneficial in the non-intermittent 3WC even from a DoF region perspective.

\section{No Intermittency}
\label{sec:no-intermittency}

The sum-DoF of the non-intermittent 3WC was investigated in \cite{MaierChaabanMathar}.
We present the DoF region of the non-intermittent 3WC and show that the non-adaptive encoding scheme presented in Section~\ref{sec:node-intermittency-achievability} is sufficient to achieve it; therefore, adaptive encoding is neither required from a sum-DoF nor from a DoF region perspective in the non-intermittent 3WC.
The non-intermittent 3WC is a special case of the intermittent 3WC with $\tau=1$.
We may assume without loss of generality $M_1 \geq M_2 \geq M_3$.

\subsection{Achievability}
\label{sec:no-intermittency-achievability}

From \eqref{eq:node-intermittency-dof-region-ib-first} to \eqref{eq:node-intermittency-dof-region-ib-last} we obtain with $\tau=1$:
\begin{IEEEeqnarray}{rCl}
  \label{eq:no-intermittency-dof-region-1}
  d_{12} + d_{13} + d_{23} &\leq& M_1 \\
  \label{eq:no-intermittency-dof-region-2}
  d_{12} + d_{13} + d_{32} &\leq& M_1 \\
  \label{eq:no-intermittency-dof-region-3}
  d_{21} + d_{31} + d_{32} &\leq& M_1 \\
  \label{eq:no-intermittency-dof-region-4}
  d_{21} + d_{31} + d_{23} &\leq& M_1 \\
  \label{eq:no-intermittency-dof-region-5}
  d_{21} + d_{13} + d_{23} &\leq& M_2 \\
  \label{eq:no-intermittency-dof-region-6}
  d_{12} + d_{31} + d_{32} &\leq& M_2 \\
  \label{eq:no-intermittency-dof-region-7}
  d_{31} + d_{32} &\leq& M_3 \\
  \label{eq:no-intermittency-dof-region-8}
  d_{13} + d_{23} &\leq& M_3 \\
  \label{eq:no-intermittency-dof-region-9}
  \min \{ d_{12}, d_{13}, d_{21}, d_{23}, d_{31}, d_{32} \} &\geq& 0
\end{IEEEeqnarray}

All DoF tuples $\bm{d}$ satisfying constraints \eqref{eq:no-intermittency-dof-region-1} to \eqref{eq:no-intermittency-dof-region-9} are achievable in the non-intermittent 3WC. Therefore, by construction in Section~\ref{sec:node-intermittency-achievability}, said set of inequalities constitutes an inner bound on the DoF region of the non-intermittent 3WC.
We show in the following that the parametrization of the ZF/IA/EC-based scheme presented in Section~\ref{sec:node-intermittency-achievability} is sufficiently general to capture the whole DoF region of the non-intermittent 3WC.
Note that only in the generality elaborated in this paper
does the ZF/IA-based scheme achieve the DoF region of the non-intermittent 3WC
(instead of just its sum-DoF as in \cite{MaierChaabanMathar}).

\subsection{Converses}
\label{sec:no-intermittency-converses}

Previous works studied the sum-DoF of different variants of the non-intermittent 3WC. To this end, several bounds are reported in the literature, among them cut-set outer bounds
\cite{CoverThomas}
on pairs of DoFs (involving two messages either intended for or originating at a certain node), and tighter genie-aided outer bounds
on triplets of DoFs (involving two messages either intended for or originating at a certain node, and one message exchanged between the remaining two nodes).
The latter bounding technique is due to \cite{MaierChaabanMathar,MaierChaabanMathar_ITW} and was later taken up in \cite{ElmahdyKeyiMohassebElBatt,ElmahdyKeyiMohassebElBatt_journal}.
\begin{IEEEeqnarray}{rCl}
  \label{eq:no-intermittency-rate-region-ob-1}
  R_{13} + R_{23} &\leq& M_3 \logp + \ologp   \nonumber\\
    \IEEEeqnarraymulticol{3}{r}{\text{(see \cite[(7)]{MaierChaabanMathar}, \cite{CoverThomas})}} \quad \IEEEeqnarraynumspace\\
  \label{eq:no-intermittency-rate-region-ob-2}
  R_{31} + R_{32} &\leq& M_3 \logp + \ologp   \nonumber\\
    \IEEEeqnarraymulticol{3}{r}{\text{(see \cite[(8)]{MaierChaabanMathar}, \cite{CoverThomas})}} \quad \IEEEeqnarraynumspace\\
  \label{eq:no-intermittency-rate-region-ob-3}
  R_{21} + R_{31} + R_{32} &\leq& \min\{ M_1, M_2 + M_3 \} \logp + \ologp   \nonumber\\
    \IEEEeqnarraymulticol{3}{r}{\text{(see \cite[(23)]{ElmahdyKeyiMohassebElBatt}, \cite[(25)]{ElmahdyKeyiMohassebElBatt_journal})}} \quad \IEEEeqnarraynumspace\\
  \label{eq:no-intermittency-rate-region-ob-4}
  R_{21} + R_{31} + R_{23} &\leq& \min\{ M_1, M_2 + M_3 \} \logp + \ologp   \nonumber\\
    \IEEEeqnarraymulticol{3}{r}{\text{(see \cite[(25)]{ElmahdyKeyiMohassebElBatt}, \cite[(27)]{ElmahdyKeyiMohassebElBatt_journal})}} \quad \IEEEeqnarraynumspace\\
  \label{eq:no-intermittency-rate-region-ob-5}
  R_{12} + R_{32} + R_{13} &\leq& M_1 \logp + \ologp   \nonumber\\
    \IEEEeqnarraymulticol{3}{r}{\text{(see \cite[(29)]{ElmahdyKeyiMohassebElBatt_journal})}} \quad \IEEEeqnarraynumspace\\
  \label{eq:no-intermittency-rate-region-ob-6}
  R_{13} + R_{23} + R_{12} &\leq& M_1 \logp + \ologp   \nonumber\\
    \IEEEeqnarraymulticol{3}{r}{\text{(see \cite[(31)]{ElmahdyKeyiMohassebElBatt_journal})}} \quad \IEEEeqnarraynumspace\\
  \label{eq:no-intermittency-rate-region-ob-7}
  R_{12} + R_{32} + R_{31} &\leq& M_2 \logp + \ologp   \nonumber\\
    \IEEEeqnarraymulticol{3}{r}{\text{(see \eqref{eq:nodeint-converse-M1geqM2geqM3-around2}, \cite[(15)]{MaierChaabanMathar}, \cite[(28)]{ElmahdyKeyiMohassebElBatt_journal})}} \quad \IEEEeqnarraynumspace\\
  \label{eq:no-intermittency-rate-region-ob-8}
  R_{13} + R_{23} + R_{21} &\leq& M_2 \logp + \ologp   \nonumber\\
    \IEEEeqnarraymulticol{3}{r}{\text{(see \eqref{eq:nodeint-converse-M1geqM2geqM3-around3}, \cite[(11)]{MaierChaabanMathar}, \cite[(30)]{ElmahdyKeyiMohassebElBatt_journal})}} \quad \IEEEeqnarraynumspace
\end{IEEEeqnarray}

Note that if $M_2 + M_3 \leq M_1$, then \eqref{eq:no-intermittency-rate-region-ob-3} and \eqref{eq:no-intermittency-rate-region-ob-4} are redundant given \eqref{eq:no-intermittency-rate-region-ob-2} and \eqref{eq:no-intermittency-rate-region-ob-8}, therefore $\min\{ M_1, M_2 + M_3 \}$ can be replaced with $M_1$ in \eqref{eq:no-intermittency-rate-region-ob-3} and \eqref{eq:no-intermittency-rate-region-ob-4}. Dividing these bounds by $\logp$ and taking $\rho \to \infty$ yields to the following DoF region outer bounds:
\begin{IEEEeqnarray}{rCl}
  \label{eq:no-intermittency-dof-region-ob-3}
  d_{21} + d_{31} + d_{32} &\leq& M_1 \\
  \label{eq:no-intermittency-dof-region-ob-4}
  d_{21} + d_{31} + d_{23} &\leq& M_1 \\
  \label{eq:no-intermittency-dof-region-ob-5}
  d_{12} + d_{32} + d_{13} &\leq& M_1 \\
  \label{eq:no-intermittency-dof-region-ob-6}
  d_{13} + d_{23} + d_{12} &\leq& M_1 \\
  \label{eq:no-intermittency-dof-region-ob-7}
  d_{12} + d_{32} + d_{31} &\leq& M_2 \\
  \label{eq:no-intermittency-dof-region-ob-8}
  d_{13} + d_{23} + d_{21} &\leq& M_2 \\
  \label{eq:no-intermittency-dof-region-ob-1}
  d_{13} + d_{23} &\leq& M_3 \\
  \label{eq:no-intermittency-dof-region-ob-2}
  d_{31} + d_{32} &\leq& M_3
\end{IEEEeqnarray}
No DoF tuple $\bm{d}$ violating any of the constraints \eqref{eq:no-intermittency-dof-region-ob-3} to \eqref{eq:no-intermittency-dof-region-ob-2} can be achievable in the non-intermittent 3WC. Therefore, said set of inequalities constitutes an outer bound on the DoF region of the non-intermittent 3WC.

\subsection{DoF Region and Sum-DoF of Non-Intermittent 3WC}
\label{sec:sum-dof-nonintermittent}

The previous achievability and converse results establish the DoF region (and thus sum-DoF) optimality of non-adaptive schemes.
In particular, the scheme introduced in Section~\ref{sec:node-intermittency-achievability} is DoF region and sum-DoF optimal. This renders adaptive encoding dispensable for the non-intermittent 3WC and
proves Theorem~\ref{thm:no-intermittency-dof-region}. From the DoF region of the 3WC and using the sum-DoF of the intermittent 3WC, we reproduce the sum-DoF of the 3WC given in \cite{MaierChaabanMathar}:

\begin{corollary}[Sum-DoF of Non-Intermittent 3WC]
  \label{thm:no-intermittency-sum-dof}
  \[ d_{\mathrm{sum}}^\mathrm{N} = 2 M_2 \]
\end{corollary}
\begin{proof}
  The statement follows from Theorems~\ref{thm:node-intermittency-sum-dof} and \ref{thm:no-intermittency-dof-region}.
\end{proof}

\section{Conclusion}
\label{sec:conclusion}

We introduced the MIMO 3WC with node-intermittency and studied its DoF region and sum-DoF. In particular, we devised a non-adaptive encoding scheme based on zero-forcing, interference alignment and erasure coding, and showed its DoF region (and thus sum-DoF) optimality for non-intermittent 3WCs and its sum-DoF optimality for node-intermittent 3WCs. This shows that adaptive encoding is not required in those cases. However, we showed by example that in general there are DoF region points in the node-intermittent 3WC that can only be achieved by adaptive schemes, such as
decode-forward relaying, making adaptive encoding a necessity. Our work contributes to a better understanding of the necessity of adaptive schemes such as relaying in multi-way communications with intermittency.

As remarked in the introduction, node-intermittency is only one of a multitude of practically relevant intermittency scenarios. Links might be intermittent independently of each other, e.g., moving objects passing by only interrupt the link between the two D2D users from time to time, while the other links remain intact.
Or all links being intermittent, but independently of each other, and with different probabilities.
Here, we speak of \emph{link intermittency} and \emph{intermittent links}. Intermittent 3WCs with other intermittency models are interesting directions for future research.

\appendix

\subsection{Proof Template for Sum-DoF Upper Bounds for Intermittent 3WC}
\label{sec:node-intermittency-sum-dof-proof-template}

Throughout the derivations of sum-DoF upper bounds for the intermittent 3WC, certain steps reappear in slight variations.
To avoid repetition, we formulate the following `proof template', where the Fraktur variables $\frakWA$, $\frakWB$, $\frakX$, $\frakY$, and $\frakZ$ serve as placeholders and need to be replaced by random variables as specified in the context of the template's invocation.
We require that
\begin{itemize}
  \item[\stepA]
    $\frakWA$ is independent of $(\frakWB, S^n, \ZcorrS^n)$,
  \item[\stepB]
    $\frakY_\ell$ is independent of $(\frakY^{\ell-1}, S^{\ell-1},\allowbreak S_{\ell+1}^n,\allowbreak \ZcorrS^{\ell-1},\allowbreak \Zcorr{\ell+1}^n,\allowbreak \W_1, \W_2, \W_3)$ given $(S_{\ell}, \Zcorr{\ell}, \frakX_\ell)$,
  \item[\stepC]
    $\MInf(\Zcorr{\ell} ; \frakY_\ell \mid S_{\ell} \frakX_\ell) = \MInf(\Zcorr{\ell} ; \frakZ_\ell \mid S_{\ell} \frakX_\ell)$, and $(\Zcorr{\ell},\frakZ_\ell)$ is independent of $\frakX_\ell$ given $S_{\ell}$,
  \item[\stepD]
    $\MInf(\Zcorr{\ell} ; \frakZ_\ell \mid S_{\ell}) = \ologp$.
\end{itemize}
Note that these preconditions are satisfied for every invocation of the template in this paper.
Then we have
\begin{IEEEeqnarray}{rCl}
  \IEEEeqnarraymulticol{3}{l}{
      \MInf(\frakWA ; \frakWB \frakY^n S^n \ZcorrS^n)
  }\nonumber\\
  \quad
    &\eqA& \MInf(\frakWA ; \frakY^n \mid \frakWB S^n \ZcorrS^n) \nonumber\\
    &=& \sum_{\ell=1}^n \MInf(\frakWA; \frakY_\ell \mid \frakY^{\ell-1} \frakWB S^n \ZcorrS^n) \nonumber\\
    &=& \sum_{\ell=1}^n \Big[ \hEntr(\frakY_\ell \mid \frakY^{\ell-1} \frakWB S^n \ZcorrS^n) \nonumber\\&&\quad {}-{} \hEntr(\frakY_\ell \mid \frakY^{\ell-1} \W_1 \W_2 \W_3 S^n \ZcorrS^n) \Big] \nonumber\\
    &\leq& \sum_{\ell=1}^n \Big[ \hEntr(\frakY_\ell \mid S_{\ell}) \nonumber\\&&\quad {}-{} \hEntr(\frakY_\ell \mid \frakY^{\ell-1} \W_1 \W_2 \W_3 \frakX_\ell S^n \ZcorrS^n) \Big] \nonumber\\
    &\eqB& \sum_{\ell=1}^n \Big[ \hEntr(\frakY_\ell \mid S_{\ell}) - \hEntr(\frakY_\ell \mid S_{\ell} \Zcorr{\ell} \frakX_\ell) \Big] \nonumber\\
    &=& \sum_{\ell=1}^n \MInf(\Zcorr{\ell} \frakX_\ell ; \frakY_\ell \mid S_{\ell}) \nonumber\\
    &=& \sum_{\ell=1}^n \Big[ \MInf(\frakX_\ell ; \frakY_\ell \mid S_{\ell}) + \MInf(\Zcorr{\ell} ; \frakY_\ell \mid S_{\ell} \frakX_\ell) \Big] \nonumber\\
    &\eqC& \sum_{\ell=1}^n \Big[ \MInf(\frakX_\ell ; \frakY_\ell \mid S_{\ell}) + \MInf(\Zcorr{\ell} ; \frakZ_\ell \mid S_{\ell}) \Big] \nonumber\\
    &\eqD& \sum_{\ell=1}^n \MInf(\frakX_\ell ; \frakY_\ell \mid S_{\ell}) + n \ologp \nonumber
\end{IEEEeqnarray}
where the letters indicate the precondition that justifies each step.

\subsection{Sum-DoF Upper Bound for Intermittent 3WC with $M_1 \geq M_2 \geq M_3$ (Part II)}
\label{sec:node-intermittency-sum-dof-case1-part2}

Since the scheme ought to be reliable, we bound the sum rate of $w_{13}$, $w_{23}$ and $w_{21}$ using Fano's inequality:
\begin{IEEEeqnarray}{rCl}
  \IEEEeqnarraymulticol{3}{l}{
      n(R_{13} + R_{23} + R_{21} - \varepsilon_n^{(2)})
  }\nonumber\\
  \quad
    &\leq& \MInf(W_{13} W_{23} W_{21} ; \W_3 \Y_3^n S^n \overbrace{W_{12} \tilde{\Y}_3^n \ZcorrS^n}^{\text{side information}} ) \nonumber\\
    &\leqA& \sum_{\ell=1}^n \MInf(\X_{1,\ell} \X_{2,\ell} ; \hat{\Y}_{3,\ell} \mid S_{\ell}) + n \ologp \nonumber\\
    &\leqB& n \left[ \tau M_2 + \ntau M_3 \right] \logp + n \ologp \nonumber
\end{IEEEeqnarray}
These steps are justified as follows:
\begin{itemize}
  \item[\stepA] Using the proof template presented in Appendix~\ref{sec:node-intermittency-sum-dof-proof-template}, with $\frakWA \eqdef (W_{13},W_{21},W_{23})$, $\frakWB \eqdef (W_{12},W_{31},W_{32})$, $\frakX \eqdef (\X_1,\X_2)$, $\frakY \eqdef \hat{\Y}_3$, $\frakZ \eqdef \hat{\Z}_3$
  \item[\stepB] $(\X_{1,\ell}, \X_{2,\ell}) \leadsto (\Y_{3,\ell}, \tilde{\Y}_{3,\ell})$ is a MIMO channel with $\min\{M_1+M_2,M_3+(M_2-M_3)\}=M_2$ DoFs if $s_\ell=1$, and $\min\{M_1+M_2,M_3+0\}=M_3$ DoFs if $s_\ell=0$
\end{itemize}
Dividing both sides by $n \logp$ and letting $\rho, n \to \infty$ we obtain
\begin{IEEEeqnarray}{rCl}
  d_{13} + d_{23} + d_{21} &\leq& \tau M_2 + \ntau M_3. \nonumber
\end{IEEEeqnarray}

\subsection{Sum-DoF Upper Bound for Intermittent 3WC with $M_2 \geq M_1 \geq M_3$ (Part I)}
\label{sec:node-intermittency-sum-dof-case2-part1}

Since this scheme ought to be reliable, we bound the sum rate of $w_{13}$, $w_{23}$ and $w_{12}$ using Fano's inequality:
\begin{IEEEeqnarray}{rCl}
  \IEEEeqnarraymulticol{3}{l}{
      n(R_{13} + R_{23} + R_{12} - \varepsilon_n^{(1)})
  }\nonumber\\
  \quad
    &\leq& \MInf(W_{13} W_{23} W_{12} ; \W_3 \Y_3^n S^n \overbrace{W_{21} \tilde{\Y}_3^n \ZcorrS^n}^{\text{side information}} ) \nonumber\\
    &\leqA& \sum_{\ell=1}^n \MInf(\X_{1,\ell} \X_{2,\ell} ; \hat{\Y}_{3,\ell} \mid S_{\ell}) + n \ologp \nonumber\\
    &\leqB& n \left[ \tau M_1 + \ntau M_3 \right] \logp + n \ologp \nonumber
\end{IEEEeqnarray}
These steps are justified as follows:
\begin{itemize}
  \item[\stepA] Using the proof template presented in Appendix~\ref{sec:node-intermittency-sum-dof-proof-template}, with $\frakWA \eqdef (W_{12},W_{13},W_{23})$, $\frakWB \eqdef (W_{21},W_{31},W_{32})$, $\frakX \eqdef (\X_1,\X_2)$, $\frakY \eqdef \hat{\Y}_3$, $\frakZ \eqdef \hat{\Z}_3$
  \item[\stepB] $(\X_{1,\ell}, \X_{2,\ell}) \leadsto (\Y_{3,\ell}, \tilde{\Y}_{3,\ell})$ is a MIMO channel with $\min\{M_1+M_2,M_3+(M_1-M_3)\}=M_1$ DoFs if $s_\ell=1$, and $\min\{M_1+M_2,M_3+0\}=M_3$ DoFs if $s_\ell=0$
\end{itemize}
Dividing both sides by $n \logp$ and letting $\rho, n \to \infty$ we obtain
\begin{IEEEeqnarray}{rCl}
  d_{13} + d_{23} + d_{12} &\leq& \tau M_1 + \ntau M_3. \nonumber
\end{IEEEeqnarray}

\subsection{Sum-DoF Upper Bound for Intermittent 3WC with $M_2 \geq M_1 \geq M_3$ (Part II)}
\label{sec:node-intermittency-sum-dof-case2-part2}

We bound the sum rate of $w_{21}$, $w_{31}$ and $w_{32}$ using Fano's inequality:
\begin{IEEEeqnarray}{rCl}
  \IEEEeqnarraymulticol{3}{l}{
      n(R_{21} + R_{31} + R_{32} - \varepsilon_n^{(2)})
  }\nonumber\\
  \quad
    &\leq& \MInf(W_{21} W_{31} W_{32} ; \W_1 \Y_1^n S^n \overbrace{W_{23} \tilde{\Y}_1^n \ZcorrS^n}^{\text{side information}} ) \nonumber\\
    &\leqA& \sum_{\ell=1}^n \MInf(\X_{2,\ell} \X_{3,\ell} ; \hat{\Y}_{1,\ell} \mid S_{\ell}) + n \ologp \nonumber\\
    &\eqB& \sum_{\ell=1}^n \Big[ \MInf(\X_{2,\ell} \X_{3,\ell} ; \Y_{1,\ell} \mid S_{\ell}) \nonumber\\&&\quad {}+{} \MInf(\X_{2,\ell} \X_{3,\ell} ; \tilde{\Y}_{1,\ell} \mid S_{\ell} \Y_{1,\ell}) \Big] + n \ologp \nonumber\\
    &\leqC& n \left[ \tau M_1 + \ntau M_3 \right] \logp + n \ologp \nonumber
\end{IEEEeqnarray}
These steps are justified as follows:
\begin{itemize}
  \item[\stepA] Using the proof template presented in Appendix~\ref{sec:node-intermittency-sum-dof-proof-template}, with $\frakWA \eqdef (W_{21},W_{31},W_{32})$, $\frakWB \eqdef (W_{12},W_{13},W_{23})$, $\frakX \eqdef (\X_2,\X_3)$, $\frakY \eqdef \hat{\Y}_1$, $\frakZ \eqdef \hat{\Z}_1$
  \item[\stepB] Chain rule for mutual information
  \item[\stepC] $(\X_{2,\ell}, \X_{3,\ell}) \leadsto \Y_{1,\ell}$ is a MIMO channel with $\min\{M_2+M_3,M_1\}=M_1$ DoFs if $s_\ell=1$, and $0$ DoFs if $s_\ell=0$; $(\X_{2,\ell}, \X_{3,\ell}) \leadsto \tilde{\Y}_{1,\ell}$ is a MIMO channel with $0$ DoFs if $s_\ell=1$ (because then $\tilde{\Y}_{1,\ell}=0$), and $\min\{ 0+M_3, M_1 \} = M_3$ DoFs if $s_\ell=0$ (because then $\Y_{1,\ell}$ is noise, and $\tilde{\Y}_{1,\ell}$ is independent of $X_{2,\ell}$)
\end{itemize}
Dividing both sides by $n \logp$ and letting $\rho, n \to \infty$ we obtain
\begin{IEEEeqnarray}{rCl}
  d_{21} + d_{31} + d_{32} &\leq& \tau M_1 + \ntau M_3. \nonumber
\end{IEEEeqnarray}

\subsection{Sum-DoF Upper Bound for Intermittent 3WC with $M_2 \geq M_3 \geq M_1$ (Part I)}
\label{sec:node-intermittency-sum-dof-case3-part1}

We bound the sum rate of $w_{13}$, $w_{23}$ and $w_{12}$ using Fano's inequality:
\begin{IEEEeqnarray}{rCl}
  \IEEEeqnarraymulticol{3}{l}{
      n(R_{13} + R_{23} + R_{12} - \varepsilon_n^{(1)})
  }\nonumber\\
  \quad
    &\leq& \MInf(W_{13} W_{23} W_{12} ; \W_3 \Y_3^n S^n \overbrace{W_{21} \ZcorrS^n}^{\text{side information}} ) \nonumber\\
    &\leqA& \sum_{\ell=1}^n \MInf(\X_{1,\ell} \X_{2,\ell} ; \Y_{3,\ell} \mid S_{\ell}) + n \ologp \nonumber\\
    &\leqB& n \left[ M_3 \right] \logp + n \ologp \nonumber
\end{IEEEeqnarray}
These steps are justified as follows:
\begin{itemize}
  \item[\stepA] Using the proof template presented in Appendix~\ref{sec:node-intermittency-sum-dof-proof-template}, with $\frakWA \eqdef (W_{12},W_{13},W_{23})$, $\frakWB \eqdef (W_{21},W_{31},W_{32})$, $\frakX \eqdef (\X_1,\X_2)$, $\frakY \eqdef \Y_3$, $\frakZ \eqdef \Z_3$
  \item[\stepB] $(\X_{1,\ell}, \X_{2,\ell}) \leadsto \Y_{3,\ell}$ is a MIMO channel with $\min\{M_1+M_2,M_3\}=M_3$ DoFs if $s_\ell=1$, and $\min\{M_2,M_3\}=M_3$ DoFs if $s_\ell=0$ (because $\X_{1,\ell}$ is independent of $\Y_{3, \ell}$)
\end{itemize}
Dividing both sides by $n \logp$ and letting $\rho, n \to \infty$ we obtain
\begin{IEEEeqnarray}{rCl}
  d_{13} + d_{23} + d_{12} &\leq& M_3. \nonumber
\end{IEEEeqnarray}

\subsection{Sum-DoF Upper Bound for Intermittent 3WC with $M_2 \geq M_3 \geq M_1$ (Part II)}
\label{sec:node-intermittency-sum-dof-case3-part2}

We bound the sum rate of $w_{21}$, $w_{31}$ and $w_{32}$ using Fano's inequality:
\begin{IEEEeqnarray}{rCl}
  \IEEEeqnarraymulticol{3}{l}{
      n(R_{21} + R_{31} + R_{32} - \varepsilon_n^{(2)})
  }\nonumber\\
  \quad
    &\leq& \MInf(W_{21} W_{31} W_{32} ; \W_1 \Y_1^n S^n \overbrace{W_{23} \tilde{\Y}_1^n \breve{\Y}_1^n \ZcorrS^n}^{\text{side information}} ) \nonumber\\
    &\leqA& \sum_{\ell=1}^n \MInf(\X_{2,\ell} \X_{3,\ell} ; \hat{\Y}_{1,\ell} \mid S_{\ell}) + n \ologp \nonumber\\
    &\leqB& n \left[ M_3 \right] \logp + n \ologp \nonumber
\end{IEEEeqnarray}
These steps are justified as follows:
\begin{itemize}
  \item[\stepA] Using the proof template presented in Appendix~\ref{sec:node-intermittency-sum-dof-proof-template}, with $\frakWA \eqdef (W_{21},W_{31},W_{32})$, $\frakWB \eqdef (W_{12},W_{13},W_{23})$, $\frakX \eqdef (\X_2,\X_3)$, $\frakY \eqdef \hat{\Y}_1$, $\frakZ \eqdef \hat{\Z}_1$
  \item[\stepB] $(\X_{2,\ell}, \X_{3,\ell}) \leadsto (\Y_{1,\ell}, \tilde{\Y}_{1,\ell}, \breve{\Y}_{1,\ell})$ is a MIMO channel with $\min\{M_2+M_3,M_1+0+(M_3-M_1)\}=M_3$ DoFs if $s_\ell=1$, and $\min\{M_2+M_3,0+M_1+(M_3-M_1)\}=M_3$ DoFs if $s_\ell=0$
\end{itemize}
Dividing both sides by $n \logp$ and letting $\rho, n \to \infty$ we obtain
\begin{IEEEeqnarray}{rCl}
  d_{21} + d_{31} + d_{32} &\leq& M_3. \nonumber
\end{IEEEeqnarray}

\bibliographystyle{IEEEtran}
\bibliography{IEEEabrv,references}

\begin{IEEEbiographynophoto}{Joachim Neu}
(S'16)
received the B.Sc.\ and M.Sc.\ degree (with high distinction) in electrical and computer engineering from the Technical University of Munich (TUM), Munich, Germany, in 2015 and 2018, respectively.
He also holds a B.A.\ degree in philosophy from Ludwig-Maximilians-University (LMU), Munich, Germany, from 2015.
Currently, he is a graduate student in electrical engineering at Stanford University, Stanford, CA, USA.
His research interests include multi-user information theory, coding for communications and storage, and distributed systems and algorithms.
Mr.\ Neu received the Reed-Hodgson Stanford Graduate Fellowship (2018--2021) and is an alumnus of the German Academic Scholarship Foundation.
\end{IEEEbiographynophoto}

\begin{IEEEbiographynophoto}{Anas Chaaban}
(S'09--M'14--SM'17)
received the Ma{\^i}trise {\`e}s Sciences degree in electronics from Lebanese University, Lebanon, in 2006, the M.Sc.\ degree in communications technology and the Dr.\ Ing.\ (Ph.D.) degree in electrical engineering and information technology from the University of Ulm and the Ruhr-University of Bochum, Germany, in 2009 and 2013, respectively. From 2008 to 2009, he was with the Daimler AG Research Group On Machine Vision, Ulm, Germany. He was a Research Assistant with the Emmy-Noether Research Group on Wireless Networks, University of Ulm, Germany, from 2009 to 2011, which relocated to the Ruhr-University of Bochum in 2011. He was a Postdoctoral Researcher with the Ruhr-University of Bochum from 2013 to 2014, and with King Abdullah University of Science and Technology from 2015 to 2017. He joined the School of Engineering at the University of British Columbia as an Assistant Professor in 2018. His research interests are in the areas of information theory and wireless communications.
\end{IEEEbiographynophoto}

\begin{IEEEbiographynophoto}{Aydin Sezgin}
(S'01--M'05--SM'13)
received the Dipl.\ Ing.\ (M.S.) degree in communications engineering from Technische Fachhochschule Berlin (TFH), Berlin, in 2000, and the Dr.\ Ing.\ (Ph.D.) degree in electrical engineering from TU Berlin, in 2005.
From 2001 to 2006, he was with the Heinrich-Hertz-Institut, Berlin. From 2006 to 2008, he held a postdoctoral position, and was also a lecturer with the Information Systems Laboratory, Department of Electrical Engineering, Stanford University, Stanford, CA, USA. From 2008 to 2009, he held a postdoctoral position with the Department of Electrical Engineering and Computer Science, University of California, Irvine, CA, USA. From 2009 to 2011, he was the Head of the Emmy-Noether-Research Group on Wireless Networks, Ulm University. In 2011, he joined TU Darmstadt, Germany, as a professor. He is currently a professor of information systems and sciences with the Department of Electrical Engineering and Information Technology, Ruhr-Universität Bochum, Germany. He is interested in signal processing, communication, and information theory, with a focus on wireless networks. He has published several book chapters more than 40 journals and 140 conference papers in these topics. He has coauthored a book on multi-way communications. Aydin is a winner of the ITG-Sponsorship Award, in 2006. He was a first recipient of the prestigious Emmy-Noether Grant by the German Research Foundation in communication engineering, in 2009. He has coauthored papers that received the Best Poster Award at the IEEE Communication Theory Workshop, in 2011, the Best Paper Award at ICCSPA, in 2015, and the Best Paper Award at ICC, in 2019. He has served as an Associate Editor for the \textsc{IEEE Transactions on Wireless Communications}, from 2009 to 2014.
\end{IEEEbiographynophoto}

\begin{IEEEbiographynophoto}{Mohamed-Slim Alouini}
(S'94--M'98--SM'03--F'09)
was
born in Tunis, Tunisia. He received the Ph.D. degree in Electrical Engineering
from the California Institute of Technology (Caltech), Pasadena,
CA, USA, in 1998. He served as a faculty member in the University of Minnesota,
Minneapolis, MN, USA, then in the Texas A\&M University at Qatar,
Education City, Doha, Qatar before joining King Abdullah University of
Science and Technology (KAUST), Thuwal, Makkah Province, Saudi
Arabia as a Professor of Electrical Engineering in 2009. His current
research interests include the modeling, design, and
performance analysis of wireless communication systems.
\end{IEEEbiographynophoto}

\end{document}